\documentclass[msom]{informs4_copy}

\usepackage{xtab,comment,colortbl,booktabs,subfig}
\usepackage{multirow}
\usepackage{dashrule}
\usepackage{stackengine}
\usepackage{rotating}\usepackage{lscape}
\usepackage{eurosym}
\usepackage{bbm}
\usepackage{bm}
\usepackage[section]{placeins}
\usepackage{enumerate,xcolor}
\usepackage{soul}
\usepackage{units,ulem}
\usepackage[pass]{geometry}
\usepackage{ textcomp }
\usepackage[hypertexnames=false]{hyperref}
\usepackage{xcolor}
\usepackage{framed}
\hypersetup{
    colorlinks,
    linkcolor={red!50!black},
    citecolor={blue!50!black},
    urlcolor={blue!80!black}
}
\expandafter\def\expandafter\UrlBreaks\expandafter{\UrlBreaks
  \do\a\do\b\do\c\do\d\do\e\do\f\do\g\do\h\do\i\do\j%
  \do\k\do\l\do\m\do\n\do\o\do\p\do\q\do\r\do\s\do\t%
  \do\u\do\v\do\w\do\x\do\y\do\z\do\A\do\B\do\C\do\D%
  \do\E\do\F\do\G\do\H\do\I\do\J\do\K\do\L\do\M\do\N%
  \do\O\do\P\do\Q\do\R\do\S\do\T\do\U\do\V\do\W\do\X%
  \do\Y\do\Z\do\0\do1\do2\do3\do4\do5\do6\do7\do8\do9}
\usepackage{tabularx}
\usepackage{longtable}
\usepackage{arydshln}
\usepackage{soul}
\usepackage{enumitem}
\usepackage{float}
\newcolumntype{L}[1]{>{\raggedright\let\newline\\\arraybackslash\hspace{0pt}}m{#1}}
\newcolumntype{R}[1]{>{\raggedleft\let\newline\\\arraybackslash\hspace{0pt}}m{#1}}
\newcolumntype{C}[1]{>{\centering\let\newline\\\arraybackslash\hspace{0pt}}m{#1}}
\newcolumntype{J}[1]{>{\let\newline\\\arraybackslash\hspace{0pt}}m{#1}}
\newcolumntype{F}[1]{>{\let\newline\\\arraybackslash\hspace{0pt}\linespread{.1}\footnotesize}m{#1}}
\usepackage{natbib}
 \bibpunct[, ]{(}{)}{,}{a}{}{,}%
 \def\bibfont{\small}%
 \def\BIBand{and}%

 \newtheorem{result}{\usefont{T1}{cmr}{b}{}\selectfont Result}
{\begin{list}{}%
         {\setlength{\leftmargin}{#1}}%
         \item[]%
}
{\end{list}}

\definecolor{Gray}{gray}{0.8}
\definecolor{Gray2}{gray}{0.5}
\xdefinecolor{olive}{cmyk}{0.64,0,0.95,0.4}

\newcounter{marginparcounter}

\newcommand{\tabnote}[1]{ \begin{minipage}{\textwidth} \scriptsize{\vspace{5pt} #1} \end{minipage} } %

\renewcommand{\footnotesize}{\small}   

\TheoremsNumberedThrough     

\EquationsNumberedThrough    

\makeatletter
\newcommand*{\da@rightarrow}{\mathchar"0\hexnumber@\symAMSa 4B }
\newcommand*{\da@leftarrow}{\mathchar"0\hexnumber@\symAMSa 4C }
\newcommand*{\xdashrightarrow}[2][]{%
  \mathrel{%
    \mathpalette{\da@xarrow{#1}{#2}{}\da@rightarrow{\,}{}}{}%
  }%
}
\newcommand{\xdashleftarrow}[2][]{%
  \mathrel{%
    \mathpalette{\da@xarrow{#1}{#2}\da@leftarrow{}{}{\,}}{}%
  }%
}
\newcommand*{\da@xarrow}[7]{%
  \sbox0{$\ifx#7\scriptstyle\scriptscriptstyle\else\scriptstyle\fi#5#1#6\m@th$}%
  \sbox2{$\ifx#7\scriptstyle\scriptscriptstyle\else\scriptstyle\fi#5#2#6\m@th$}%
  \sbox4{$#7\dabar@\m@th$}%
  \dimen@=\wd0 %
  \ifdim\wd2 >\dimen@
    \dimen@=\wd2 %
  \fi
  \count@=2 %
  \def\da@bars{\dabar@\dabar@}%
  \@whiledim\count@\wd4<\dimen@\do{%
    \advance\count@\@ne
    \expandafter\def\expandafter\da@bars\expandafter{%
      \da@bars
      \dabar@
    }%
  }%
  \mathrel{#3}%
  \mathrel{%
    \mathop{\da@bars}\limits
    \ifx\\#1\\%
    \else
      _{\copy0}%
    \fi
    \ifx\\#2\\%
    \else
      ^{\copy2}%
    \fi
  }%
  \mathrel{#4}%
}
\makeatother

\newcommand{\sym}[1]{^{#1}} 
\newcolumntype{d}[1]{D{.}{.}{#1}} 
\newcommand{\dor}[1]{\multicolumn{1}{c}{#1}} 

\usepackage{endnotes}
\let\endnote=\endnote
\begin{document}


\RUNAUTHOR{Kagan, Hyndman and Qi}

\RUNTITLE{Startup Contracting and Entrepreneur-Investor Bargaining}

\TITLE{Startup Contracting and Entrepreneur-Investor Bargaining} 

\ARTICLEAUTHORS{%
\AUTHOR{Evgeny Kagan}
\AFF{Carey Business School, Johns Hopkins University, \EMAIL{ekagan@jhu.edu}, \URL{https://carey.jhu.edu/faculty/faculty-directory/evgeny-kagan-phd}}

\AUTHOR{Kyle Hyndman}
\AFF{Naveen Jindal School of Management, University of Texas at Dallas, \EMAIL{KyleB.Hyndman@utdallas.edu}, \URL{https://www.kylehyndman.com/}}

\AUTHOR{Anyan Qi}
\AFF{Naveen Jindal School of Management, University of Texas at Dallas, \EMAIL{axq140430@utdallas.edu}, \URL{https://qi-anyan.github.io/}}
} 

\ABSTRACT{\textbf{Problem definition:} To grow their businesses, entrepreneurs often rely on equity funding. This paper focuses on two elements of entrepreneur-investor equity negotiations: the number of potential investors and the contractual complexity surrounding investor protection. \textbf{Methodology/Results:} Our approach involves a theoretical model and a series of laboratory experiments that analyze the effects of different bargaining conditions and contractual terms on the equity (ownership) split between entrepreneurs and their investors. We show that the conventional wisdom that entrepreneurs should seek to negotiate with as many investors as possible, while consistent with the theoretical model, is not true in the data. Indeed, negotiating with multiple investors reduces the entrepreneur's profits under most conditions. We also show that investor downside protections may disadvantage early-stage startups, but can be beneficial to later-stage startups. A refinement of belief modeling in multi-party bargaining, as well as a stylized risk allocation framework, reconcile these results with theory predictions. \textbf{Managerial Implications:} Our findings provide a decision framework for entrepreneurs to optimize their approach to investors and negotiate favorable contractual terms.}



\maketitle
\normalem
\vspace{-0.5cm}
  \section{Introduction \label{sec:intro}}

Over 54,000 startups in the United States currently receive venture capital financing, reflecting the significance of entrepreneurial finance in the broader economy \citep{NVCA2024Yearbook}. Deal‑matching platforms, startup accelerators, and syndicated angel networks have sharply lowered search frictions, making it common for founders to pitch several investors in parallel. Separately, term sheets that assign equity positions to startup founders and their investors have become increasingly sophisticated. Many new investment rounds now routinely include investor protection clauses, ranging from liquidation preferences to participating preferred shares,  anti‑dilution, cumulative dividends, and redemption rights, all of which are designed to shift financial risks from investors to entrepreneurs \citep{Carta2024Q1Report}. 
\vspace{0.12cm}

\noindent These trends raise two strategic questions for founders:
\vspace{0.12cm}

\noindent \textbf{How many investors to approach?} The number of investors participating in a venture capital round varies considerably, averaging between three and four \citep{PitchBook2024VentureReport}. Dealing with a single large investor simplifies coordination but deprives the entrepreneur of outside options; inviting two or more smaller investors increases complexity but may help ensure that at least some funding is secured.
    \vspace{0.12cm}

\noindent \textbf{What equity class to offer?} Equity (or ``stock'') class is a key contractual mechanism for allocating financial risk. \textit{Common} equity leaves all parties equally exposed to risk, whereas \textit{Preferred} equity grants investors priority or guaranteed payouts, potentially allowing founders to bargain for a larger residual stake in exchange for risk transfer.
    \vspace{0.12cm}


 \subsection{Field Data}\label{sec:crunchbase}
To ground our research in real-world data, we first examined startup funding data available on \url{www.crunchbase.com}, a technology startup database widely used in entrepreneurship research \citep{dalle2017,contigiani2023}.  Figure \ref{fig:crunchbase} shows the data for the sample of US-based software startups over the 2010–2024 period, covering 2747 distinct investment rounds.\endnote{We focus on the software sector, because it is the sector with the largest number of firms on Crunchbase; however, similar insights emerge if we use biotech (second-largest sector) data instead.} We restrict the sample to firms with disclosed funding amounts and valuations. Our analysis focuses on pre-seed, seed, and Series A/B rounds, i.e., primarily early to mid-stage startups. The figure plots the number of investors against the combined investor share across funding amount deciles (to control for investment amount). Blue annotations indicate the within‐decile Pearson correlation coefficients between the number of investors and combined investor share, along with their significance levels. 

\begin{figure}[b!]
\FIGURE
{\includegraphics[width=\textwidth]{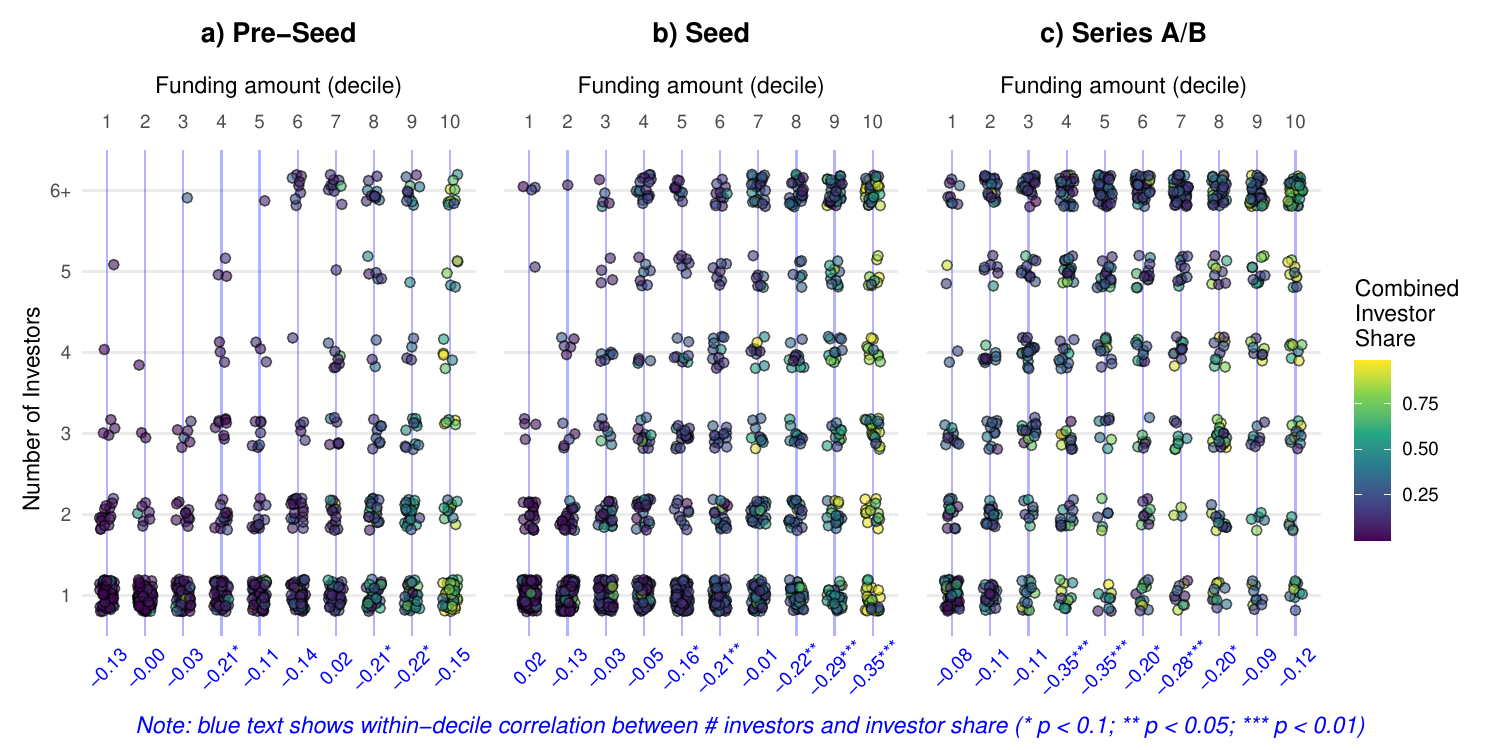}}
{Crunchbase Data: Funding Events (Software Sector, 2010-2024) \label{fig:crunchbase}}
	{}
\end{figure}

What do we learn from the Crunchbase data? First, early-stage startups tend to attract between one and three investors. More mature startups at the Series A/B stage tend to have a larger investor base, particularly those with larger investment amounts. Indeed, the percentage of funding rounds with a single investor goes down from 61.0 to 47.8 and to 18.9 in the Pre-Seed, Seed and Series A/B stages, respectively. Nevertheless, there is substantial heterogeneity in the number of investors, even conditional on the startup stage and investment amount. Second, the figure suggests a negative relationship between the number of investors and the combined investor share. A closer look suggests that this relationship is more pronounced -- and in some funding deciles statistically significant -- for later-stage startups and larger funding amounts. In particular, the correlation is significantly different from 0 (at least at $p<0.05$) for four of the deciles in the Seed round and for three of the deciles in the Series A/B round.  

While the Crunchbase data offer some suggestive insights, the observations are limited by several empirical issues. First, the data may suffer from survivorship bias, as unsuccessful startups often lack publicly disclosed valuation and investment data. Particularly for the later stages and larger funding amounts, the sample is thus skewed toward more successful firms. Second, the observed negative relationship between the number of investors and combined investor share could reflect reverse causality. This is because higher-quality startups may naturally attract more investors and those investors may be willing to accept less equity for a chance to be included in a round. Conversely, there may be investor selection effects as skilled entrepreneurs may deliberately seek high-quality investors despite their demands for larger equity stakes, as these investors provide significant value beyond capital. This further complicates causal interpretations. Finally, the absence of detailed contractual data limits our ability to isolate the specific contractual mechanisms. Indeed, the negative relationship between the number of investors and the combined investor share may be explained by certain (e.g., institutional) investors preferring a smaller share in exchange for greater downside protections \citep{kaplan2016}. In contrast, rounds with just a handful of investors might opt for  Common Stock contracts. 

The Crunchbase data illustrates common identification challenges in the observational data on venture financing and highlights several open questions in our understanding of entrepreneur-investor interactions. \textit{Should entrepreneurs negotiate with multiple investors simultaneously? Do investor protection clauses (such as Preferred equity) allow entrepreneurs to retain a larger share of the startup? Do the answers to these questions vary depending on the startup's development stage?} Existing studies in entrepreneurial finance, though extensive \citep[see][for reviews]{darin2013,kerr2015,lerner2020}, do not provide answers to these questions. To address these gaps, we develop an analytical model of entrepreneur-investor bargaining and contracting and then test key predictions of this model experimentally. 


\subsection{Overview of Studies and Results}\label{sec:crunchbase}
We begin by introducing a model of founder-investor bargaining and contracting (Section 3). Our modeling approach builds on a growing stream of work in Operations Management that applies Nash‑in‑Nash theory \citep{horn1988bilateral,davidson1988multiunit} to multilateral procurement negotiations \citep{feng2012strategic,chu2019strategic}. We extend the Nash-in-Nash theory to early‑stage equity contracting in entrepreneurship by incorporating uncertain venture value, multiple-investor case and contracts that are prevalent in the entrepreneurial setting. The key model predictions are as follows. First, the entrepreneur should benefit from negotiating with a larger number of investors. The logic is that, keeping the maximal investment amount constant, the entrepreneur is better off with two smaller investors because the entrepreneur can leverage the ability to walk away from one investor and still obtain (some) funding from the other investor. In contrast, with a single large investor, the entrepreneur does not have this leverage.\endnote{The Nash-in-Nash logic that ``the more investors, the better'' can be further extended to scenarios with more than two investors, but this is outside of the scope of this research.} Second, the model predicts that entrepreneurs can use the reduction in investor risk enabled by Preferred Stock contracts as a source of leverage for retaining a larger share of the startup. We test these predictions experimentally by varying (i) the number of investors and (ii) contract type.

In Section 4 we test the logic that ``the more investors, the better'' and examine how the number of investors affects the share of the startup retained by the entrepreneur. Our first experimental result is that, contrary to the Nash-in-Nash logic, entrepreneurs retain substantially more equity when negotiating with a single investor compared to the two-investor case: approximately eight percentage points more equity, conditional on all agreements being successful. However, the two-investor scenario enables split investments (where only one investor participates in the final agreement). As a result, the two-investor scenario has a significantly higher chance of reaching at least a partial agreement. This creates a trade-off: conditional on agreement, entrepreneurs are able to secure a greater share of ownership in the single-investor case, but they are able to come to an agreement more frequently in the two-investor case. As a consequence of this trade-off, the entrepreneur's \textit{unconditional} expected profits are statistically indistinguishable between single and two-investor scenarios when considering all negotiations, both successful and failed.

We continue to examine the relationship between the number of investors and negotiation outcome by conducting treatments in which entrepreneurs are endowed with a stronger outside option. This is to represent negotiation scenarios in which a startup has already found a product-market fit and can survive even without investment. In this more balanced bargaining environment, we continue to see that multiple investors do not benefit the entrepreneur. Indeed, with a larger outside option, entrepreneurs are strictly better off with a single investor, whether we focus on the negotiated shares, or on the (unconditional) expected profits. This is driven largely by an increase in the gap in obtained equity between the single and the two-investor cases (relative to the scenario where the entrepreneur had no outside options). 

In Section 5, we turn to the question of whether offering investors downside protection (in the form of Preferred Stock) helps entrepreneurs retain larger ownership. Under Preferred Stock contracts, investors are repaid in full before the entrepreneur receives any money; therefore, Preferred Stock contracts shift a large portion of risk from the investors to the entrepreneur. Our data suggest that Preferred Stock contracts disadvantage the entrepreneur when the entrepreneur has no other outside options. In this scenario, bargaining outcomes do not adjust sufficiently to compensate the entrepreneur for bearing the risk, which goes against the bargaining model predictions. However, in the more balanced bargaining environment where the startup has some intrinsic value, Preferred Stock contracts help entrepreneurs retain a larger share, consistent with the model predictions. 

In Section 6, we revisit the Nash-in-Nash theory and propose two refinements that help reconcile the model with our findings. First, we show that an appropriate adjustment of off-equilibrium belief modeling aligns model predictions more closely with observed data. Second, to explain the effects of different contract types, we introduce a framework that accounts for the effects of risk exposure on bargaining outcomes. Specifically, we posit that a non-zero outside option changes the bargaining dynamics, placing the entrepreneur on more equal footing with the investor(s) and providing a compelling argument for retaining a larger share with Preferred Stock, compared to scenarios where the entrepreneur has no outside options and no capital at risk.

\subsection{Contributions}
We contribute to the literature in two ways. First, we contribute to the entrepreneurship literature by studying conditions under which multiple investors may benefit or hinder entrepreneurs in retaining a larger share of their startups. In doing so, we highlight the trade-off between agreement rates and equity allocation which has not been previously emphasized. Additionally, we show a negative effect of excessive investor downside protection on entrepreneurial profits, particularly for early-stage startups, providing a potential causal mechanism for similar findings in the empirical finance literature \citep[see, for example,][]{ewens2022}. Second, our experimental results help advance bargaining theory. We highlight the importance of carefully modeling off-equilibrium beliefs and show that doing so improves the predictive accuracy of the Nash-in-Nash model. Additionally, we introduce a simple framework that links each party's money-at-risk with the expected negotiation outcomes and suggest how the Nash-in-Nash model can be revised to accommodate such behaviors. Our experiments thus provide new behavioral micro-foundations for future models of multi-party bargaining and contracting under uncertainty.  

\section{Literature \label{sec:lit}}
Entrepreneur-investor contracting has not received much attention in the operations management literature  \citep{krishnan2001,kavadias2020}; indeed, a recent review has identified both contract design and entrepreneurship as two areas that remain understudied from an operational perspective \citep[][pp. 86-87]{kavadias2020b}. To bridge this gap, our work draws on the extensive research in finance and entrepreneurship, as well as on the bargaining literature in economics and operations management.

\subsection{Finance and Entrepreneurship \label{sec:lit:fin}}
The relevant finance and entrepreneurship literature can be organized into two distinct streams: the literature that documents the prevalence of different types of investors and contracting models as well as the literature that studies strategic interactions between startups and investors.

\subsubsection*{Types of Investors and Contracts} 
Historically, the investment landscape in entrepreneurship included mainly two types of investors: smaller-scale angel investors and angel funds, and larger, often industry-focused venture capital (VC) firms \citep{kaplan2016}. More recently, the investor landscape has seen a broadening of investor types with increasingly disparate characteristics. This is partly due to the rise of startup accelerators, incubators and similar entrepreneurial programs that have significantly reduced the barriers to matching startups with prospective investors \citep{cohen2019}. For example, \cite{amore2023} document the emergence of micro VCs, which are -- similar to classic VCs -- professionally managed funds, but tend to make smaller investments (typically under \$100K) and are less likely to syndicate \citep{cbinsights2015}. 
 
The literature further documents a multitude of different contracting models, frequently including protective clauses for investors \citep{darin2013}. Particularly common are convertible preferred equity contracts, an ownership model that provides a (fixed) multiple return on investment in low-exit scenarios, yet converts to common equity in high-exit scenarios. While such contracts create stronger incentives for entrepreneurs to exert effort \citep{schmidt2003,hellmann2006,debettignies2008}, their inherent complexity can make it difficult to understand valuations and cash flow rights \citep{gornall2020}. Indeed, a recent trend is towards more founder-friendly contracts with simpler term sheets \citep{ycombinator2023}. 

Taken together, both the trend towards broader investor types and the lack of convergence on a single contract form present an opportunity for a more strategic approach to fundraising, as startups navigate a wider range of funding options. This motivates our investigation into the effects of multiple investor types and contract forms on bargaining outcomes.

\subsubsection*{Strategic Interactions and Startup Ownership}
Firm ownership and financing is one of the classic microeconomic questions, popularized as the ``theory of the firm'' \citep{grossman1986,hart1990}. Studies in this stream assume take-it-or-leave-it behavior, bypassing any bargaining or negotiation dynamics. Studies that take a more cooperative approach to bargaining are \cite{hellmann,hossain2019} and \cite{kagan}. Different from us, they study ownership allocation \textit{within} the entrepreneurial team and not \textit{between} the entrepreneur and investors. \cite{akerlof2019} and \cite{halac2020} analytically examine sequential investor choice. Both papers assume leader-takes-all behaviors and do not consider bargaining dynamics. \cite{ewens2022} use observational data and find that investors generally receive greater ownership stakes relative to what would be ``optimal'' from the welfare maximization perspective. A key observation in \cite{ewens2022} is that the bulk of the excess profits received by investors is due to liquidation preferences, suggesting that simpler equity split contracts are better for the entrepreneur. Our study provides a plausible causal mechanism for this observation: we show that liquidation preferences are not fully accounted for in the negotiations, leading to inflated investor ownership.

\subsection{Bargaining in Economics and Operations Management \label{sec:lit:barg}}
Bargaining problems have attracted significant interest in the academic literature \citep{Roth1995}. Most of these studies examine the problem of splitting a pie of a given size, and do not consider the relevant features of the entrepreneurial setting such as multiple investors, size of investment, uncertainty, or equity contract types.

\subsubsection*{Cooperative Bargaining}
The early experimental economics literature focused mainly on testing Nash solution predictions for bilateral negotiations with complete information  \citep{NO1974,RR1982,MRS1988}.
Two features of the entrepreneurial setting, that the entrepreneur may bargain with multiple investors, and that the size of the pie (value of startup equity) is both endogenous and uncertain, have attracted relatively little attention. The studies of multilateral bargaining in economics \citep[see, e.g.,][]{frechette2005b,frechette2005a} focus on legislative bargaining. These papers have highly structured bargaining formats in order to test features of interest to these models. There is a small literature in which players negotiate over a pie, the size of which is endogenous \citep{Gantner-et-al,BK2013,Rod-Lara,Baranski2019}. The main insight is that players often take self-serving bargaining positions. The closest study in this stream is \cite{EHR2020}. Like us, they study bargaining over risky pies where risk exposure is asymmetric, with one player receiving a fixed amount, and the second player receiving the (uncertain) remainder amount. Their main finding is that the player bearing the risk is able to extract a premium. No studies that we are aware of examine the types of equity division contracts that are prevalent in entrepreneurial practice (Common vs. Preferred Stock) or compare the outcomes of single vs. multiple investor bargaining.

\subsubsection*{Applications in Operations Management}
While there has been extensive research on bargaining in operations management -- much of it has focused on the supply chain context \citep{DavisLeider2018,DavisHyndman2018,DHQH2020}. The closest to us is \cite{DHQH2020} who study an original equipment manufacturer (OEM) purchasing from two suppliers under different bargaining sequences and power distributions. They focus on procurement decisions in assembly systems with private information, whereas our work examines entrepreneurial funding scenarios with different contracts (risk allocation mechanisms). Somewhat surprisingly, some of the results that hold in the (more abstract and sterile) economic setting do not carry over when context is added. For example, unlike \cite{EHR2020} who find that residual claimants are able to negotiate a high premium compensating them for risk exposure, \cite{DavisHyndman2018} find that the supply chain partner carrying inventory risk is not fully compensated for that risk. Thus, institutional context matters, even for problems that are mathematically equivalent. 

A rich behavioral literature studies how risks (of excess inventory, poor product quality or unserved demand) are allocated between supply chain partners, in serial inventory ordering settings \citep[e.g.,][]{sterman1989,croson2014,moritz2022} and in supplier-retailer contracting settings \citep[e.g.,][]{OZC2011,ozer2014,beer2018}. The closest papers in this literature are \cite{Kalkanci2011,Kalkanci2014} and \cite{Zhang2016}. Also related are the papers on supply chain coordination \citep{lim2007,ho2008,katok2013}, which find that rejection rates can dominate the effects of specific contract types. Different from this literature, we look at bargaining in the entrepreneurial context, allow the parties to endogenously determine contractual terms (rather than make take-it-or-leave-it offers), and do not consider how contracts affect subsequent behaviors; however, similar to this literature, we are also interested in the role of contract complexity in driving efficiency (surplus size) and the allocation of risk and profits between the contracting parties. 

Finally, the scenarios examined in our study include negotiations with horizontal competition. This is an understudied problem in the literature, with the closest being \citet{Lovejoy2010} and \citet{leider2016} who study bargaining in supply chains. In their setting, multiple firms with different cost structures compete for being included in the multi-tier supply chain contract.  Different from these studies, which assume single sourcing/contracting within a tier, we examine a setting where multiple parties from the same tier (multiple investors) can be included in the contract and examine contracts that differ in their allocation of risk.

\section{Theory and Experiment Overview \label{sec:overview}}
We examine negotiation scenarios defined by two factors: the number of investors (one or two) and contract type (Common or Preferred Stock). The single investor (SI) scenario can be reflective of a large venture capital firm, or can represent an angel or investor syndicate where a lead investor or representative handles negotiations on behalf of a group of smaller investors. In contrast, in the two investor (TI) case, each investor acts independently.  Furthermore, we examine two distinct environments: one in which the entrepreneur has no outside options if the negotiations fail (``PoorEnt''), and a second one in which the entrepreneur can still launch the business, though at a lower scale if the negotiations fail (``RichEnt''). The assumption that the entrepreneur is ``penniless'' and has no outside options is standard in the contracting literature \citep{bolton2004,aghion2011} and in the more recent game-theoretic work on capital assembly \citep{akerlof2019,halac2020}. However, the presence of outside options is a key determinant of outcomes in the bargaining literature (see Section \ref{sec:lit:barg}); we therefore examine scenarios in which one of the parties (in our case, the entrepreneur) has a smaller or a larger outside option. 

\subsection{Theory Predictions \label{sec:theory}}
To develop analytical benchmarks, we apply the classic Nash-in-Nash model \citep{horn1988bilateral,davidson1988multiunit} to our setting, which includes uncertainty, the multiple investor case and the different contracts (Common/Preferred Stock). In the single-investor case, the Nash-in-Nash model reduces to the Nash Bargaining Solution \citep{Nash1950}. In the two-investor case, each bilateral disagreement point depends on the expected equilibrium payoff from all other negotiations, making the bargaining units interdependent. By using the cooperative concept of the Nash bargaining solution for each entrepreneur–investor pair, we capture the mutually agreeable surplus split in bilateral bargaining; embedding these in a noncooperative equilibrium then captures how the parties strategically view and react to each other's outside options. In addition to being broadly reflective of the entrepreneur-investor dynamics, the Nash-in-Nash framework has the advantage of nesting other solution concepts as special cases through appropriate parameterization of the bargaining parameters (see Remark 1).

We next provide the analysis with Common Stock contract in Section~\ref{sec:theory:common} and Preferred Stock contract in Section~\ref{sec:poor:ent:contracts}. The complete listing of notation used in our theoretical exposition is in Table \ref{tab:notation}.


\subsubsection{Common Stock Contract}\label{sec:theory:common} We focus on a setting with a single entrepreneur and a set of potential investors, $\mathcal{I}$. The entrepreneur seeks investment of up to $e$ units of capital from one or more investors from the set $\mathcal{I}$. Each investor, $i$, has an initial endowment of $\bar{I}^t \le e$ and can make any investment amount from the set $I^t=[0,\bar{I}^t]$, where $t \in \{\textrm{SI, TI}\}$ denotes the bargaining institution. We assume that $0 \in I^t$, which simply means that an investor is not obligated to invest. The entrepreneur engages in simultaneous bilateral negotiations with each investor $i$ over the share, $s^t_i$, that investor $i$ will receive in exchange for an investment of $I^t_i\in I^t$. Denote by $\mathcal{I}' \subseteq \mathcal{I}$ the set of investors with which the entrepreneur reaches an agreement. The investments and shares must satisfy the constraints that $\sum_{i\in \mathcal{I}'}I^t_i \le e$ and $0 \le \sum_{i\in \mathcal{I}'}s^t_i \le 1$. That is, the total investment must be no greater than $e$, and the entrepreneur cannot give away more than the entire business.

If an investment has been agreed to, i.e., $\sum_{i\in \mathcal{I}'}I^t_i > 0$, the startup may succeed or fail. Let $\alpha$ be the random variable representing startup success. With probability $p$ the startup succeeds and $\alpha=\alpha_H > 1$; in this case, the value of the startup is $V = \alpha_H\sum_{i\in \mathcal{I}'}I^t_i$. With probability $1-p$, the startup fails and $\alpha=a_L = 1$; in this case, $V = \alpha_L\sum_{i\in \mathcal{I}'}I^t_i$. Denote by ${\mu_\alpha} = \mathbb{E}[\alpha]=\alpha_H p + \alpha_L(1-p)$ the expected value of the multiplier $\alpha$. The realized payoff to investor~$j$, who reached an agreement with the entrepreneur is:
\begin{equation}
s^t_j \alpha \sum_{i\in \mathcal{I}'}I^t_i + \bar{I}^t - I_j^t. \label{eq:invP}
\end{equation}
The realized payoff to the entrepreneur is then:
\begin{equation}
\left(1- \sum_{i\in\mathcal{I}'} s^t_i\right) \alpha \sum_{i\in \mathcal{I}'}I^t_i. \label{eq:entP}
\end{equation}
To solve for equilibria, we require expected profits. To this end, given agreed upon shares, $\bm{s} $, and investments, $\bm{I}$, we denote the expected profits for the investor(s) by $\pi_i(\bm{I},\bm{s})$ and for the entrepreneur by $\pi_e(\bm{I},\bm{s})$. These are calculated by taking expectations of realized profit in \eqref{eq:invP} and \eqref{eq:entP} over $\alpha$.

Finally, any investor who does not reach an agreement with the entrepreneur holds onto their initial capital, $\bar{I}^t$, and if the entrepreneur cannot secure investment from \textit{any} investor, then the entrepreneur earns $d_\text{e} \ge 0$. That is, the entrepreneur has an outside option of $d_\text{e}$. In our experiments, we consider both the Poor Entrepreneur (PoorEnt) condition where $d_\text{e}=0$ and the Rich Entrepreneur (RichEnt) condition where $d_\text{e}>0$. In the RichEnt case, we consider values of $d_\text{e}$ that allow both the entrepreneur and the investor(s) to achieve gains from reaching an agreement.

The characterization of equilibria for these bargaining problems depends on $(p,\alpha_H,\alpha_L)$ and, following the Nash bargaining framework, on the relative bargaining power of the players, indexed by $\theta_i \in [0,1]$ to denote the relative bargaining power of investor $i$ when bargaining with the entrepreneur. Equal bargaining power is given by $\theta_i = \nicefrac{1}{2}$. We separately characterize equilibria for both bargaining institutions. In the single investor case, the equilibrium is always unique. In the two investor cases, multiple equilibria are possible, and we will provide further discussion below. To the extent possible, the theoretical results below will focus on parametrizations that give rise to a unique equilibrium. Further, for ease of exposition, we will assume equal bargaining powers. A more complete characterization of the equilibria under general ${\mu_\alpha}$ and general bargaining powers is relegated to Appx. \ref{sec:appx}. Lastly, we assume risk neutrality (see Appx. \ref{sec:risk:aversion} for the risk-averse case).

\subsubsection*{Bargaining with a Single Investor (SI)}
In this case, there is a single investor, whom we will refer to as Investor 0. We assume that $\bar{I} = e$ and that the set of feasible investments, $I = [0,e]$ (when there is no scope for confusion, we will drop the superscript indicating the bargaining institution), although under risk neutrality, efficiency will always lead to full investment. Denote by $I_0$ the amount invested by Investor 0. The value of the business becomes $V=\alpha I_0$. 

Investor~$0$ and the entrepreneur bargain over the size of the investment made by the investor, $I_0$, and the share of the startup, $s_0$, that the investor will receive in exchange for making the investment. The entrepreneur's share is given by $s_\text{e} = 1-s_0$. Let $d_\text{e}$ and $d_0$ denote the disagreement point of the entrepreneur and Investor~$0$, respectively; i.e., their respective profits if the negotiation breaks down. As stated above, an investor who does not invest keeps their endowment; hence,  $d_0=e$. The entrepreneur keeps the outside option $d_\text{e} \ge 0$ if the investor does not invest. Given our assumptions, the expected profit of the entrepreneur is $\pi_\text{e}(I_0, s_0)={\mu_\alpha} I_0  (1-s_0)$, while the expected profit of the investor is $\pi_0(I_0, s_0)={\mu_\alpha} I_0  s_0+e -I_0$. If a deal is settled, the investment $I_0$ and the share $s_0$ maximize the following Nash product:
\begin{align}
\begin{split}
\max_{I_0\in[0, e],~ s_0\in[0, 1]}~&\left[\pi_0(I_0, s_0)-d_0\right]\left[\pi_\text{e}(I_0, s_0)-d_\text{e}\right]\\
&\pi_0(I_0, s_0)\geq d_0,~\pi_\text{e}(I_0, s_0)\geq d_\text{e}.
\end{split}
\end{align}
Solving (3), we obtain the following bargaining outcome.
\begin{proposition}[Single investor]\label{prop:single}
The investor invests $I_0^{SI}= e$. The shares are as follows:
\begin{align*}
&s_0^{SI}=\frac{{\mu_\alpha} +1}{2 {\mu_\alpha}}-\frac{d_\text{e}}{2 e {\mu_\alpha}}, \quad
s_\text{e}^{SI}=1-s_0^{SI}=\frac{{\mu_\alpha} -1}{2 {\mu_\alpha}}+\frac{d_\text{e}}{2 e {\mu_\alpha}}.
\end{align*}

\end{proposition}

Proposition \ref{prop:single} reproduces the standard result from the Nash Bargaining literature: converted to expected profits, the shares equalize the negotiators' gains from negotiating minus their disagreement payoffs.

\subsubsection*{Bargaining with Two Investors (TI)}
In this setting, the set of investors is $\mathcal{I} = \{\textrm{Investor 1, Investor 2}\}$. Each investor has an endowment of $\bar{I} = \nicefrac{e}{2}$ and the set of feasible investments for each investor is $I = [0,\nicefrac{e}{2}]$. That is, each investor has an endowment of half the total desired investment and can invest up to their endowment. The value of the startup after the bargaining is $V=\alpha(I_1 + I_2)$.\endnote{We focus on the two investor case because it captures many of the first-order bargaining dynamics relative to the single investor case, for example, the improved bargaining position of the entrepreneur with multiple investors. However, much of the theoretical analysis can be readily extended to an arbitrary number of investors.}

With multiple investors, each one (Investors $i=1,2$) engages in separate bilateral bargaining with the entrepreneur about the investment amounts, $I_i$, and the shares, $s_i$, received in exchange for their investment. We adopt the Nash-in-Nash solution approach to determine the negotiation outcome; i.e., the negotiation outcomes are derived as a Nash equilibrium of two simultaneous Nash bargaining problems. We denote the outcome of each bargaining unit~$i$ (the bargaining between the entrepreneur and Investor~$i$) by $(I_i, s_i)$ and the collective outcomes by $\bm{I}=(I_1, I_2)$ and $\bm{s}=(s_1, s_2)$. Then, the expected profit of the entrepreneur is $\pi_\text{e}(\bm{I}, \bm{s})= {\mu_\alpha} (I_1+I_2)  (1-s_1-s_2)$  and the expected profit of Investor~$i$ is $\pi_i(\bm{I}, \bm{s})= {\mu_\alpha} (I_1+I_2)  s_i+ \nicefrac{e}{2} -I_i$.

With two investors, if the entrepreneur fails to agree with \textit{both} investors, it is still the case that the startup fails to launch, and the entrepreneur earns the outside option $d_\text{e} \ge 0$. However, the entrepreneur's disagreement point versus any one of the two investors is more subtle, because the entrepreneur may still earn a profit from agreement with the other investor. Hence the entrepreneur will have a disagreement point versus each investor $i$, denoted by $d_\text{e}^{-i}$, which will depend on the shared \textit{beliefs} about what would happen if the entrepreneur and that investor disagreed. We make the assumption common in the literature that the agreement with Investor $j$ is the same, \textit{whether or not the entrepreneur agreed with} Investor $i$ \citep{Yurukoglu}.\endnote{This assumption is plausible given that the process and the outcome of the negotiation between the entrepreneur and investor $i$ are not observable to investor $j$. It is, however, an assumption, and the validity of the assumption can be evaluated with our experimental data.} Then $d_\text{e}^{-1}=\pi_\text{e}(0, I_2, 0, s_2)$ is the profit of the entrepreneur when Investor~$2$ is the only investor. Similarly $d_\text{e}^{-2} = \pi_\text{e}(I_1, 0, s_1, 0)$.  Further, the disagreement point of Investor~$i$ is $d_i=\nicefrac{e}{2}$ since each investor has $\nicefrac{e}{2}$ units of capital as the endowment. Then, the investments $\bm{I}$ and the shares $\bm{s}$ maximize the Nash products simultaneously for each $i = 1,2$:
\begin{align}
\begin{split}
\max_{I_i\in[0, \nicefrac{e}{2}],~ s_i\in[0, 1]}~&\left[\pi_i(\bm{I}, \bm{s})-d_i\right]\left[\pi_\text{e}(\bm{I}, \bm{s})-d_\text{e}^{-i}\right]\\
&\pi_i(\bm{I}, \bm{s})\geq d_i,~\pi_\text{e}(\bm{I}, \bm{s})\geq d_\text{e}^{-i}, ~i \in \{1,2\}.
\end{split}
\end{align}

\noindent Solving (4), we obtain the following proposition:
\begin{proposition}[Two investors] \label{prop:sim}
Both investors invest the endowed capital in equilibrium; i.e., $I_i^{TI}=\nicefrac{e}{2}$ for $i\in\{1,2\}$. The equilibrium shares are as follows:
\begin{align*}
&s_i^{TI}=\frac{{\mu_\alpha}+1}{5 {\mu_\alpha}}-\frac{d_\text{e}}{5 e  {\mu_\alpha}},~i=1,2, \quad  s_\text{e}^{TI}=1-s_1^{TI}-s_2^{TI}=\frac{3 {\mu_\alpha}-2}{5 {\mu_\alpha}}+\frac{2 d_\text{e}}{5 e  {\mu_\alpha}}.
\end{align*}
\end{proposition}

A straightforward comparison of $s_e^t$, $t \in \{\textrm{SI, TI}\}$, yields the following corollary.
\begin{corollary}[Comparison of entrepreneur's share]\label{corollary: ent_share} The entrepreneur's equilibrium shares satisfy that 
$s_e^{TI} > s_e^{SI}$.
\end{corollary}

\begin{remark}[Bargaining Power] 
    Although we presented the theory for equal bargaining power, as we show in the appendix, as bargaining power varies, a wide range of outcomes are possible. If $\theta_i = 0$ for all $i$, then in SI and TI, the entrepreneur has full bargaining power and will extract all the surplus. Such a scenario could be viewed as one in which the entrepreneur can make ultimatum offers. The flip side is when $\theta_i = 1$ for all $i$. In this case, the investors have all the bargaining power. In SI, therefore, the investor will extract all the gains. Interestingly, because of the more competitive nature of TI, even though $\theta_i = 1$, the entrepreneur is still able to extract some surplus. For details, please see EC.1 for theoretical analysis with general bargaining power and Section \ref{sec:revised:theory:beliefs} where we estimate bargaining power parameters. \hfill $\blacksquare$
\end{remark} 

\begin{remark}[Larger Investors] \label{rem:TI:L}
In our analysis of the TI scenario above, we consider investors with limited budgets to fund the startup, leading to complementary roles for the two relatively small investors, i.e., the entrepreneur would benefit from the full investment by both investors. In an alternative setting, each investor may have sufficient capital to unilaterally provide the startup with adequate amount of funding. In this case, the two investors are larger than before and are effectively substitutes in the sense that if the entrepreneur receives the full investment from one, they will not need the investment from the other. We call this scenario TI-L. To conserve space, the analytical results for this scenario are provided in Online Appendix~\ref{sec:TIL_App}.\hfill $\blacksquare$
\end{remark}


\subsubsection{Preferred Stock Contract \label{sec:poor:ent:contracts}} 
With Preferred Stock contracts investors receive downside protection in the form of liquidation preferences. Specifically, we set downside protection to be \textit{exactly equal} to the investment amount. This means, in the high state of the world ($\alpha=\alpha_H > 1$), the value $V$ is divided according to the negotiated shares as long as the investors' share is sufficient to cover their investment amount. If it is not, investor(s) receive their investment amount back. Further, in the low state of the world ($\alpha=\alpha_L = 1$) investor(s) receive their investment back. Thus, under Preferred Stock contracts investor(s) are insured against potential losses in both states of the world.\endnote{In practice, the extent to which the investor is protected from potential losses (sometimes referred to as ``Liquidation Multiple'') may be set endogenously by the negotiators. To simplify the analysis and the experiment, we assume an exogenous liquidation multiple of 1.} Apart from this downside protection for investor(s), the model remains the same as in the Common Stock contract. 

To distinguish between contracts, we use ${s}_i^t$ and $I_i^t$ (resp., $\tilde{s}_i^t$ and $\tilde{I}_i^t$) to denote the equilibrium share and investment amount for Common Stock (resp., Preferred Stock) contracts with $i\in\{e, 0, 1, 2\}$ and $t\in\{SI,~TI\}$.  Much of the analysis is analogous to the Common Stock contracts. Therefore, we only present the main results below. Detailed analysis and proofs are available in Appx. \ref{sec:protection_appendix}.

\textbf{Single Investor (SI)}  Under Preferred Stock contracts Investor 0 is paid up to $I_0$ before the entrepreneur receives any proceeds. This is true in both states of the world. Recall that the low state multiplier $\alpha_L=1$. Thus, in the low state of the world, Investor 0 receives exactly $I_0$ while the entrepreneur receives nothing.  The following proposition summarizes the bargaining outcome.
\begin{proposition}[Single investor bargaining]\label{prop:single_protection}
The investor invests $\tilde{I}_0^{SI}=e$. The shares are as follows:
\begin{align*}
\tilde{s}_0^{SI}=\frac{\alpha_H    + 1}{2 \alpha_H }-\frac{ d_\text{e}}{2 e \alpha_H  p}, \quad \tilde{s}_\text{e}^{SI}=1-\tilde{s}_0^{SI}=\frac{\alpha_H  - 1}{ 2 \alpha_H }+\frac{ d_\text{e}}{2 e \alpha_H  p}.
\end{align*}
\end{proposition}

\textbf{Two Investors (TI)} Under Preferred Stock contracts both investors, if they choose to invest, receive at least their endowments back in both states of the world. The following proposition summarizes the equilibrium bargaining outcomes in these scenarios.
\begin{proposition}[Two investor bargaining]   \label{prop:sim_protection}
There exists an equilibrium bargaining outcome in which both investors invest; i.e., $\tilde{I}_i^{TI}=\nicefrac{e}{2}$ for $i\in\{1,2\}$. The equilibrium shares are as follows:
\begin{align*}
\tilde{s}_i^{TI}=\frac{\alpha_H+1}{5\alpha_H}-\frac{d_\text{e} }{5 e \alpha_H p  },~i=1,2, \quad \tilde{s}_\text{e}^{TI}=1-\tilde{s}_1^{TI}-\tilde{s}_2^{TI}=\frac{3 \alpha_H -2}{5 \alpha_H}+\frac{2 d_\text{e} }{5 e \alpha_H p }.
\end{align*}

\end{proposition}

\noindent Note that Propositions \ref{prop:sim} and \ref{prop:sim_protection} provide existence results for two investor scenarios. For the parameter values used in our experiments, these equilibria are also unique. Details are provided in Appx.~\ref{sec:appx}.

\subsection{Experiment Overview \label{sec:overview:design}}
To test the Nash-in-Nash theory we conducted a series of incentivized human-subjects experiments. Below we present the experiment design and the hypotheses derived from our theoretical analysis.

\subsubsection{Experimental Design\label{sec:parameters}}
The experiments were organized into a two (PoorEnt vs RichEnt, varied between-subjects) $\times$ two (Single Investor vs. Two investors, varied between-subjects) $\times$ two (Common Stock vs Preferred Stock contracts, varied within-subjects) design. Table \ref{tab:treat} summarizes the design and reports session numbers and sample sizes. A total of 290 subjects participated in the experiments across 24 independent sessions. Subjects were recruited at a large public US university. Each subject experienced both Common and Preferred Stock contracts. Sessions were conducted on weekdays across multiple days of the week and times of day, with pre-scheduled session slots randomly allocated to treatments to mitigate potential confounds from session timing. To ensure that our results are robust to learning, we conducted four additional sessions that included three practice rounds before the main negotiation rounds. These sessions were conducted for the SI treatment. In the practice rounds, negotiators were presented with a simplified version of the problem. Additionally, to ensure that there were no order effects of contracts, three of the six regular SI sessions used a reverse sequence of contracts (first Preferred Stock, then Common Stock). Please see Appx. \ref{sec:RobustnessRegression} for detailed analysis, which shows that our results are robust to both the inclusion of practice rounds and to an alternative sequence of contracts. 

\begin{table}[b!] 
\footnotesize
\centering
{
  \caption{Summary of Treatments, Rounds and Subject Numbers  \label{tab:treat}}  
    \begin{tabular}{L{2.5cm} L{6cm} C{2cm}  C{2cm}}
    \toprule
  \textbf{Treatment arm  }      & \textbf{Treatment condition  } & \textbf{Number of sessions}  & \textbf{Number of subjects} \\
    \midrule
              \addlinespace              
\multirow{4}{2.5cm}{PoorEnt} & Single Investor (SI) &    \\
      & \quad Regular Sessions &  6 & 74 \\
      & \quad Additional sessions with practice rounds &   4 & 40 \\
      \addlinespace
      & Two Investors (TI)   & 6 & 72 \\
      \addlinespace
      \hdashline
      \addlinespace
\multirow{2}{2.5cm}{RichEnt} & Single Investor (SI)   & 4 & 50 \\ 
      & Two Investors (TI) & 4 & 54 \\       
   
    \bottomrule
             
        \end{tabular}
        
        \smallskip
\begin{minipage}{0.8\textwidth}\scriptsize \textit{Notes:}  Of the ten sessions of the SI-PoorEnt treatment, four sessions were conducted with a practice round phase. The remaining six sessions were conducted without the practice rounds, with half of the sessions having a reverse ordering of the contracts (Preferred $\longrightarrow$ Common). This was done to check for robustness to learning. 
\end{minipage}
}
\end{table}%

\subsubsection{Experimental Parameters and Theory-Motivated Hypotheses\label{sec:parameters}}
To test the theory predictions, we chose the following parameters. First, in all scenarios, the total amount to invest is 200 units. In SI, this comes entirely from the single investor, while in TI, each investor can only provide half of the maximum funding, i.e., 100 units. In Appx. \ref{sec:TI-Alt} we present the experiment design and analysis for the alternative case with two larger investors, i.e., TI-L mentioned in Remark \ref{rem:TI:L}, where each investor can alone provide the maximum funding. If the negotiations fail, the respective investor's payoff is simply their initial endowment. We set parameter values $(\alpha_H,\alpha_L,p) = (11,1,0.2)$. That is, the value of the start-up is 11 times the invested amount if it is successful and equal to the invested amount if it is unsuccessful. Further, the probability of success is 20\%.  A low value of $p$ and a high value of $\alpha_H$ are reflective of the entrepreneurial context in which there is a small probability of large profits and a large probability of failure. The expected return on investment is $\mathbbm{E}[\alpha]=0.2\times11 + 0.8\times1=3$; thus the expected size of the pie is also held constant at $200\times3 = 600$ in all treatments, provided that all agreements are secured. These parameters were chosen to (a) yield unique equilibrium predictions (b) produce noticeable differences in the anticipated treatment effects, and (c) facilitate calculations and intuition building for untrained participants in the lab. 
 
In addition to the number of investors (one vs. two), the key experimental manipulations are the entrepreneur's outside option (PoorEnt vs RichEnt) and the contract type (Common vs Preferred). In PoorEnt, the entrepreneur's outside option, $d_\text{e}$ is set to 0. Thus, if the negotiations fail, the entrepreneur walks away with nothing. In RichEnt, the entrepreneur's outside option, $d_\text{e}$ is set to 160. This is close to the investors' (combined) outside option of 200, resulting in similar bargaining powers between the entrepreneur and the investor(s). At the same time, walking away from the negotiations yields a substantially lower payoff than the expected value of the venture under full investment (under most reasonable splits), resulting in an incentive for both parties to come to an agreement. Under Common Stock contracts each party is paid based on their share, while under Preferred Stock, the investor is guaranteed to be paid back their initial investment, before the entrepreneur receives any money.\endnote{To account for the possibility of partial investments, the outside options of both parties are pro-rated. For example, in the TI case, if the entrepreneur only secures 100 out of the maximum of 200 units, then the entrepreneur retains half of their outside option, i.e., 80 units.} 
  
Under these parameterizations, Corollary \ref{corollary: ent_share} holds for both PoorEnt and RichEnt scenarios: the entrepreneur obtains a larger share in the two-investor scenario (TI) than in the single investor scenario (SI). Further, a straightforward comparison of the investor shares in Propositions 1 and 3 (for SI) and in Propositions 2 and 4 (for TI) shows that the entrepreneur's share should go up with Preferred Stock relative to Common Stock (See Table \ref{tab:theory:predictions} for calculations). Thus, we hypothesize: 

\vspace{0.2cm}

\noindent \textbf{H1 (Number of Investors):} \textit{In both PoorEnt and RichEnt, the entrepreneur's share is higher under TI than under SI.} 

\vspace{0.2cm}
\noindent \textbf{H2 (Contracts):} \textit{Holding the size of the entrepreneur's outside option and the number of investors constant, the entrepreneur obtains a smaller share under Common Stock contracts than under Preferred Stock contracts.}

In addition to looking at shares, we will examine expected profits in each treatment. Because the theory prediction is that investors should always invest in equilibrium and total investment amounts are held constant across treatments, a comparison of expected profits should yield the same result as a comparison of shares (conditional on the contract). 

\section{How Does the Number of Investors Affect Bargaining? \label{sec:num:inv}}
In this section, we experimentally examine the effects of the number of investors on the negotiation outcomes (i.e., test Hypothesis 1).  We first examine scenarios where the entrepreneur is ``poor'' and has no outside options if the negotiations fail. This is to represent the bargaining dynamics that is tipped towards the investors, as is often the case for an early-stage startup. We then examine scenarios in which the entrepreneur has a non-zero outside option and is on a more equal footing with the investor(s). This section focuses on Common Stock contracts (Preferred Stock contracts are deferred to Section \ref{sec:preferred}).\endnote{We do not report detailed results for TI vs. SI comparisons under Preferred Stock contracts to conserve space. However, we note that all H1 test results hold under Preferred Stock contracts, with the same significance levels as under Common Stock.}  

\subsection{Experimental Setup}
\subsubsection*{Experimental Protocols}
At the beginning of each session, subjects were assigned the role of either an entrepreneur or an investor. Subjects kept that role for the duration of the experiment. For brevity, we will refer to subjects by their role: entrepreneur or investor. The experiment was programmed in oTree \citep{chen2016otree} and conducted virtually via Zoom, using a protocol that was adapted from \citet{Zoom} and \citet{Zoom-Mich}.  Each subject could participate in only one session and, within a session, participated in several rounds of the same treatment. At the beginning of each round subjects were randomly matched into a dyad (SI treatments) or triad (TI treatments).\endnote{The number of rounds played by each subject varies by treatment (between 6 and 10); this was done to ensure that each experimental session would last no longer than 90 minutes. The analysis presented in the main text uses two-sided comparisons of subject-level averages to test hypotheses. However, controlling for round and experimental session does not affect the direction or significance of our results.} Subjects were paid for one randomly selected round and we did not reveal the realized startup value in any round until after all rounds were completed. This was to avoid wealth effects and ensure incentive compatibility \citep{CSS98,Azrieli}. Average dollar earnings were \$17.50 (min. \$5; max. \$48.40).

\subsubsection*{Negotiation Format}
In each condition (SI, TI), the entrepreneur engages in bilateral negotiations with each investor present. Importantly, in TI, investor $i$ cannot see the offers (or any agreement) made between the entrepreneur and investor $j$. The negotiation format is semi-structured: players can make or accept offers specifying a share of the realized startup value (between 0 and 100\%) that the investor will receive in exchange for their investment. Further, investments are all or nothing. That is, if the investor and entrepreneur reach an agreement, it is for the full amount of their endowment. No structure is placed on who may propose first, or on the order of proposals. Players do not see each other and may not exchange verbal messages during the negotiations. At the end of each round, the results of all negotiations are announced to all players in a dyad/triad. See Fig. \ref{fig:screen:SI}-\ref{fig:screen:TI:S} for screenshots of the negotiation interface. 

In each period, the players in a negotiation dyad/triad have 90 seconds in the SI treatment or 180 seconds in the TI treatment to reach an agreement.\endnote{We chose 90 seconds based on the bilateral bargaining times commonly used in the experimental literature \citep{Isoni}. We opted for a doubling of the bargaining time in the TI treatments, relative to SI. The rationale for this design choice is that the TI treatments included two bilateral negotiations, between the entrepreneur and investor $i$, as well as between the entrepreneur and investor $j$.} If no agreement is reached with any investor, then all players receive their outside options. In the two investor cases, the entrepreneur can secure deals with zero, one or both investors.  If negotiations succeed with only one investor, then the excluded investor receives their initial endowment, while the entrepreneur and investor who did reach an agreement are paid according to the terms of the agreement and the realized value of the startup. The complete set of experimental instructions and procedures is reproduced in Appx. \ref{sec:instructions}.

\subsection{Results\label{sec:shares:results}}
Recall that in equilibrium, the entrepreneur and the investor(s) always come to an agreement (Propositions \ref{prop:single}-\ref{prop:sim}). Further, equilibrium total investment amounts  and expected size of the pie (to be divided) are set to be the same in all treatments.  Thus, H1 applies to both the share of the startup received by the entrepreneur, as well as to the expected profits of the entrepreneur. We will examine both metrics when reporting our results. 

\subsubsection*{PoorEnt\label{sec:poor:ent:results}}
Table \ref{tab:poor:ent:summary} provides a summary of results for the PoorEnt case. Consider first the frequency of agreements in the left panel of Table \ref{tab:poor:ent:summary}. We decompose agreements into subcategories for full/efficient investment ($I_T = 200$), partial/inefficient investment ($I_T = 100$) and disagreement (i.e., $I_T = 0$).  As can be seen, full agreements are achieved between 77\% and 83\% of the time. Not surprisingly, full agreement is directionally less common in TI treatments ($p=0.134$). At the same time, negotiating with multiple investors has the advantage that complete disagreement is exceedingly rare in the TI treatments (2.08\% of the time vs. 17.08\% in SI, $p \ll 0.01$). That is, with two investors the entrepreneur is almost always able to secure at least a partial investment.\endnote{The increase in agreement rates can also be interpreted as an increase in total expected surplus. This is because, given our parameterizations, the expected value of the startup is always greater than the sum of the outside options of the bargaining parties.}

Consider next the middle panel, which shows the final share negotiated by the entrepreneur. Focusing on efficient agreements, the results show that the two investor case leads to a worse outcome for the entrepreneur relative to the single investor case ($p=0.006$), with a difference of about 8.5 percentage points. Finally, consider the right panel in which we compute expected profits for the entrepreneur. The first row shows expected profits conditional on agreement. Given that these conditional profits are computed by multiplying the shares by a constant factor, the comparison is identical to the comparison of shares in the middle panel  ($p=0.006$). However, if we consider overall expected profits (after accounting for the partial agreements), these are not significantly different between treatments ($p=0.246$). This null result is driven by a relatively high rate of partial agreements in TI, which allows the entrepreneur to receive a non-zero payoff even when one of the investors does not invest. 

\subsubsection*{RichEnt\label{sec:rich:ent:results}}
Table \ref{tab:rich:ent:summary} repeats the analysis for the RichEnt treatment arm. As shown in the left panel, the bargaining environment becomes considerably more challenging: efficient agreements occur less often in RichEnt compared to PoorEnt ($p = 0.007$), while partial agreements and disagreements are more frequent in RichEnt ($p = 0.007$ and $p = 0.756$, respectively). Although not predicted by theory, this is unsurprising. As  the entrepreneur's outside option $d_e$ increases from 0 (in PoorEnt) to 160 (in RichEnt), the range of agreements benefiting both parties shrinks substantially, reducing incentives to reach an agreement. This effect is amplified by uncertainty in the startup's value, making potential losses a real concern for both parties. (Recall that the startup fails with an 80\% chance.)

\begin{table}[t]
\centering
\footnotesize
\caption{\label{tab:poor:ent:summary} Summary of Agreements, Shares and Profits in PoorEnt Treatments}
\begin{tabular}{@{}lcccC{0.15cm}cccC{0.15cm}ccc@{}}
\toprule
 &
  \multicolumn{3}{c}{\textbf{Frequency of Outcome (\%)}} &
  \textbf{} &
  \multicolumn{3}{c}{\textbf{Entrepreneur Share (\%)}} &
  \textbf{} &
  \multicolumn{3}{c}{\textbf{Entrepreneur Expected Profit}} \\ \cmidrule{2-4}\cmidrule{6-8} \cmidrule{10-12}
\textbf{Outcome Type} & SI    & TI    & $p-$value &  & SI     & TI     & $p-$value &  & SI     & TI     & $p-$value \\ \midrule
Full Investment               & 82.92 & 77.08 & 0.134     &  & 43.30  & 34.81  & 0.006     &  & 259.78 & 208.87 & 0.006     \\
Partial Investment          & ---  & 20.83 & ---       &  & ---    & 71.86  & ---       &  & ---    & 215.57 & ---       \\
No Investment         & 17.08 & 2.08  & $\ll0.001$       &  &  --- & --- & ---       &  & 0.00   & 0.00   & ---       \\
Overall  &  ---  & ---   & ---       &  & ---    & ---    & ---       &  & 215.56 & 195.38 & 0.246     \\ \bottomrule
\end{tabular}
    \vspace{-3mm}

\begin{minipage}{0.98\textwidth} \scriptsize
Note: $p-$values are based on $t-$tests on subject averages (focusing on entrepreneurs only to avoid double-counting).
\end{minipage}
\end{table}

\begin{table}[t]
\centering
\footnotesize
\caption{\label{tab:rich:ent:summary} Summary of Agreements, Shares and Profits in RichEnt Treatments}
\begin{tabular}{@{}lcccC{0.15cm}cccC{0.15cm}ccc@{}}
\toprule
 &
  \multicolumn{3}{c}{\textbf{Frequency of Outcome (\%)}} &
  \textbf{} &
  \multicolumn{3}{c}{\textbf{Entrepreneur Share (\%)}} &
  \textbf{} &
  \multicolumn{3}{c}{\textbf{Entrepreneur Expected Profit}} \\ \cmidrule{2-4}\cmidrule{6-8} \cmidrule{10-12}
\textbf{Outcome Type} & SI    & TI    & $p-$value &  & SI     & TI     & $p-$value   &  & SI     & TI     & $p-$value   \\ \midrule
Full Investment             & 77.60 & 72.53 & 0.284     &  & 48.35  & 34.37  & $\ll 0.001$ &  & 290.08 & 206.20 & $\ll 0.001$ \\
Partial Investment          & ---   & 23.30 & ---       &  & ---    & 68.05  & ---         &  & ---    & 284.14 & ---         \\
No Investment         & 22.40 & 4.17  & $\ll0.001$       &  & --- & --- & ---         &  & 160.00 & 160.00 & ---         \\
Overall   & ---   & ---   & ---       &  & ---    & ---    & ---         &  & 260.67 & 221.12 & 0.005       \\ \bottomrule
\end{tabular}
    \vspace{-3mm}
    
\begin{minipage}{0.98\textwidth} \scriptsize
Note: $p-$values are based on $t-$tests on subject averages (focusing on entrepreneurs only to avoid double-counting).
\end{minipage}
\end{table}

Next, consider the middle panel. As in the PoorEnt case, if we compare the final negotiated shares, entrepreneurs are worse off in TI relative to SI ($p=0.000$). Indeed, the gap between TI and SI grows from 8.5 percentage points in PoorEnt to about 14 percentage points in RichEnt. This is driven largely by an increase in the entrepreneur's bargaining power in the SI treatment, where the entrepreneur is now able to secure 48.35\% of equity compared to 43.30\% in the SI treatment in the PoorEnt case, i.e., an increase of about 5 percentage points ($p=0.027$). In contrast, the TI shares are not significantly different between PoorEnt and RichEnt ($p=0.918$). As before, it is also informative to examine the treatment differences in expected profits.  With increased bargaining power, expected profit (in the right panel) results now mirror negotiated share results (in the middle panel): the entrepreneur is uniformly better off with a single large investor than with two smaller investors.\endnote{Tables 2 and 3 include the data from all experimental rounds. However, it is plausible that learning affects the results differently in each treatment. To ensure that our results were not affected by learning, we replicated the analysis using the last round in each treatment. Our test results fully replicated under this alternative specification.}
\begin{result} \label{res:shares}
H1 is not supported. Contrary to H1, the entrepreneur's share is higher when negotiating with a single large investor (SI) than with two investors (TI). In addition, the entrepreneur's expected profits are higher under SI when the entrepreneur is rich. However, due to a frequent occurrence of partial investments, expected profits are not significantly different between SI and TI when the entrepreneur is poor. 
\end{result}

\begin{remark}[Results with Larger Investors]
    \label{rem:TI:L:results} In Appx. \ref{sec:TI-Alt} we present  results for the large investor case (TI-L). The richer set of equilibria in that environment (characterized in EC.1.1.3) makes the theory predictions contingent on the inclusion of each investor into the final agreement. However, what stands out is that entrepreneurs do best with exclusionary agreements in which one investor provides full investment amount. In this case, the entrepreneur obtains a higher share in TI-L than in SI. In contrast, for agreements in which both investors invest, the results are not qualitatively different from TI.
\end{remark}

\subsection{Bargaining Process}
To better understand Result \ref{res:shares} we examine the bargaining processes in each treatment (detailed analysis of the bargaining process is available in Appendix \ref{sec:bargaining:appendix}. We first show that the differences in bargaining outcomes between SI and TI conditions are primarily driven by initial offers rather than differences in concession behavior. When entrepreneurs make more generous initial offers (i.e., propose larger shares to investors), they significantly increase the chances of successfully reaching an agreement. Conversely, higher initial demands (asking for a larger share) by investors tend to lower agreement rates, significantly so in the SI condition. Moreover, initial positions strongly predict final outcomes: more generous initial offers by entrepreneurs, as well as higher initial demands by investors, lead to investors ultimately obtaining larger equity shares and vice versa for entrepreneurs. Final agreements tend to converge around the midpoint between the entrepreneur's and investor's initial proposals, suggesting that both parties make substantial concessions. This suggests a strategic trade-off faced by negotiators: setting generous initial terms facilitates reaching an agreement but reduces one’s own final payoff. The stronger presence of this trade-off in SI helps explain the observed differences in agreement rates between SI and TI.

The analysis also suggests that entrepreneurs and investors adjust their opening strategies based on relative bargaining power, which drives the observed final outcomes in PoorEnt and RichEnt. Investors demand roughly six percentage points less when entrepreneurs have a positive outside option (RichEnt) than when entrepreneurs have no outside options (PoorEnt), indicating some responsiveness to the shift in bargaining power. However, these adjustments are minimal in the TI condition -- less than one percentage point -- suggesting that entrepreneurs are unable to leverage their improved outside option effectively when negotiating with two investors. This lack of strategic adjustment in TI, combined with the convergence of outcomes towards the midpoint of the negotiators' initial positions, explains the difference in final outcomes and highlights the limits of the entrepreneur's leverage in multi-party bargaining.

\section{How do Common/Preferred Stock Contracts Affect Bargaining?  \label{sec:preferred}}
We now turn our attention to the effects of contract type (Common/Preferred) on bargaining outcomes. Recall that Preferred Stock contracts eliminate the risk to investors of losing money. In Section  \ref{sec:theory}, we used Nash-in-Nash theory to show that this risk transfer should theoretically reduce investors' shares (Hypothesis 2). We next present a test of this hypothesis.

\subsection{Results}
As before, we report treatment means for agreement rates and entrepreneur shares. The results are summarized in Table \ref{tab:preferred:summary}. Consider first the frequency of agreements reported in the left part of the table. Agreement rates are somewhat higher under Preferred relative to Common Stock contracts in the PoorEnt treatment arm ($p=0.091$) but not in the RichEnt treatment arm ($p=0.497$). Further, consistent with the results in Section \ref{sec:shares:results} we observe a drop in agreement rates in RichEnt relative to PoorEnt for both the SI and TI scenarios under Preferred Stock contracts  (SI: $p < 0.001$; TI: $p =0.762$). A larger outside option for the entrepreneur makes it more challenging for the bargaining parties to come to a mutually beneficial agreement.

The middle panel of Table \ref{tab:preferred:summary} shows the final shares obtained by the entrepreneur in each scenario. Several observations are in order. First, unsurprisingly, in three out of four pairwise comparisons, the entrepreneur's share goes up as the entrepreneur's outside option goes up. 
This is a useful manipulation check that confirms that negotiations respond to the presence of outside options consistent with what bargaining theory would predict. Second, recall that H2 states that the entrepreneur's share should be higher under Preferred Stock than under Common Stock. We find only partial support for this hypothesis. Indeed, the contract effect depends on whether the entrepreneur is poor or rich. In the PoorEnt case, the shares are at most 1.5 percentage points apart between contracts conditional on the number of investors (SI: $p=0.837$, TI: $p=0.590$, and Pooled: $p=0.948$). In RichEnt, however, the entrepreneur is able to obtain a significantly higher share under Preferred than under Common Stock contracts (SI: $p=0.087$, TI: $p=0.017$, Pooled: $p=0.007$). These results suggest that the entrepreneur's outside option serves as an important moderator on the effect of contracts on equity division. 

\begin{table}[t!]
\centering
\footnotesize
\caption{\label{tab:preferred:summary} Summary of Agreements and Shares with Common and Preferred Stock Contracts}
    \begin{tabular}{llccrccrC{3.6cm}}
    \toprule
        &        & \multicolumn{2}{C{3.2cm}}{\textbf{Frequency of efficient agreements (\%)}} &       & \multicolumn{2}{C{3cm}}{\textbf{Entrepreneur share (\%)}} &       & \textbf{Entrepreneur share comparisons ($p$-value)} \\  
          &       &  Common &  Preferred &       & Common & Preferred &       & Common vs. Preferred \\
          \cmidrule{3-4}\cmidrule{6-7}\cmidrule{9-9}   
   \textbf{Treatment arm} & \textbf{Condition}       &    &    &       &   &       &  \\
             \cmidrule{1-2}
\multirow{3}{*}{PoorEnt} & SI      &  82.92 &  87.84 &       & 44.30  & 42.97 &       & 0.837 \\
                         & TI  &  77.08 &  79.29 &       & 34.81 & 36.09 &       & 0.590 \\
                         & Pooled  & 80.72  & 84.58  &       & 41.13 & 41.22 &       & 0.948 \\                                       
\cmidrule{1-2}
\multirow{3}{*}{RichEnt} & SI    & 77.60 &  76.00 &       & 48.35 & 52.28 &       & 0.087 \\
                         & TI  & 72.53 &  77.78 &       & 34.37 & 38.71 &       & 0.017 \\
                         & Pooled& 74.97 &  76.92 &       & 42.49 & 46.61 &       & 0.007 \\
    \bottomrule
    \end{tabular}%
    \vspace{0.1cm}
\begin{minipage}{0.9\textwidth} \scriptsize
Note: The $p-$values are derived from paired $t-$tests on subject averages of the average entrepreneur share by contract type for the relevant treatment condition/arm.
\end{minipage}
\end{table}
  \begin{result}
  \label{res:common:preferred} H2 is partially supported. Preferred Stock contracts lead to higher entrepreneur shares, but only when the entrepreneur has a strong outside option.
  \end{result}
 \subsection{Bargaining Process and Fairness}
\subsubsection*{Bargaining Process}
To better understand Result \ref{res:common:preferred}, in Appendix \ref{sec:bargaining:appendix} we examine the bargaining behaviors in each treatment and under each contract. As before, we see that the differences in final outcomes emerge primarily from the differences in opening offers. Entrepreneurs make lower first offers to investors under Preferred Stock contracts relative to Common Stock contracts (except in the SI-PoorEnt scenario). This behavior is consistent with theory, which predicts that entrepreneurs should be compensated for the additional risk they carry under Preferred Stock contracts. However, investors do not appear to respond to contracts in the same way. Indeed, their opening demands are virtually identical across contracts in the SI treatments ($p=0.500$ and $p=0.858$). Even more surprisingly, in the TI scenarios investors make significantly \textit{more} aggressive opening demands under Preferred Stock than under Common Stock (increasing from approximately 44\% to approximately 49\% in both PoorEnt and RichEnt conditions, both $p<0.01$). The latter behavior contradicts theoretical expectations but may be plausibly explained by peer-induced fairness -- each investor's desire to outperform their peer investor \citep[See][for similar behaviors in ultimatum games and in supply chains]{ho2009,ho2014}. Without a peer investor in SI scenarios, this aggressive behavior diminishes. To better understand the drivers of these behaviors, we next examine the fairness beliefs of the negotiators. 

\subsubsection*{Fairness Beliefs}
At the end of the experiment, we administered a short survey about the subjects' fairness perceptions. The survey questions were based on \citet{BLIC1995} and asked subjects: \textit{``According to your opinion, from the vantage point of a non-involved neutral arbitrator, what would be a `fair' share of the business that should go to the entrepreneur?''} This question was asked separately for Common and Preferred Stock contracts. Figure \ref{fig:fairness} summarizes the results by entrepreneur's outside option (PoorEnt vs. RichEnt) and contract (Common vs. Preferred). All measurements are with respect to the entrepreneur's share, with the bars showing negotiation outcomes and the horizontal lines indicating fairness beliefs.\endnote{About three percent of participants in TI  reported fairness beliefs contracts that did not sum up to 100\% (Three participants for Common Stock, five participants for Preferred Stock contracts). Their data is excluded from the analysis of fairness beliefs.}
\begin{figure}[tbh!]
\centering\caption{Fairness vs. Negotiation Outcomes \label{fig:fairness}}
\includegraphics[width=\textwidth]{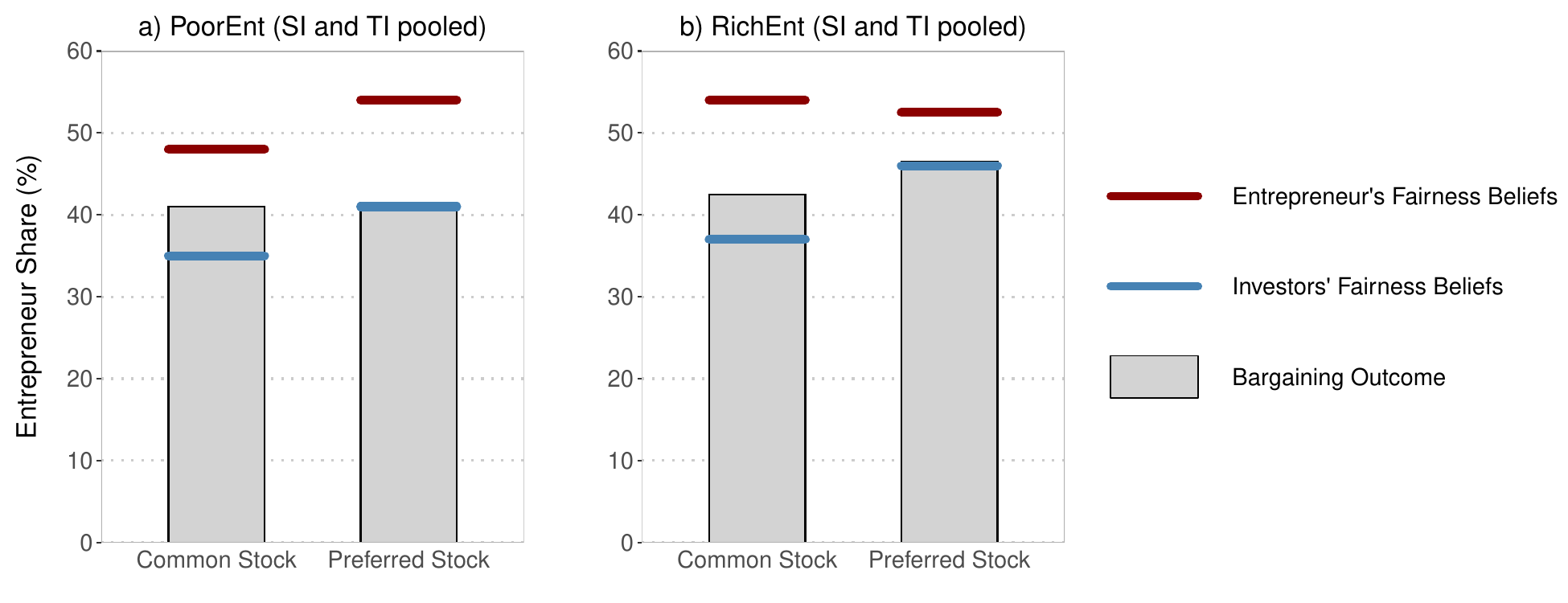}
\end{figure}

Several observations are in order. First, unsurprisingly, in all four cases, the investors' fairness beliefs are substantially below those of the entrepreneurs' (average across all treatments and contracts, 39.10 vs 52.35, $p\ll0.01$).  Second, in all four cases, the final bargaining outcome is closer to the investors' fairness beliefs than to those of the entrepreneur. That is, investors were better able to employ tactics to pull the outcome closer to their perceived fair outcome.  In particular, for entrepreneurs, agreements are between 5.8 and 13.1 points below their fairness beliefs (all $p \ll 0.01$). For investors, agreements were above their fairness beliefs under Common Stock contracts, and close to their fairness beliefs under Preferred Stock contracts (with neither difference being statistically significant). Third, while contract type does not always affect the negotiation outcomes (Table \ref{tab:preferred:summary}), it does significantly shift what the parties consider fair: fairness beliefs for the entrepreneurial share increase significantly as we go from Common to Preferred Stock contracts,  for both the PoorEnt (40.53 vs 46.60, $p\ll0.01$) and the RichEnt case (44.87 vs 48.96, $p=0.048$),  which is consistent with Hypothesis 2.\endnote{The one exception is for entrepreneurs in the RichEnt treatments, where entrepreneurs' fairness beliefs are higher for Common than for Preferred Stock contracts.}

Taken together, the analysis of fairness beliefs suggests that investors know that entrepreneurs \textit{should} receive a larger share under Preferred Stock contracts. Nevertheless, investors bargain aggressively to maximize their own share. While this behavior may appear counterintuitive at first, it is consistent with \cite{cui2016}, who show that institutional norms and game structure (in their case, being a first-mover in a sequential game) can affect fairness perceptions. In our context, the institutional norm of investor downside protection in Preferred Stock may grant more aggressive bargaining despite investors' awareness of the entrepreneur's weaker position.\endnote{Indeed, such behavior was even noted by investors in the exit survey, with one investor reporting, \textit{``It was great as I can freely invest without fear of [losing] my investment. Can take a lot of risk in this scenario.''} Another investor said, \textit{``I tried to be more aggressive. I would get back my money either way if it failed.''}}  This has implications for real-world bargaining. First, contract structure alone does not determine outcomes; contextual factors like entrepreneur resources and investor competition significantly moderate effects. Second, investors may strategically compensate for theoretical disadvantages through more aggressive opening positions when competing with other investors. This suggests that entrepreneurs should carefully consider how to position themselves in multi-investor negotiations, where investors might counterintuitively adopt more extreme positions despite theoretical bargaining disadvantages. In the next section, we will examine alternative ways to model beliefs in three-party negotiations which offer a theory-based explanation and a theory refinement that reconciles observed results.
 
\section{Theory Refinements \label{sec:discussion}}
Our results so far offer mixed support for the Nash-in-Nash model predictions. In this section, we revisit the model and propose several refinements that help reconcile the model with our experimental results. First, in Section  \ref{sec:revised:theory:beliefs} we propose an alternative approach for modeling off-equilibrium beliefs and show that this approach more closely aligns model predictions with data (Result \ref{res:shares}). Second, in Section  \ref{sec:revised:theory:contracts} we propose a risk exposure framework that provides a plausible explanation for the underreaction to contract type when the entrepreneur is ``poor'' (Result \ref{res:common:preferred}).

\subsection{Off-Equilibrium Belief Modeling \label{sec:revised:theory:beliefs}}
Contrary to model predictions, access to multiple investors does not necessarily translate into a negotiation advantage for the entrepreneur (Result 1). To arrive at this prediction, our theoretical development in Section  \ref{sec:theory} made two important assumptions. First, we assumed equal bargaining powers, i.e., $\theta_i=0.5$ across all treatments. This assumption is standard in bargaining experiments \citep[see e.g.,][among others]{AnbarciFeltovich2013,AnbarciFeltovich2018,EHR2020}. However, the presence of the entrepreneurial context may signal to the negotiators which party holds more leverage, potentially leading to unequal bargaining power. Therefore, it is informative to examine a more flexible specification in which bargaining power is estimated from the data rather than imposed ex-ante. 

Relaxing the equal bargaining power assumption by explicitly estimating bargaining power parameters ($\theta_i$) within the Nash-in-Nash framework  substantially improves model fit. Evidence of this improvement is provided in Table \ref{tab:leverage}, where we compare the mean squared errors (MSE) of the original model (equal bargaining power) with Revised model (i), which incorporates estimated bargaining parameters. As shown, MSE substantially decreases in both SI (from 0.055 to 0.040, a 27.78\% reduction) and TI (from 0.102 to 0.066, a 35.29\% reduction). 

In addition to MSE (a measure of model fit), it is also useful to examine the estimated bargaining power parameter $\hat{\theta}_i$. In the SI treatment, estimating bargaining power yields $\hat{\theta}_i = 0.355$ (or, equivalently, $\hat{\theta}_e = 0.645$), indicating relatively high entrepreneur bargaining power. This likely occurs because the bargaining parties often follow a 50-50 norm, i.e., equal splitting of firm ownership. In contrast, theory predictions specify that investors should receive a larger share due to their capital being put at risk, particularly in the PoorEnt case. Further, in the TI treatment, the estimated parameter is $\hat{\theta}_i = 1$, suggesting complete investor dominance. Such extreme estimates likely indicate model misspecification and highlight fundamental limitations in the descriptive validity of classical Nash-in-Nash theory in multiple-investor contexts. Given the random assignment of subject roles, one would expect roughly equal bargaining power across experimental roles. Furthermore, if the theory is correct, there is no reason to expect bargaining power to differ so substantially between the SI and TI treatments.

\begin{table}[bt]
\small
\centering\caption{Model Fit (Common Stock Contracts Treatments Only) \label{tab:leverage}}
\begin{tabular}{L{3.2cm}L{4.4cm}d{2.4}@{}cC{0.1cm}d{2.4}@{}c} \toprule
&& \multicolumn{2}{c}{SI} && \multicolumn{2}{c}{TI}\\  \cmidrule{3-4}\cmidrule{6-7}
Model   & Description  & \multicolumn{1}{c}{MSE} & \multicolumn{1}{c}{$\hat{\theta}_i$} &&
 \multicolumn{1}{c}{MSE} & \multicolumn{1}{c}{$\hat{\theta}_i$} \\  \cmidrule{1-2} \cmidrule{3-4}\cmidrule{6-7}
Original model (Section  \ref{sec:theory}) & Nash-in-Nash, \newline  Equal bargaining powers   & 0.055 & 0.500 && 0.102 &  0.500 \\  \addlinespace
Revised model (i) & Nash-in-Nash, \newline  Best fitting bargaining powers & 0.040 & 0.355       && 0.066 & 1.000\\  \addlinespace
Revised model (ii) & Nash-in-Nash, \newline Best fitting bargaining powers	& \multicolumn{2}{c}{No change}	&&	0.045 & 0.462  \\\bottomrule
\end{tabular}

\begin{minipage}{0.78\textwidth}\scriptsize
  Notes: The investor's bargaining power $\hat{\theta}_i$ is either set to 0.5 (Original predictions) or set to the bargaining power with the best fit for a given treatment (Revised models (i) and (ii)). In SI, $\hat{\theta}_i$ denotes the bargaining power of the single investor. In these treatments, models (i) and (ii) predictions coincide because there is no possibility of partial disagreement, and therefore no belief model. In the TI treatment, $\hat{\theta}_i$ denotes the average bargaining power of the two investors.
\end{minipage}
\end{table}
To explore a potentially more accurate modeling approach, we reconsider the second key assumption in our model: namely, that the agreement between the entrepreneur and investor $i$ is unaffected by whether the entrepreneur reaches agreement with investor $j$. This assumption gives the entrepreneur theoretical leverage in the TI treatment, since the entrepreneur can walk away from one negotiation without adversely affecting the other. Although common in the literature \citep{horn1988bilateral,Yurukoglu}, this assumption may oversimplify actual bargaining dynamics. An alternative assumption about off-the-equilibrium-path beliefs in TI is that disagreement with one investor implies disagreement with both investors. Under this revised assumption, the entrepreneur’s bargaining leverage is diminished, as walking away from one negotiation necessarily terminates the other, reducing the entrepreneur’s ability to pit investors against each other. Incorporating this modified assumption yields “Revised model (ii)” in Table \ref{tab:leverage} (model details and proofs are available in Appendix \ref{RB_sec:appx_alternative}). As shown in the last row of the table, this specification achieves a further reduction in MSE (by 31.82\%) and, importantly, yields more plausible estimates of bargaining power ($\hat{\theta}_i$).\endnote{While we do not report estimates for Preferred Stock contracts, the same general patterns emerge, but the estimates of $\theta_i$ are higher, indicating more investor bargaining power, consistent with our previous results.}

The takeaway from this exercise is that the standard modeling approach of Nash-in-Nash bargaining with equal power and independent bilateral agreements may not adequately capture the complexities of multi-party entrepreneurial negotiations. How parties form beliefs about outcomes in alternative negotiation scenarios is an important consideration, particularly in settings where agreements are not independent. Our revised model that incorporates more realistic off-equilibrium beliefs offers a better fit to the experimental data and provides more plausible estimates of bargaining power.

\subsection{Risk Exposure and Contracts \label{sec:revised:theory:contracts}}
Our second result was that the downside protection for investors (afforded by ``Preferred Stock'' contracts) was only reflected in the equity split when the entrepreneur was ``rich.'' In contrast, we did not see a fair reflection of contracts in the negotiated terms when the entrepreneur was ``poor.''

Why do poor entrepreneurs struggle to see a fair reflection of investor liquidation preferences in their contractual terms? One possibility, not currently considered in our theoretical development, is that fairness norms depend on how each party's payoffs compare to their outside option. Specifically, we can distinguish between being in the ``gain domain'' (payoff always at or above outside option) versus the ``gain or loss domain'' (payoff may fall below outside option). The gain/loss domain positions for each treatment are summarized in Table~\ref{tab:risk:exposure}.

Notably, the entrepreneur's share is significantly higher only in the RichEnt + Preferred condition (46.6\%) compared to the other three conditions (41.1\%, 41.2\%, 42.5\%), which are not significantly different from each other. This might suggest that investor concessions require two conditions to be met: (1) the entrepreneur must have skin in the game, meaning they can lose relative to their outside option, and (2) the investor must be protected from downside risk. Only RichEnt + Preferred satisfies both conditions. In PoorEnt, entrepreneurs have no outside option ($d_e = 0$) and thus always remain in the gain domain regardless of contract type, failing condition 1. In RichEnt + Common, the investor also faces potential losses, failing condition 2. The fairness beliefs reported in Section 5.2 support this interpretation: investors' stated fair shares are closer across contract types in PoorEnt than in RichEnt, and closer across outside option conditions under Common Stock than under Preferred Stock. That is, investors perceive it to be fair to concede equity only when both conditions are met.

These findings have implications for entrepreneurs negotiating with investors. Early-stage entrepreneurs with limited outside options may not benefit from offering investor protections, as the resulting risk transfer is unlikely to be reflected in more favorable equity terms. In contrast, more established entrepreneurs with viable alternatives may find that offering Preferred Stock allows them to retain a larger share of the venture. Future analytical work could further examine these negotiation dynamics by incorporating reference-dependent preferences into the Nash-in-Nash framework.\endnote{See \cite{shalev2002} and \cite{karagozouglu2018} for similar refinements in simpler, bilateral bargaining models.}

\begin{table}[bt!]
\small
  \centering \caption{\label{tab:risk:exposure}Risk Exposure Framework}
  \begin{tabular}{cccccc} \toprule
  & & Common Stock & & Preferred Stock \\  \midrule
  \multirow{6}{*}[0.15em]{\parbox{1.1in}{\centering Entrepreneur's \newline Outside Option}}
  & \multirow{2}{*}[0.5em]{$d_\text{e}=0$ (PoorEnt)}
    & \fbox{\parbox{1.7in}{\centering Entrepreneur: Gain \\ Investor(s): Gain or Loss}}
    & $\approx$
    & \fbox{\parbox{1.7in}{\centering Entrepreneur: Gain \\ Investor(s): Gain}} \\[0.15cm]
  & & $\rotatebox{-90}{$\approx$}$ & & $\rotatebox{-90}{$<$}$\raisebox{-0.3ex}{\scriptsize **}
 \\[0.45cm]
  & \multirow{2}{*}[0.5em]{$d_\text{e}>0$ (RichEnt)}
    & \fbox{\parbox{1.7in}{\centering Entrepreneur: Gain or Loss \\ Investor(s): Gain or Loss}}
    & $<$\textsuperscript{\scriptsize **}
    & \fbox{\parbox{1.7in}{\centering Entrepreneur: Gain or Loss \\ Investor(s): Gain}} \\
  \bottomrule
 \end{tabular}
\vspace{0.1cm}

 \begin{minipage}{0.99\textwidth}\scriptsize
 Note: ``Gain'' means the party's payoff is always at or above their outside option. ``Gain or Loss'' means the party's payoff may fall below their outside option depending on the startup outcome. Inequality signs indicate pairwise comparisons of entrepreneur's share between conditions; $\approx$ indicates no significant difference. Significance levels based on analysis reported in Section~5.1. \textsuperscript{**}~$p<0.05$.
 \end{minipage}
\end{table}

\section{Conclusions}
Equity negotiations are an essential part of entrepreneurial growth. In this paper, we explicitly modeled the key features of entrepreneur-investor bargaining, which includes the intrinsic value of the startup, uncertain valuation, multiple investors and the contract structures common in the industry. We used the Nash-in-Nash framework to uncover the theoretical sources of leverage and to develop hypotheses regarding the split of shares between the entrepreneur and the investors. We then conducted (virtual) lab experiments to test whether the leverage that is available in theory is exploitable in practice.

\subsection{Summary of Results and Contributions}
Our investigation offers several insights into equity bargaining. Consistent with theory, entrepreneurs benefit from having a stronger outside option. Although this result is not surprising, it is still worth emphasizing. Going beyond this simple insight, our theory also predicts that entrepreneurs should be able to benefit from negotiating with multiple investors and that they should receive a higher share when offering investors downside protection via Preferred Stock contracts. Neither of these predictions was fully supported by our data. Instead, we saw that multiple investors strongly dilute the entrepreneur's stake in the company. Further, we saw  that Preferred Stock contracts can increase the entrepreneur's share when the entrepreneur has a positive outside option in negotiations. However, Preferred Stock contracts can be harmful for the entrepreneur when the entrepreneur has no outside options.

Our results contribute to the bargaining literature by examining negotiation behavior and fairness norms in settings with uncertainty, multilateral negotiations, and risk-contingent division rules.  To reconcile our results with theory, we propose several refinements to the classic Nash-in-Nash model, as well as suggest a novel framework for bargaining under risk. Methodologically, the full-circle approach of analyzing a standard model, testing it in the lab, and then revising its key assumptions, as we have done in this paper, can serve as a template for future research in this area. 

Our results also contribute to the entrepreneurship literature. Observational data on contract structure and entrepreneurial exits \citep[][and references therein]{gornall2020,ewens2022} suggest that equity contract terms significantly affect venture performance and entrepreneurial profitability. In particular, \cite{ewens2022} show that excessive liquidation preferences tend to reduce both entrepreneur profits and overall venture valuation, even after accounting for the quality of entrepreneurs and investors. Our experimental findings in Section \ref{sec:preferred} suggest a plausible causal mechanism for this result. We show that bargaining dynamics, especially the limited negotiating leverage of entrepreneurs without viable outside options, constrain entrepreneurs’ ability to secure a favorable share of equity. Over time, perceived unfairness from such inequitable bargaining outcomes may reduce entrepreneurial motivation, ultimately leading to lower venture performance and diminished total value creation.

\subsection{Practical Implications}
Our findings inform entrepreneurs seeking equity funding. The conventional wisdom that ``the more investors, the better'', while consistent with classical bargaining models, is not uniformly true according to our data. The effects of having access to multiple investors can depend on factors such the entrepreneur's outside options (i.e., the intrinsic value of the firm) as well as on contract. Therefore, entrepreneurs should carefully consider their individual circumstances when developing negotiation strategies and proposing terms. More mature startups are particularly well-positioned to leverage their outside options with a single than with multiple investors. Further, while early-stage entrepreneurs should be cautious about downside protection in the term sheet, more mature startups will generally see a fairer reflection of their startup value in the contractual terms. 

\subsection{Limitations and Future Work\label{sec:discussion:extensions}}
Our investigation does not consider several bargaining features that may play a role in negotiations. First, in our experiment, bargaining outcomes are public once bargaining is completed, but the bilateral offer exchange is private. It may be interesting to explore behavior in a setting where the offer exchanges with the other investor can be observed by other prospective investors. This may also be more reflective of entrepreneurial pitch competitions and other large events where offers to invest can be made publicly and observed by others. Second, our experiments do not examine the matching process between entrepreneurs and investors \citep{bengtsson2010,ewens2022}. Future studies may be able to extend our setting by including a matching market, in which a heterogeneous set of investors is endogenously matched with a heterogeneous set of founders \citep[See, for example,][for similar matching experiments in the supply chain context]{leider2016}. Other interesting extensions include richer negotiation settings where investors receive some control rights in addition to equity, as well as settings with informational asymmetries about the value of the startup, or settings with a collaborative stage (in addition to the negotiation stage), where the parties invest costly effort into the venture.

\theendnotes

\begin{APPENDIX}{Additional Analysis and Results}
\counterwithin{table}{section}
\renewcommand{\thetable}{\thesection.\arabic{table}}

\section{Robustness Regressions \label{sec:RobustnessRegression}}
In this section we replicate our analysis for PoorEnt treatment (Table \ref{tab:poor:ent:summary}) using a series of random effects regressions.  Further, age and gender do not have significant effects on outcomes. Participants with entrepreneurial experience enter fewer agreements. There are also some marginally significant effects of the sequence of contracts and the presence of practice rounds (all at $p<0.1$). Despite these differences, the regression results show that the main treatment effect of TI vs. SI persists (in size and statistical significance) after including the controls.

\begin{table}[htbp]
\centering \small
\caption{Replication of Table \ref{tab:poor:ent:summary} (Using Regressions with Control Variables) \label{tab:robustness:regressions}}
{
\def\sym#1{\ifmmode^{#1}\else\(^{#1}\)\fi}
\begin{tabular}{l*{4}{cc}}
\toprule

                &\multicolumn{2}{c}{Efficient Outcome}&\multicolumn{2}{c}{Disagreement}&\multicolumn{2}{c}{Ent Share}&\multicolumn{2}{c}{Ent Exp Profit}\\
\midrule
TI              &   -0.101         &  (0.072)&   -0.110\sym{***}&  (0.025)&  -13.591\sym{**} &  (5.307)&  -81.549\sym{**} & (31.843)\\
Age             &   -0.004         &  (0.004)&    0.003         &  (0.003)&    0.118         &  (0.361)&    0.711         &  (2.163)\\
Female          &    0.045         &  (0.035)&   -0.012         &  (0.018)&    1.538         &  (2.371)&    9.227         & (14.226)\\
Ent. Experience &   -0.153\sym{***}&  (0.036)&    0.136\sym{***}&  (0.033)&    0.751         &  (2.076)&    4.505         & (12.458)\\
CommonFirst-Practice&   -0.033         &  (0.048)&    0.035         &  (0.051)&   -8.846\sym{*}  &  (4.974)&  -53.078\sym{*}  & (29.846)\\
PreferredFirst-NoPractice&   -0.058\sym{*}  &  (0.033)&    0.063\sym{*}  &  (0.037)&   -0.433         &  (5.049)&   -2.596         & (30.297)\\
Constant        &    0.960\sym{***}&  (0.105)&    0.043         &  (0.075)&   43.001\sym{***}& (10.752)&  258.006\sym{***}& (64.510)\\
\midrule
Observations    &      634         &         &      634         &         &      240         &         &      240         &         \\
$R^2$              &    0.043         &         &    0.069         &         &    0.121         &         &    0.121         &         \\
\bottomrule
\end{tabular}
}
\vspace{0.1cm} 

\begin{minipage}{\textwidth}\scriptsize
Note: Standard errors accounting for clustering at the session level in parentheses, * p$<$0.10, ** p$<$0.05, *** p$<$0.01. The third and fourth columns restrict attention to the case in which an efficient agreement was reached.
\end{minipage}
\end{table}

\section{TI-L: Two Investors With Large Endowments. Experimental Design and Results\label{sec:TI-Alt}}

This section presents the results for TI-L -- the alternative regime with two investors, where the overall endowment of the two investors exceeds the maximum investment amount. In particular, we examine the case where the sum of the two endowments is 400, i.e., twice the amount of the maximum investment of 200. A total of 96 participants were recruited for this treatment. A total of eight sessions were conducted (four sessions for PoorEnt, four sessions for RichEnt). The protocols, student pool, and other design details remained unchanged relative to the original SI and TI treatments.

Table \ref{tab:TI-ALT:Outcomes} provides a summary of the bargaining results for the TI-L treatments for both the PoorEnt and RichEnt cases. As before, we report two metrics: first, the frequency of each type of agreement, whether efficient or inefficient; second, the average entrepreneur share for each agreement type. Not surprisingly, full agreement is less common in the TI-L treatments than those reported in the main text. This is because investors can invest partial amounts. At the same time, the TI-L treatment has the advantage that complete disagreement is exceedingly rare (0.83\% in PoorEnt and 3.83\% in RichEnt). The entrepreneur is almost always able to secure at least a partial investment. Finally, of particular interest is that exclusionary agreements where one investor invests the full 200 units of capital are very common, occurring 46.80\% of the time in PoorEnt and 31.75\% of the time in RichEnt. 

\begin{table}[t]
    \centering \small
    \caption{Bargaining Outcomes in TI-L}\label{tab:TI-ALT:Outcomes}
    \begin{tabular}{L{3.5cm}lrrrr}
\toprule
    &       & \multicolumn{2}{c}{\textbf{PoorEnt}} & \multicolumn{2}{c}{\textbf{RichEnt}} \\
    \cmidrule(lr){3-4} \cmidrule(lr){5-6}
    &       & Frequency & Entrepreneur & Frequency & Entrpreneur \\
    &       & (\%)  & Share (\%) & (\%) & Share (\%) \\
\midrule
\textbf{Efficient Agreement} & \textbf{All} & \textbf{72.22}$~~$ & \textbf{50.02}$~~$ & \textbf{53.77}$~~$ & \textbf{42.54}$~~$ \\
$~~$Exclusionary & $~~(200,0)$ & 46.80 & 55.66 & 31.75 & 51.55 \\
$~~$Asymmetric & $~~(150,50)$ & 16.65 & 35.63 & 11.85 & 26.25 \\
$~~$Symmetric & $~~(100,100)$ & 8.77 & 40.24 & 10.17 & 35.50 \\
    &       &       &       &       &  \\
\textbf{Inefficient} & \textbf{All} & \textbf{27.78}$~~$ & \textbf{50.93}$~~$ & \textbf{46.23}$~~$ & \textbf{49.34}$~~$ \\
$~~$Partial & $~~0 < I_T < 200$ & 26.95 & 50.93 & 42.40 & 49.34 \\
$~~$Disagreement & $~~I_T = 0$ & 0.83 & --- & 3.83 & --- \\
\bottomrule
\end{tabular}
\end{table}
Consider next the final share negotiated by the entrepreneur. Focusing on efficient agreements, we see that entrepreneurs receive an approximately 50\% share in PoorEnt and 42.54\% in RichEnt. In the PoorEnt case, this is statistically significantly higher than the shares in both SI and TI (SI vs. TI-L: $p = 0.054$; TI vs. TI-L: $p = 0.010$). In the RichEnt case, the entrepreneur receives a larger share than in TI ($p = 0.051$) but a smaller share than in SI ($p = 0.098$). In fact, in both PoorEnt and RichEnt, the entrepreneur's share is especially high for exclusionary agreements, though the difference is only statistically significant when comparing with SI in the PoorEnt case. (Recall that exclusionary efficient agreements were not possible in TI, making such a comparison impossible.) Looking at symmetric cases where each investor invests 100 units of capital, the entrepreneur's share is higher under not statistically significantly different than in the corresponding TI cases (PoorEnt: $p=0.460$); RichEnt: $p=0.807$). In sum, there may be clear advantages to negotiating with two substitutable entrepreneurs relative to either a single investor or two complementary investors. This is especially so if the entrepreneur can negotiate an efficient and exclusionary agreement with just one investor. 

In addition to comparing outcomes, we also examined bargaining behavior. Poor entrepreneurs strongly favored exclusionary offers (52\% requested the full 200 investment), while rich entrepreneurs shifted toward smaller investments (only 25\% requested 200; the difference between PoorEnt and RichEnt is significant at $p=0.011$).  TI-L exhibited strong anchoring effects similar to other treatments -- more aggressive opening offers reduced agreement likelihood but improved outcomes conditional on agreement. Critically, entrepreneurs wanting large investments faced a trade-off: higher investment requests led to giving up more equity, but the net effect on expected profit still favored aggressive maximum-amount strategies where entrepreneurs could play the two investors against each other.
\end{APPENDIX}

\newpage
\clearpage

\ECSwitch


\ECHead{Supplementary Materials (Electronic Companion)}


\setcounter{lemma}{0}
	\renewcommand{\thelemma}{EC\arabic{lemma}}

	\setcounter{proposition}{0}
	\renewcommand{\theproposition}{EC\arabic{proposition}}

	\setcounter{corollary}{0}
	\renewcommand{\thecorollary}{EC\arabic{corollary}}
	
	\setcounter{equation}{0}
	\renewcommand{\theequation}{EC-\arabic{equation}}
	
	\setcounter{table}{0}
	\renewcommand{\thetable}{EC\arabic{table}}
	
		\setcounter{figure}{0}
	\renewcommand{\thefigure}{EC\arabic{figure}}
	
	\setcounter{section}{0}
	\renewcommand{\thesection}{EC.\arabic{section}}
	
	\setcounter{subsection}{0}
	\renewcommand{\thesubsection}{EC.\arabic{section}.\arabic{subsection}}

\section{Theory\label{sec:appx}}
We begin by presenting theoretical predictions for Common Stock contracts. The analysis of Preferred Stock contracts follows. The notation used in our theoretical exposition is listed in Table \ref{tab:notation}.

\begin{table}[b!]
\centering
\footnotesize
\caption{Notation and Definitions}
\label{tab:notation}
\begin{tabular}{p{0.22\textwidth} p{0.73\textwidth}}
\toprule
\multicolumn{2}{l}{\textbf{Model Parameters}} \\
\midrule
$\mathcal{I}$ 
  & Set of all potential investors. \\
  $\mathcal{I}' \subseteq \mathcal{I}$ 
  & Set of investors who actually invest. \\
$t \in \{\mathrm{SI}, \mathrm{TI}, \mathrm{TI-L}\}$ 
  & Bargaining institution: Single Investor (SI), Two Investors (TI), or Two Large Investors (TI-L). \\
$e$ 
  & Total capital requirement. \\
$\bar{I}^t$ 
  & Endowment of each investor under institution $t$. \\
$I^t = [0,\bar{I}^t]$ 
  & Feasible set of investments for an investor under $t$. \\
$p$ 
  & Probability that the startup succeeds. \\
$\alpha$ 
  & Multiplier that can realize as $\alpha_H$ or $\alpha_L$ indicating high/low states, with $\alpha_L < \alpha_H$. \\
$\mu_\alpha$ 
  & Expected multiplier, $\mu_\alpha = \mathbb{E}[\alpha] = p\,\alpha_H + (1-p)\,\alpha_L$. \\
  $V = \alpha \sum_{i \in \mathcal{I}'} I_i^t$ 
  & Realized value of the startup. \\
$d_e \ge 0$ 
  & Entrepreneur's outside option. \\
  $d_i$ 
  & Investor~$i$'s outside option. \\
$\theta_i \in [0,1]$ 
  & Investor $i$'s relative bargaining power. \\
\midrule
\multicolumn{2}{l}{\textbf{Decision Variables (Common Stock)}} \\
\midrule
$I_i^t \in I^t$ 
  & Investment chosen by investor $i$ under institution $t$. \\
$s_i^t \in [0,1]$ 
  & Equity share allocated to investor $i$ under institution $t$. \\
\midrule
\multicolumn{2}{l}{\textbf{Decision Variables (Preferred Stock)}} \\
\midrule
$\tilde{I}_i^t \in I^t$ 
  & Investment chosen by investor $i$ (with downside protection). \\
$\tilde{s}_i^t \in [0,1]$ 
  & Equity share allocated to investor $i$ (with downside protection). \\
\midrule
\multicolumn{2}{l}{\textbf{Additional Notation}} \\
\midrule
$\pi_i(\bm{I}, \bm{s})$ 
  & Expected payoff of investor $i$. \\
$\pi_e(\bm{I}, \bm{s})$ 
  & Expected payoff of the entrepreneur. \\
\bottomrule
\end{tabular}
\end{table}

We present the results and the proofs with the general bargaining powers of the investor(s) relative to the entrepreneur. Specifically, let $\theta_0 \in (0,1)$ denote the bargaining power of the single investor relative to the entrepreneur (i.e., the entrepreneur's bargaining power is $1 - \theta_0$) in the SI setting. Let $\theta_i\in(0,1)$, $i\in\{1,2\}$, denote the bargaining power of Investor~$i$ relative to the entrepreneur  (i.e., the entrepreneur's bargaining power is $ 1 - \theta_i$) in the TI setting. To obtain the results when the investor(s) have equal bargaining power relative to the entrepreneur, we set $\theta_i=\nicefrac{1}{2}$, $i\in\{s,1,2\}$. Recall that $\alpha$ follows a two-point distribution: $\alpha_H>1$ w.p. $p\in(0,1)$ and $\alpha_L\leq 1$ w.p. $1-p$.
\begin{assumption}\label{assumption1}
Assume that the expected investment multiplier ${\mu_\alpha}=\alpha_H p +\alpha_L (1-p) \geq 2$.
\end{assumption} 
In the analysis, if the bargaining unit between the entrepreneur and Investor~$i$ is indifferent among multiple investment levels in equilibrium, we assume that the largest investment level is made. 

\subsection{Common Stock Contracts}\label{sec:no_protection_appendix}
We consider the setting of Common Stock contracts in this section.

\subsubsection{The Single Investor Model \label{sec:single:investor_appx}}

The investment $I_0$ and the share $s_0$ maximize the following Nash product:
\begin{align}\label{eq:single_1}
\max_{I_0\in[0, e],~ s_0\in[0, 1]}~&\left[\pi_0(I_0, s_0)-d_0\right]^{\theta_0}\left[\pi_\text{e}(I_0, s_0)-d_\text{e}\right]^{1-\theta_0}\\
&\pi_0(I_0, s_0)\geq d_0,~\pi_\text{e}(I_0, s_0)\geq d_\text{e}.\nonumber
\end{align}

The following proposition is Proposition~\ref{prop:single} under general bargaining powers.
\begin{proposition}[Single investor]\label{prop:single_appx}
The investor invests $I_0^{SI}=e$. The share of the investor is
\begin{align*}
&s_0^{SI}=\frac{({{\mu_\alpha}} -1)\theta_0+1}{{\mu_\alpha}}-\frac{\theta_0 d_\text{e}}{e {\mu_\alpha}}.
\end{align*}
The corresponding entrepreneur's share is
\begin{align*}
s_\text{e}^{SI}=1-s_0^{SI}=\frac{({\mu_\alpha} -1)(1-\theta_0)}{ {\mu_\alpha}}+\frac{\theta_0 d_\text{e}}{e {\mu_\alpha}}.
\end{align*}
\end{proposition}

\noindent \textbf{Proof of Proposition~\ref{prop:single_appx}.} Recall that $d_0=e$. We also have that the expected profit of the entrepreneur is $$\pi_\text{e}(I_0, s_0)={\mu_\alpha} I_0  (1-s_0);$$
the expected profit of investor~$s$ is $$\pi_0(I_0, s_0)={\mu_\alpha} I_0  s_0+ e -I_0.$$

Solving the problem \eqref{eq:single_1} above, we have that,
\begin{align}\label{eq:single_2}
&\pi_0(I_0, s_0)-d_0=\theta_0\left(\pi_\text{e}(I_0, s_0)+\pi_0(I_0, s_0)-d_\text{e}-d_0 \right); \\
&\pi_\text{e}(I_0, s_0)-d_\text{e}=(1-\theta_0)\left(\pi_\text{e}(I_0, s_0)+\pi_0(I_0, s_0)-d_\text{e}-d_0 \right).\nonumber
\end{align}
Recall that  ${\mu_\alpha}>2$, and we have that
\begin{align*}
&I_0^{SI}=\arg\max_{I_0\in[0, e]} \left\{\pi_\text{e}(I_0, s_0)+\pi_0(I_0, s_0)-d_\text{e}-d_0\right\}= e.
\end{align*}
By Eq.~\eqref{eq:single_2}, we have that
\begin{align*}
&s_0^{SI}=\frac{({\mu_\alpha}-1)\theta_0+1}{{\mu_\alpha}}-\frac{\theta_0 d_\text{e}}{e {\mu_\alpha}}.
\end{align*}
\hfill $\blacksquare$

\subsubsection{The Two Investors Model \label{sec:simultaneous_appx}}

The investments $I_i$ and the share $s_i$ maximize the following Nash product simultaneously:
\begin{align}\label{eq:simultaneous_1}
\max_{I_i\in[0, \nicefrac{e}{2}],~ s_i\in[0, 1]}~&\left[\pi_i(\bm{I}, \bm{s})-d_i\right]^{\theta_i}\left[\pi_\text{e}(\bm{I}, \bm{s})-d_\text{e}^{-i}\right]^{1-\theta_i}\\
&\pi_i(\bm{I}, \bm{s})\geq d_i,~\pi_\text{e}(\bm{I}, \bm{s})\geq d_\text{e}^{-i}.\nonumber
\end{align}

The following proposition is Proposition~\ref{prop:sim} under general bargaining powers.
\begin{proposition}[Two investors]~  \label{prop:sim_appx} There are two types of equilibrium bargaining outcomes: both investors investing and only one investor investing. 
\begin{itemize}
\item In the both-investor-investing equilibrium, the investors invest the endowed capital; i.e., $I_i^{TI}=\nicefrac{e}{2}$ for $i\in\{1,2\}$. The equilibrium share of investor $i$ is
\begin{align}\label{eq:simultaneous_7}
&s_i^{TI}=\frac{(3-2 {\mu_\alpha})(2-\theta_i)}{{\mu_\alpha}(4-\theta_1\theta_2)}+\frac{{\mu_\alpha}-1}{ {\mu_\alpha}}-\frac{d_\text{e}(2\theta_i-\theta_1\theta_2)}{e {\mu_\alpha}(4-\theta_1\theta_2)}.
\end{align}
\item If ${\mu_\alpha}<\frac{2-\theta_i}{1-\theta_i}-\frac{2 d_\text{e} \theta_i}{ e (1-\theta_i)}$, the one-investor-investing equilibrium exists. The equilibrium investment level $I_i^{TI}=\nicefrac{e}{2}$ and $I_j^{TI}=0$ for $i,j\in\{1,2\}$ and $i\neq j$, and the equilibrium share of Investor~$i$ is
 \begin{align*}
&s_i^{TI}=\frac{({\mu_\alpha}-1)\theta_i+1}{{\mu_\alpha}}-\frac{2 d_\text{e} \theta_i}{e {\mu_\alpha}}.
\end{align*}

\end{itemize}
\end{proposition}

\noindent{\bf Proof of Proposition~\ref{prop:sim_appx}.} Recall that the expected profit of the entrepreneur is
\begin{align}\label{eq:profit-ent}
\pi_\text{e}(\bm{I}, \bm{s})= {\mu_\alpha} (I_1+I_2)  (1-s_1-s_2),
\end{align}
and the expected profit of Investor~$i$ is
\begin{align}\label{eq:profit_inv}
\pi_i(\bm{I}, \bm{s})= {\mu_\alpha} (I_1+I_2)  s_i+ \nicefrac{e}{2} -I_i.
\end{align}
The disagreement point of the entrepreneur when negotiating with Investor~$1$ is $$d_\text{e}^{-1}=\pi_\text{e}(0, I_2, 0, s_2)=\frac{d_\text{e} (e-I_2)}{e}+{\mu_\alpha} I_2  (1-s_2),$$ which is the sum of the prorated outside option and the profit of the entrepreneur when Investor~$2$ is the only investor. Similarly, the disagreement point of the entrepreneur when negotiating with Investor~$2$ is $$d_\text{e}^{-2}=\pi_\text{e}(I_1, 0, s_1, 0)=\frac{d_\text{e} (e-I_1)}{e}+{\mu_\alpha} I_1  (1-s_1).$$  The disagreement point of Investor~$i$ is $d_i=\nicefrac{e}{2}$ since the investor has $\nicefrac{e}{2}$ units of capital as the endowment.

We first solve the bargaining problem between the entrepreneur and Investor~$1$ as specified in \eqref{eq:simultaneous_1}. Following the similar analysis as in the proof of Proposition~\ref{prop:single_appx}, we have that
\begin{align}\label{eq:simultaneous_2}
&\pi_1(\bm{I}, \bm{s})-d_1= \theta_1 \left( \pi_1(\bm{I}, \bm{s})+\pi_\text{e}(\bm{I}, \bm{s})-d_1-d_\text{e}^{-1}\right);\\
&\pi_\text{e}(\bm{I}, \bm{s})-d_\text{e}^{-1}= (1-\theta_1) \left( \pi_1(\bm{I}, \bm{s})+\pi_\text{e}(\bm{I}, \bm{s})-d_1-d_\text{e}^{-1}\right).\nonumber
\end{align}
Note that the best-response investment level
\begin{align}\label{eq:simultaneous_3}
&I_1(I_2, s_2)=\arg\max_{I_1\in[0, e/2]} \left\{\pi_1(\bm{I}, \bm{s})+\pi_\text{e}(\bm{I}, \bm{s})-d_1-d_\text{e}^{-1}\right\}=\left\{
\begin{array}{ll}
 \nicefrac{e}{2}&\text{ if } {\mu_\alpha} (1-s_2)\geq 1;\\[1ex]
 0&\text{ otherwise}.\\
\end{array}\right.
\end{align}
By Eq.~\eqref{eq:simultaneous_2}, the best-response share for Investor~$1$ is
\begin{align}\label{eq:simultaneous_4}
s_1(I_2, s_2)=\left\{
\begin{array}{ll}
\frac{\theta_1\left[{\mu_\alpha} e^2  (1-s_2)-e^2 - 2 d_\text{e}(e - I_2)\right]+e^2}{{\mu_\alpha}e (e+ 2 I_2)}.&\text{ if } {\mu_\alpha} (1-s_2)\geq 1;\\[1ex]
 0&\text{ otherwise}.\\
\end{array}\right.
\end{align}

Similarly, we have that the best-response investment level and share for Investor~$2$ are
\begin{align}\label{eq:simultaneous_5}
&I_2(I_1, s_1)=\left\{
\begin{array}{ll}
 \nicefrac{e}{2}&\quad\quad\quad\quad\quad\quad\quad\text{ if } {\mu_\alpha} (1-s_1)\geq1;\\[1ex]
 0&\quad\quad\quad\quad\quad\quad\quad\text{ otherwise};\\[1ex]
\end{array}\right.\\\label{eq:simultaneous_6}
&s_2(I_1, s_1)=\left\{
\begin{array}{ll}
 \frac{\theta_2\left[ {\mu_\alpha} e^2  (1-s_1)-e^2 - 2 d_\text{e} (e-I_1)\right]+e^2}{{\mu_\alpha}e(e+ 2 I_1)}&\text{ if }{\mu_\alpha} (1-s_1)\geq 1;\\[1ex]
 0&\text{ otherwise}.\\
\end{array}\right.
\end{align}
Solving the system of the best-response functions Eqs.~\eqref{eq:simultaneous_3} through \eqref{eq:simultaneous_6}, we have that if ${\mu_\alpha}\geq \max \bigg\{\frac{3}{2}-\frac{d_\text{e}(2\theta_1-\theta_1\theta_2)}{ 2 e (2-\theta_1)},~\frac{3}{2}-\frac{d_\text{e}(2\theta_2-\theta_1\theta_2)}{ 2 e (2-\theta_2)}\bigg\}$  (note that this condition is satisfied by Assumption~\ref{assumption1}), there exists an equilibrium in which both investors invest $I_i^{TI}=\nicefrac{e}{2}$ with the share for Investor~$i$ as
\begin{align*}
s_i^{TI}=\frac{(3-2{\mu_\alpha})(2-\theta_i)}{{\mu_\alpha} (4-\theta_1 \theta_2)}+\frac{{\mu_\alpha}-1}{{\mu_\alpha}}-\frac{d_\text{e}(2\theta_i-\theta_1\theta_2)}{ e {\mu_\alpha}(4-\theta_1\theta_2)}.
\end{align*}
Similarly, we have that, if ${\mu_\alpha}<\frac{2-\theta_i}{1-\theta_i}-\frac{2 d_\text{e} \theta_i}{ e (1-\theta_i)}$, there exists an equilibrium in which Investor~$i$ is the only investor with the investment level $I_i^{TI}=\nicefrac{e}{2}$ in equilibrium and the share for Investor~$i$ is
\begin{align*}
s_i^{TI}=
\frac{\left({\mu_\alpha} -1\right)\theta_i+1}{{\mu_\alpha}}-\frac{2 d_\text{e} \theta_i}{ e {\mu_\alpha}}.
\end{align*}
\hfill $\blacksquare$

\subsubsection{The Two Large Investors Model  
}\label{sec:TIL_App}

Since the startup only needs $e$ units of capital at this stage, at most one investor can invest all the endowment. Therefore, there is an additional constraint that $I_1+I_2\leq e$ imposed on the bargaining problem. Then, the investments $\bm{I}$ and the shares $\bm{s}$ maximize the following Nash product simultaneously:
\begin{align}
\begin{split}
\max_{I_i\in[0, e],~ s_i\in[0, 1]}~&\left[\pi_i(\bm{I}, \bm{s})-d_i\right]^{\theta_i}\left[\pi_\text{e}(\bm{I}, \bm{s})-d_\text{e}^{-i}\right]^{1-\theta_i}\\
&\pi_i(\bm{I}, \bm{s})\geq d_i,~\pi_\text{e}(\bm{I}, \bm{s})\geq d_\text{e}^{-i}, ~i \in \{1,2\},\\
&I_1+I_2\leq e.
\end{split}
\end{align}

Solving the problem, we have the following proposition.
\begin{proposition}[Two large investors]~ \label{prop:sim_appx_rich_investor}
 Consider $(I_1^{TI-L}, I_2^{TI-L}, s_1^{TI-L}, s_2^{TI-L})$ such that  $I_1^{TI-L}+I_2^{TI-L}=e$, $I_i^{TI-L}\geq 0$, $i=1,2$, and
\begin{align}\label{eq:simultaneous_rich_investor}
&s_1^{TI-L}=\frac{I_1^{TI-L} \bigg( I_1^{TI-L}   e \theta_1  ( 1- \theta_2 + {\mu_\alpha}   \theta_2 )-e^2 \Big( \theta_1  \big( 2 - {\mu_\alpha}  (1 - \theta_2)- \theta_2 \big)-1\Big) -d_e\Big(\theta_1(1-\theta_2) e +I_1^{TI-L} \theta_1 \theta_2\Big)\bigg)}{ {\mu_\alpha} e ( e^2 - I_1^{TI-L}  I_2^{TI-L}  \theta_1   \theta_2  )},\nonumber\\[1ex]
&s_2^{TI-L}=\frac{I_2^{TI-L} \bigg( I_2^{TI-L}  e \theta_2  (1 - \theta_1 + {\mu_\alpha}   \theta_1 ) - e^2 \Big( \theta_2  \big( 2 - {\mu_\alpha}  (1 - \theta_1)- \theta_1 \big)-1\Big) -d_e\Big(\theta_2(1-\theta_1) e +I_2^{TI-L} \theta_1 \theta_2\Big) \bigg)}{ {\mu_\alpha} e  (  e^2 -I_1^{TI-L}  I_2^{TI-L}  \theta_1   \theta_2  )};
\end{align}
if ${\mu_\alpha}\geq \max_{i=1,2} \bigg\{\frac{e^3+e^2 I_i^{TI-L} - 2 e^2 I_i^{TI-L}\theta_i+ e(I_i^{TI-L})^2 \theta_i - d_e \big(e I_i^{TI-L}  (\theta_i-\theta_1\theta_2)+(I_i^{TI-L})^2 \theta_1\theta_2\big)}{e^3 - e^2 I_i^{TI-L} \theta_i}\bigg\}$, then $(I_1^{TI-L}, I_2^{TI-L}, s_1^{TI-L}, s_2^{TI-L})$ is an equilibrium bargaining outcome with Investor~$i$ investing $I_i^{TI-L}$ for the share of $s_i^{TI-L}$.
\begin{itemize}
\item If ${\mu_\alpha}<1+\frac{1}{1-\theta_1}-\frac{d_e \theta_1}{e(1-\theta_1)}$, there exists an equilibrium with the investment: $I_1^{TI-L}=e$ and  $I_2^{TI-L}=0$, and the share of investor~$i$: $s_1^{TI-L}=\frac{e+ e ({\mu_\alpha} -1)  \theta_1 -d_e \theta_1}{ e {\mu_\alpha}  }$ and  $s_2^{TI-L}=0$. 
\item If ${\mu_\alpha}<1+\frac{1}{1-\theta_2}-\frac{d_e \theta_2}{e(1-\theta_2)}$, there exists an equilibrium with the investment: $I_1^{TI-L}=0$ and  $I_2^{TI-L}=e$, and the share of investor~$i$: $s_1^{TI-L}=0$ and $s_2^{TI-L}=\frac{e+ e({\mu_\alpha}  -1)\theta_2 -d_e \theta_2 }{ e{\mu_\alpha}  }.$
\end{itemize}
\end{proposition}

\noindent{\bf Proof of Proposition~\ref{prop:sim_appx_rich_investor}.} Recall that the expected profit of the entrepreneur is
\begin{align*}
\pi_\text{e}(\bm{I}, \bm{s})= {\mu_\alpha} (I_1+I_2)  (1-s_1-s_2),
\end{align*}
and the expected profit of Investor~$i$ is
\begin{align*}
\pi_i(\bm{I}, \bm{s})= {\mu_\alpha} (I_1+I_2)  s_i+ e -I_i.
\end{align*}
The disagreement point of the entrepreneur when negotiating with Investor~$1$ is $$d_\text{e}^{-1}=\pi_\text{e}(0, I_2, 0, s_2)=\frac{d_\text{e}(e-I_2)}{e}+{\mu_\alpha} I_2  (1-s_2),$$ which is the sum of the prorated outside option and the profit of the entrepreneur when Investor~$2$ is the only investor. Similarly, the disagreement point of the entrepreneur when negotiating with Investor~$2$ is $$d_\text{e}^{-2}=\pi_\text{e}(I_1, 0, s_1, 0)=\frac{d_\text{e}(e-I_1)}{e}+{\mu_\alpha} I_1  (1-s_1).$$  The disagreement point of Investor~$i$ is $d_i=e$ since the investor has $e$ units of capital as the endowment.

We first solve the bargaining problem between the entrepreneur and Investor~1. 
Following the similar analysis as in the proof of Proposition~\ref{prop:single_appx}, we have that
\begin{align}\label{eq:simultaneous_2_rich_investor}
&\pi_1(\bm{I}, \bm{s})-d_1= \theta_1 \left( \pi_1(\bm{I}, \bm{s})+\pi_\text{e}(\bm{I}, \bm{s})-d_1-d_\text{e}^{-1}\right);\\
&\pi_\text{e}(\bm{I}, \bm{s})-d_\text{e}^{-1}= (1-\theta_1) \left( \pi_1(\bm{I}, \bm{s})+\pi_\text{e}(\bm{I}, \bm{s})-d_1-d_\text{e}^{-1}\right).\nonumber
\end{align}
Note that the best-response investment level
\begin{align}\label{eq:simultaneous_3_rich_investor}
&I_1(I_2, s_2)=\arg\max_{I_1\in[0, e], I_1+I_2\leq e} \left\{\pi_1(\bm{I}, \bm{s})+\pi_\text{e}(\bm{I}, \bm{s})-d_1-d_\text{e}^{-1}\right\}=\left\{
\begin{array}{ll}
 e - I_2&\text{ if } {\mu_\alpha} (1-s_2)\geq 1;\\[1ex]
 0&\text{ otherwise}.\\
\end{array}\right.
\end{align}
By Eq.~\eqref{eq:simultaneous_2_rich_investor}, the best-response share for Investor~$1$ is
\begin{align}\label{eq:simultaneous_4_rich_investor}
s_1(I_2, s_2)=\left\{
\begin{array}{ll}
\frac{\theta_1\left[{\mu_\alpha}  e (1-s_2)-e -d_e \right]+e}{{\mu_\alpha}e^2}(e-I_2).&\text{ if } {\mu_\alpha} (1-s_2)\geq 1;\\[1ex]
 0&\text{ otherwise}.\\
\end{array}\right.
\end{align}
Similarly, we have that the best-response investment level and share for Investor~$2$ are
\begin{align}\label{eq:simultaneous_5_rich_investor}
&I_2(I_1, s_1)=\left\{
\begin{array}{ll}
 e-I_1&\quad\quad\quad\quad\quad\quad\quad\text{ if } {\mu_\alpha} (1-s_1)\geq1;\\[1ex]
 0&\quad\quad\quad\quad\quad\quad\quad\text{ otherwise};\\[1ex]
\end{array}\right.\\\label{eq:simultaneous_6_rich_investor}
&s_2(I_1, s_1)=\left\{
\begin{array}{ll}
 \frac{\theta_2\left[ {\mu_\alpha} e   (1-s_1)-e - d_e \right]+e}{{\mu_\alpha} e^2 }(e-I_1)&\text{ if }{\mu_\alpha} (1-s_1)\geq 1;\\[1ex]
 0&\text{ otherwise}.\\
\end{array}\right.
\end{align}
Solving the system of the best-response functions Eqs.~\eqref{eq:simultaneous_3_rich_investor} through \eqref{eq:simultaneous_6_rich_investor}, we have that for $(I_1^{TI-L}, I_2^{TI-L}, s_1^{TI-L}, s_2^{TI-L})$ such that  $I_1^{TI-L}+I_2^{TI-L}=e$, $I_i^{TI-L}\geq 0$, $i=1,2$, and
\begin{align}\label{eq:simultaneous_rich_investor}
&s_1^{TI-L}=\frac{I_1^{TI-L} \bigg( I_1^{TI-L}   e \theta_1  ( 1- \theta_2 + {\mu_\alpha}   \theta_2 )-e^2 \Big( \theta_1  \big( 2 - {\mu_\alpha}  (1 - \theta_2)- \theta_2 \big)-1\Big) -d_e\Big(\theta_1(1-\theta_2) e +I_1^{TI-L} \theta_1 \theta_2\Big)\bigg)}{ {\mu_\alpha} e ( e^2 - I_1^{TI-L}  I_2^{TI-L}  \theta_1   \theta_2  )},\nonumber\\[1ex]
&s_2^{TI-L}=\frac{I_2^{TI-L} \bigg( I_2^{TI-L}  e \theta_2  (1 - \theta_1 + {\mu_\alpha}   \theta_1 ) - e^2 \Big( \theta_2  \big( 2 - {\mu_\alpha}  (1 - \theta_1)- \theta_1 \big)-1\Big) -d_e\Big(\theta_2(1-\theta_1) e +I_2^{TI-L} \theta_1 \theta_2\Big) \bigg)}{ {\mu_\alpha} e  (  e^2 -I_1^{TI-L}  I_2^{TI-L}  \theta_1   \theta_2  )};
\end{align}
if ${\mu_\alpha}\geq \max_{i=1,2} \bigg\{\frac{e^3+e^2 I_i^{TI-L} - 2 e^2 I_i^{TI-L}\theta_i+ e(I_i^{TI-L})^2 \theta_i - d_e \big(e I_i^{TI-L}  (\theta_i-\theta_1\theta_2)+(I_i^{TI-L})^2 \theta_1\theta_2\big)}{e^3 - e^2 I_i^{TI-L} \theta_i}\bigg\}$, then $(I_1^{TI-L}, I_2^{TI-L}, s_1^{TI-L}, s_2^{TI-L})$ is an equilibrium bargaining outcome with Investor~$i$ investing $I_i^{TI-L}$ for the share of $s_i^{TI-L}$.

If ${\mu_\alpha}<1+\frac{1}{1-\theta_1}-\frac{d_e \theta_1}{e(1-\theta_1)}$, there exists an equilibrium with the investment: $I_1^{TI-L}=e$ and  $I_2^{TI-L}=0$, and the share of investor~$i$: $s_1^{TI-L}=\frac{e+ e ({\mu_\alpha} -1)  \theta_1 -d_e \theta_1}{ e {\mu_\alpha}  }$ and  $s_2^{TI-L}=0$. 

If ${\mu_\alpha}<1+\frac{1}{1-\theta_2}-\frac{d_e \theta_2}{e(1-\theta_2)}$, there exists an equilibrium with the investment: $I_1^{TI-L}=0$ and  $I_2^{TI-L}=e$, and the share of investor~$i$: $s_1^{TI-L}=0$ and $s_2^{TI-L}=\frac{e+ e({\mu_\alpha}  -1)\theta_2 -d_e \theta_2 }{ e{\mu_\alpha}  }.$
\hfill $\blacksquare$

\medskip

\subsection{Preferred Stock Contracts}\label{sec:protection_appendix}

We consider the setting with Preferred Stock contracts and $\alpha_L=1$. In this case, Assumption~\ref{assumption1} reduces to $\alpha_H p +(1-p)\geq 2$. It follows that $\alpha_H\geq 1+
\frac{1}{p}>2$.

\subsubsection{The Single Investor Model \label{sec:single:investor_appx_protection}}

The investment $I_0$ and the share $s_0$ maximize the following Nash product with $\pi_\text{e}(I_0, s_0)=E\left[\min\{\alpha (1-s_0),\alpha-1\} \right] I_0$ and $\pi_0(I_0,  s_0)=E\left[\max\{\alpha s_0, 1\} \right]I_0+e- I_0$:
\begin{align}\label{eq:single_1_protection}
\max_{I_0\in[0, e],~ s_0\in[0, 1]}~&\left[\pi_0(I_0, s_0)-d_0\right]^{\theta_0}\left[\pi_\text{e}(I_0, s_0)-d_\text{e}\right]^{1-\theta_0}\\
&\pi_0(I_0, s_0)\geq d_0,~\pi_\text{e}(I_0, s_0)\geq d_\text{e}.\nonumber
\end{align}

The following proposition is Proposition~\ref{prop:single_protection} under general bargaining powers.
\begin{proposition}[Single investor]\label{prop:single_appx_protection}
The investor invests $\tilde{I}_0^{SI}=e$. The share of the investor is
\begin{align*}
&\tilde{s}_0^{SI}=\frac{\theta_0 (\alpha_H  - 1 ) + 1}{\alpha_H }-\frac{\theta_0 d_\text{e}}{e \alpha_H  p}.
\end{align*}
The corresponding entrepreneur's share is
\begin{align*}
\tilde{s}_\text{e}^{SI}=1-\tilde{s}_0^{SI}=\frac{(\alpha_H  - 1)(1-\theta_0)}{ \alpha_H }+\frac{\theta_0 d_\text{e}}{e \alpha_H  p}.
\end{align*}
\end{proposition}

\noindent \textbf{Proof of Proposition~\ref{prop:single_appx_protection}.} Recall that $d_0=e$. Also, we have that $\alpha_H> 2$. We focus on the scenario where $s_0\geq 1/\alpha_H$. Otherwise, the investor does not have incentives to invest. Thus, the expected profit of the entrepreneur should be
\begin{align*}
\pi_\text{e}(I_0, s_0)&=E\left[\min\{\alpha (1-s_0),\alpha-1\} \right] I_0 =\alpha_H   (1-s_0) I_0 p,
\end{align*}
and the expected profit of investor~$s$ should be
\begin{align*}
\pi_0(I_0,  s_0)&=E\left[\max\{\alpha s_0, 1\} \right]I_0+e- I_0 =\alpha_H   s_0 I_0 p + e -I_0 p .
\end{align*}
Solving the problem \eqref{eq:single_1_protection} above, we have that,
\begin{align}\label{eq:single_2_protection}
&\pi_0(I_0, s_0)-d_0=\theta_0\left(\pi_\text{e}(I_0, s_0)+\pi_0(I_0, s_0)-d_\text{e}-d_0 \right); \\
&\pi_\text{e}(I_0, s_0)-d_\text{e}=(1-\theta_0)\left(\pi_\text{e}(I_0, s_0)+\pi_0(I_0, s_0)-d_\text{e}-d_0 \right).\nonumber
\end{align}
Recall that  $\alpha_H >2$, and we have that
\begin{align*}
&\tilde{I}_0^{SI}=\arg\max_{I_0\in[0, 2e]} \left\{\pi_\text{e}(I_0, s_0)+\pi_0(I_0, s_0)-d_\text{e}-d_0\right\}= e.
\end{align*}
By Eq.~\eqref{eq:single_2_protection}, we have that
\begin{align*}
\tilde{s}_0^{SI}=
\frac{\theta_0 (\alpha_H  - 1 ) + 1}{\alpha_H }-\frac{\theta_0 d_\text{e}}{e \alpha_H  p}.
\end{align*}
\hfill $\blacksquare$

\subsubsection{The Two Investors Model \label{sec:simultaneous_appx_protection}}
In this case, both investors simultaneously negotiate with the entrepreneur. The bargaining outcome is a pair of the share $s_i$ for Investor~$i$ in return for the investment $I_i$. With the downside protection for the investors, when the return of the startup is realized as $\alpha_L=1$, the investors are able to recover their investment. That is, in addition to the negotiated $\alpha_L s_i$, Investor~$i$ is able to recover his potential loss $\alpha_L(1-s_i)$ from the entrepreneur. Thus, both investors obtain their investment back and the entrepreneur earns zero. When the return of the startup is realized as $\alpha_H$, the protection for the investors is invoked only if one investor negotiated for a share that is significantly low. In such an event, the investor who invoked the protection will first be compensated by the profit of the entrepreneur, and then by the profit of the other investor (if the other investor invests as well).

Therefore, the expected profit of the entrepreneur is
\begin{align}\label{eq:profit-ent_protection}
\pi_\text{e}(\bm{I}, \bm{s})
&=\alpha_H(I_1+I_2)p-\sum_{i=1}^2\max\bigg\{\alpha_H (I_1+I_2)s_i-\bigg(I_{3-i}-\alpha_H (I_1+I_2)(1-s_i)\bigg)^+, I_i\bigg\} p ,
\end{align}
and the expected profit of Investor~$i$ is
\begin{align}\label{eq:profit_inv_protection}
\pi_i(\bm{I}, \bm{s})
&=\max\bigg\{\alpha_H (I_1+I_2)s_i-\bigg(I_{3-i}-\alpha_H (I_1+I_2)(1-s_i)\bigg)^+, I_i\bigg\} p + I_i(1-p) + \frac{e}{2} -I_i\nonumber\\
&=\max\bigg\{\alpha_H (I_1+I_2)s_i-\bigg(I_{3-i}-\alpha_H (I_1+I_2)(1-s_i)\bigg)^+, I_i\bigg\} p + \frac{e}{2} -I_i p.
\end{align}
The disagreement point of the entrepreneur when negotiating with Investor~$1$ is
$$d_\text{e}^{-1}=\pi_\text{e}(0, I_2, 0, s_2)=\frac{d_\text{e}(e-I_2)}{e}+\alpha_H I_2 p - \max\bigg\{\alpha_H I_2  s_2, I_2\bigg\} p,$$
which is the sum of the prorated outside option and the profit of the entrepreneur when Investor~$2$ is the only investor. Similarly, the disagreement point of the entrepreneur when negotiating with Investor~$2$ is
$$d_\text{e}^{-2}=\pi_\text{e}(I_1, 0, s_1, 0)=\frac{d_\text{e}(e-I_1)}{e}+\alpha_H I_1 p - \max\bigg\{ \alpha_H I_1  s_1, I_1\bigg\} p.$$
The disagreement point of Investor~$i$ is $d_i=\nicefrac{e}{2}$ since the investor has $\nicefrac{e}{2}$ units of capital as the endowment.

Then, the investments $\bm{I}$ and the shares $\bm{s}$ maximize the following Nash product simultaneously:
\begin{align}\label{eq:simultaneous_1_protection}
\max_{I_i\in[0, e/2],~ s_i\in[0, 1]}~&\left[\pi_i(\bm{I}, \bm{s})-d_i\right]^{\theta_i}\left[\pi_\text{e}(\bm{I}, \bm{s})-d_\text{e}^{-i}\right]^{1-\theta_i}\\
&\pi_i(\bm{I}, \bm{s})\geq d_i,~\pi_\text{e}(\bm{I}, \bm{s})\geq d_\text{e}^{-i}.\nonumber
\end{align}

We first establish a lemma, which helps us simplify the profits in Eqs.~\eqref{eq:profit-ent_protection} and \eqref{eq:profit_inv_protection}.
\medskip

\begin{lemma}\label{lemma:no-low-share} In any equilibrium bargaining outcome $(\bm{I}, \bm{s})$, the following conditions are satisfied:
\begin{align*}
\alpha_H (I_1+I_2)(1-s_i)\geq I_{3-i},~i=1,2.
\end{align*}
\end{lemma}
\noindent{\bf Proof of Lemma~\ref{lemma:no-low-share}.} It is easy to observe that the result holds if only one investor invests in equilibrium. The following proof focuses on the case where both investors invest in equilibrium.

 We first observe that $\alpha_H (I_1+I_2)(1-s_i)< I_{3-i}$ when the entrepreneur's profit is not enough to cover the compensation to protect investor~$(3-i)$. Since $\alpha_H>1$ and $s_1+s_2<1$, we note that $\alpha_H (I_1+I_2)(1-s_i)< I_{3-i}$ can hold for at most one bargaining unit. We next prove the lemma by showing that in any bargaining outcome $(\bm{I}, \bm{s})=(I_1, I_2, s_1, s_2)$, if $\alpha_H (I_1+I_2)(1-s_1)< I_{2}$ and  $\alpha_H (I_1+I_2)(1-s_2)\geq  I_{1}$, then $(\bm{I}, \bm{s})$ is not feasible for the bargaining problem between the entrepreneur and Investor~$1$.

Since $\alpha_H (I_1+I_2)(1-s_1)< I_{2}$ and  $\alpha_H (I_1+I_2)(1-s_2)\geq  I_{1}$, by Eqs.~\eqref{eq:profit-ent_protection}, we have
\begin{align*}
\pi_\text{e}(\bm{I}, \bm{s})
=&\alpha_H(I_1+I_2)p- \max\bigg\{\alpha_H (I_1+I_2)-I_2, I_1\bigg\} p-\max\bigg\{\alpha_H (I_1+I_2)s_2, I_2\bigg\} p  \\
=&\alpha_H(I_1+I_2)p- \bigg(\alpha_H (I_1+I_2)-I_2 \bigg) p - I_2 p  \\
=&0,
\end{align*}
Note that the disagreement point of the entrepreneur is that
$$d_\text{e}^{-1}=\pi_\text{e}(0, I_2, 0, s_2)=\frac{d_\text{e} (e-I_2)}{e}+\alpha_H I_2 p - \max\bigg\{\alpha_H I_2  s_2, I_2\bigg\} p=\frac{d_\text{e}(e-I_2)}{e}+(\alpha_H-1)I_2 p>0.$$
It follows that $\pi_\text{e}(\bm{I}, \bm{s})<d_\text{e}^{-1}$ and therefore, $(\bm{I}, \bm{s})$ is not feasible for the bargaining problem between the entrepreneur and Investor~$1$. \hfill $\blacksquare$

\medskip

By Lemma~\ref{lemma:no-low-share}, we can further simplify the expected profit of the entrepreneur as
\begin{align}\label{eq:profit-ent_protection_new}
\pi_\text{e}(\bm{I}, \bm{s})
&=\alpha_H(I_1+I_2)p-\sum_{i=1}^2\max\bigg\{\alpha_H (I_1+I_2)s_i, I_i\bigg\} p ,
\end{align}
and the expected profit of Investor~$i$ as
\begin{align}\label{eq:profit_inv_protection_new}
\pi_i(\bm{I}, \bm{s})
&=\max\bigg\{\alpha_H (I_1+I_2)s_i, I_i\bigg\} p + \frac{e}{2} -I_i p.
\end{align}

The following proposition is Proposition~\ref{prop:sim_protection} under general bargaining powers.
\begin{proposition}[Two investors]   \label{prop:sim_appx_protection} There are two types of equilibrium bargaining outcomes: both investors investing and only one investor investing. 
\begin{itemize}
\item In the both-investor-investing equilibrium, the investors invest the endowed capital; i.e., $\tilde{I}_i^{TI}=\nicefrac{e}{2}$ for $i\in\{1,2\}$, and the equilibrium share of investor $i$ is as follows.
\begin{itemize}
\item If $\alpha_H\leq \min\left\{\frac{1-2 \theta_1\theta_2 +\theta_1}{\theta_1-\theta_1\theta_2}+\frac{d_\text{e}}{e p}, \frac{1-2 \theta_1\theta_2 +\theta_2}{\theta_2-\theta_1\theta_2}+\frac{d_\text{e}}{e p}\right\}$,
\begin{align}\label{eq:sim_protection_1}
&\tilde{s}_1^{TI}=\frac{1-\alpha_H \theta_1 \theta_2+\alpha_H \theta_1-\theta_1}{2 \alpha_H (1-\theta_1 \theta_2)}-\frac{d_\text{e}(\theta_1-\theta_1\theta_2)}{2 \alpha_H p e (1-\theta_1 \theta_2)},\\
&\tilde{s}_2^{TI}=\frac{1-\alpha_H \theta_1 \theta_2+\alpha_H \theta_2-\theta_2}{2 \alpha_H (1-\theta_1 \theta_2)}-\frac{d_\text{e}(\theta_2-\theta_1\theta_2)}{2 \alpha_H p e (1-\theta_1 \theta_2)}
\end{align}
\item If $\frac{1-2 \theta_1\theta_2+\theta_1}{\theta_1-\theta_1\theta_2}+\frac{d_\text{e}}{ e p}<\alpha_H \leq \frac{2 +3 \theta_2 - 2 \theta_1 \theta_2}{2\theta_2-\theta_1\theta_2}+\frac{d_\text{e}}{e p}$
\begin{align}\label{eq:sim_protection_2}
&\tilde{s}_1^{TI}=\frac{1-\alpha_H \theta_1 \theta_2+\alpha_H \theta_1+\theta_1 \theta_2-\theta_1}{\alpha_H (2 - \theta_1 \theta_2)}-\frac{d_\text{e} (\theta_1-\theta_1\theta_2)}{ \alpha_H p e (2- \theta_1\theta_2)},\\
&\tilde{s}_2^{TI}=\frac{2-\alpha_H \theta_1 \theta_2+2 \alpha_H \theta_2-3 \theta_2}{2 \alpha_H (2 - \theta_1 \theta_2)}-\frac{d_\text{e}(2\theta_2-\theta_1\theta_2)}{2 \alpha_H p e (2-\theta_1\theta_2)}
\end{align}
\item If  $\frac{1-2 \theta_1\theta_2+\theta_2}{\theta_2-\theta_1\theta_2} +\frac{d_\text{e}}{ e p}<\alpha_H \leq \frac{2 +3 \theta_1 - 2 \theta_1 \theta_2}{2 \theta_1-\theta_1\theta_2}+\frac{d_\text{e}}{ e p}$
\begin{align}\label{eq:sim_protection_3}
&\tilde{s}_1^{TI}=\frac{2-\alpha_H \theta_1 \theta_2+2 \alpha_H \theta_1-3 \theta_1}{2 \alpha_H (2 - \theta_1 \theta_2)}-\frac{d_\text{e}(2\theta_1-\theta_1\theta_2)}{2 \alpha_H p e (2-\theta_1\theta_2)},\\
&\tilde{s}_2^{TI}=\frac{1-\alpha_H \theta_1 \theta_2+\alpha_H \theta_2+\theta_1 \theta_2-\theta_2}{\alpha_H (2 - \theta_1 \theta_2)}-\frac{d_\text{e} (\theta_2-\theta_1\theta_2)}{ \alpha_H p e (2- \theta_1\theta_2)}
\end{align}
\item If $\alpha_H>\max\left\{\frac{2+3\theta_1-2\theta_1\theta_2}{2\theta_1-\theta_1\theta_2}+\frac{d_\text{e}}{ e p}, \frac{2+3\theta_2-2\theta_1\theta_2}{2\theta_2-\theta_1\theta_2}+\frac{d_\text{e}}{ e p}\right\}$
\begin{align}\label{eq:sim_protection_4}
&\tilde{s}_1^{TI}=\frac{2-\alpha_H \theta_1 \theta_2+2 \alpha_H \theta_1+\theta_1 \theta_2-3 \theta_1}{\alpha_H (4 - \theta_1 \theta_2)}-\frac{d_\text{e} (2 \theta_1 - \theta_1 \theta_2)}{ \alpha_H p e (4-\theta_1\theta_2)},\\
&\tilde{s}_2^{TI}=\frac{2-\alpha_H \theta_1 \theta_2+2 \alpha_H \theta_2+\theta_1 \theta_2-3 \theta_2}{\alpha_H ( 4 - \theta_1 \theta_2)}-\frac{d_\text{e} (2 \theta_2 - \theta_1 \theta_2)}{ \alpha_H p e (4-\theta_1\theta_2)}
\end{align}
\end{itemize}

\item If $\alpha_H<\frac{2-\theta_i}{1-\theta_i}-\frac{2 d_\text{e} \theta_i}{ e p (1-\theta_i)}$, an equilibrium bargaining outcome in which only Investor~$i$ invests exists; i.e., $\tilde{I}_i^{TI}=\nicefrac{e}{2}$ and $\tilde{I}_j^{TI}=0$ for $i,j\in\{1,2\}$ and $i\neq j$, and the equilibrium share of Investor~$i$ is
 \begin{align*}
&\tilde{s}_i^{TI}=\frac{\left(\alpha_H -1\right)\theta_i+1}{\alpha_H}-\frac{2 \theta_i d_\text{e}}{ e \alpha_H  p}.
\end{align*}

\end{itemize}
\end{proposition}

\noindent{\bf Proof of Proposition~\ref{prop:sim_appx_protection}.}
We solve the bargaining problem between the entrepreneur and Investor~$1$ as specified in \eqref{eq:simultaneous_1_protection} for the best-response investment level and the share of the investor. We first note that if $s_1< \frac{I_1}{\alpha_H(I_1+I_2)}$, it follows that $\pi_1(\bm{I}, \bm{s})=\nicefrac{e}{2}$ and the Nash product for the bargaining between the entrepreneur and Investor~$1$ in problem~\eqref{eq:simultaneous_1_protection} is zero. In the following analysis, we restrict attention to the case where Investor~$1$'s share $s_1\geq  \frac{I_1}{\alpha_H(I_1+I_2)}$ and later verify that the equilibrium bargaining outcome leads to a strictly positive Nash product. In this case, by the first order condition, we have that

\begin{align}\label{eq:simultaneous_2_protection}
&\pi_1(\bm{I}, \bm{s})-d_1= \theta_1 \left( \pi_1(\bm{I}, \bm{s})+\pi_\text{e}(\bm{I}, \bm{s})-d_1-d_\text{e}^{-1}\right);\\
&\pi_\text{e}(\bm{I}, \bm{s})-d_\text{e}^{-1}= (1-\theta_1) \left( \pi_1(\bm{I}, \bm{s})+\pi_\text{e}(\bm{I}, \bm{s})-d_1-d_\text{e}^{-1}\right).\nonumber
\end{align}
Note that the best-response investment level
\begin{align}\label{eq:simultaneous_3_protection}
\tilde{I}_1(I_2, s_2)&=\arg\max_{I_1\in[0, e/2]} \left\{\pi_1(\bm{I}, \bm{s})+\pi_\text{e}(\bm{I}, \bm{s})-d_1-d_\text{e}^{-1}\right\}
=\left\{
\begin{array}{ll}
 \nicefrac{e}{2}&\text{ if }  s_2 \leq \frac{\alpha_H-1}{\alpha_H} ;\\[1ex]%
 0 &\text{ if } s_2 > \frac{\alpha_H-1}{\alpha_H}.\\ 
\end{array}\right.
\end{align}

By Eq.~\eqref{eq:simultaneous_2_protection}, the best-response share for Investor~$1$ is

  \begin{align}\label{eq:simultaneous_4_protection_prime}
\tilde{s}_1(I_2, s_2)=\left\{
\begin{array}{ll}
\frac{\theta_1\left[\alpha_H (e+2 I_2)  (1-s_2)- e - 2 \alpha_H I_2  + 2 I_2 \right]+ e }{\alpha_H(e + 2 I_2)}-\frac{2 d_\text{e} (e-I_2) \theta_1}{ e \alpha_H p (e + 2 I_2)}&\text{ if } s_2 \leq \frac{1}{\alpha_H  } ;\\[1ex]
\frac{\theta_1\left[\alpha_H e  (1-s_2)- e  \right] + e }{\alpha_H( e + 2 I_2)}-\frac{2 d_\text{e} (e-I_2) \theta_1}{ e \alpha_H p (e + 2 I_2)}&\text{ if } \frac{1}{\alpha_H  } < s_2\leq \frac{\alpha_H-1}{\alpha_H} ;\\[1ex]
 0&\text{ if } s_2 > \frac{\alpha_H-1}{\alpha_H}.\\
\end{array}\right.
\end{align}

Similarly, we have that the best-response investment level and share for Investor~$2$ are
\begin{align}\label{eq:simultaneous_5_protection}
&\tilde{I}_2(I_1, s_1)=\left\{
\begin{array}{ll}
 \nicefrac{e}{2}&\text{ if }  s_1 \leq \frac{\alpha_H-1}{\alpha_H};\\[1ex]%
 0&\text{ if } s_1 > \frac{\alpha_H-1}{\alpha_H}.\\
\end{array}\right.\\[1ex]
\label{eq:simultaneous_6_protection}
&\tilde{s}_2(I_1, s_1)=\left\{
\begin{array}{ll}
\frac{\theta_2\left[\alpha_H (2 I_1+e)  (1-s_1) - e  - 2 \alpha_H I_1  + 2 I_1 \right] + e }{\alpha_H(2 I_1+e)}-\frac{2 d_\text{e} (e-I_1) \theta_2}{ e \alpha_H p (e + 2 I_1)}&\text{ if }s_1 \leq \frac{1}{\alpha_H  } ;\\[1ex]
\frac{\theta_2\left[\alpha_H e  (1-s_1) - e  \right] + e }{\alpha_H(2 I_1+ e)}-\frac{2 d_\text{e} (e-I_1) \theta_2}{ e \alpha_H p (e + 2I_1)}&\text{ if } \frac{1}{\alpha_H  } < s_1\leq \frac{\alpha_H-1}{\alpha_H} ;\\[1ex]
 0&\text{ if } s_1 > \frac{\alpha_H-1}{\alpha_H} .\\
\end{array}\right.
\end{align}

Solving the system of the best-response functions Eqs.~\eqref{eq:simultaneous_3_protection} through \eqref{eq:simultaneous_6_protection}, we have that there exists an equilibrium in which both investors invest $\tilde{I}_i^{TI}=\nicefrac{e}{2}$ with the share for Investor~$i$ as follows.
\begin{itemize}
\item If $\alpha_H\leq \min\left\{\frac{1-2 \theta_1\theta_2 +\theta_1}{\theta_1-\theta_1\theta_2}+\frac{d_\text{e}}{e p}, \frac{1-2 \theta_1\theta_2 +\theta_2}{\theta_2-\theta_1\theta_2}+\frac{d_\text{e}}{e p}\right\}$,
\begin{align*}
&\tilde{s}_1^{TI}=\frac{1-\alpha_H \theta_1 \theta_2+\alpha_H \theta_1-\theta_1}{2 \alpha_H (1-\theta_1 \theta_2)}-\frac{d_\text{e}(\theta_1-\theta_1\theta_2)}{2 \alpha_H p e (1-\theta_1 \theta_2)},\\
&\tilde{s}_2^{TI}=\frac{1-\alpha_H \theta_1 \theta_2+\alpha_H \theta_2-\theta_2}{2 \alpha_H (1-\theta_1 \theta_2)}-\frac{d_\text{e}(\theta_2-\theta_1\theta_2)}{2 \alpha_H p e (1-\theta_1 \theta_2)}
\end{align*}
\item If $\frac{1-2 \theta_1\theta_2+\theta_1}{\theta_1-\theta_1\theta_2}+\frac{d_\text{e}}{ e p}<\alpha_H \leq \frac{2 +3 \theta_2 - 2 \theta_1 \theta_2}{2\theta_2-\theta_1\theta_2}+\frac{d_\text{e}}{ e p}$
\begin{align*}
&\tilde{s}_1^{TI}=\frac{1-\alpha_H \theta_1 \theta_2+\alpha_H \theta_1+\theta_1 \theta_2-\theta_1}{\alpha_H (2 - \theta_1 \theta_2)}-\frac{d_\text{e} (\theta_1-\theta_1\theta_2)}{ \alpha_H p e (2- \theta_1\theta_2)},\\
&\tilde{s}_2^{TI}=\frac{2-\alpha_H \theta_1 \theta_2+2 \alpha_H \theta_2-3 \theta_2}{2 \alpha_H (2 - \theta_1 \theta_2)}-\frac{d_\text{e}(2\theta_2-\theta_1\theta_2)}{2 \alpha_H p e (2-\theta_1\theta_2)}
\end{align*}
\item If  $\frac{1-2 \theta_1\theta_2+\theta_2}{\theta_2-\theta_1\theta_2} +\frac{d_\text{e}}{ e p}<\alpha_H \leq \frac{2 +3 \theta_1 - 2 \theta_1 \theta_2}{2 \theta_1-\theta_1\theta_2}+\frac{d_\text{e}}{ e p}$
\begin{align*}
&\tilde{s}_1^{TI}=\frac{2-\alpha_H \theta_1 \theta_2+2 \alpha_H \theta_1-3 \theta_1}{2 \alpha_H (2 - \theta_1 \theta_2)}-\frac{d_\text{e}(2\theta_1-\theta_1\theta_2)}{2 \alpha_H p e (2-\theta_1\theta_2)},\\
&\tilde{s}_2^{TI}=\frac{1-\alpha_H \theta_1 \theta_2+\alpha_H \theta_2+\theta_1 \theta_2-\theta_2}{\alpha_H (2 - \theta_1 \theta_2)}-\frac{d_\text{e} (\theta_2-\theta_1\theta_2)}{ \alpha_H p e (2- \theta_1\theta_2)}
\end{align*}
\item If $\alpha_H>\max\left\{\frac{2+3\theta_1-2\theta_1\theta_2}{2\theta_1-\theta_1\theta_2}+\frac{d_\text{e}}{ e p}, \frac{2+3\theta_2-2\theta_1\theta_2}{2\theta_2-\theta_1\theta_2}+\frac{d_\text{e}}{ e p}\right\}$
\begin{align*}
&\tilde{s}_1^{TI}=\frac{2-\alpha_H \theta_1 \theta_2+2 \alpha_H \theta_1+\theta_1 \theta_2-3 \theta_1}{\alpha_H (4 - \theta_1 \theta_2)}-\frac{d_\text{e} (2 \theta_1 - \theta_1 \theta_2)}{ \alpha_H p e (4-\theta_1\theta_2)},\\
&\tilde{s}_2^{TI}=\frac{2-\alpha_H \theta_1 \theta_2+2 \alpha_H \theta_2+\theta_1 \theta_2-3 \theta_2}{\alpha_H ( 4 - \theta_1 \theta_2)}-\frac{d_\text{e} (2 \theta_2 - \theta_1 \theta_2)}{ \alpha_H p e (4-\theta_1\theta_2)}
\end{align*}
\end{itemize}
It is easy to verify that the equilibrium shares satisfies that $\tilde{s}_i^{TI}\geq \frac{I_i^{TI}}{\alpha_H(I_1^{TI}+I_2^{TI})}=\frac{1}{2\alpha_H}$.

Similarly, we have that if $\alpha_H<\frac{2-\theta_i}{1-\theta_i}-\frac{2 d_\text{e} \theta_i}{e p (1-\theta_i)}$, there exists an equilibrium in which Investor~$i$ is the only investor with the investment level $\tilde{I}_i^{TI}=\nicefrac{e}{2}$ in equilibrium and the share for Investor~$i$ is
\begin{align*}
\tilde{s}_i^{TI}=
\frac{\left(\alpha_H -1\right)\theta_i+1}{\alpha_H}-\frac{2 \theta_i d_\text{e}}{ e \alpha_H  p}.
\end{align*}
\hfill $\blacksquare$

\medskip

\subsection{Theoretical Predictions Given Experimental Parameterization \label{sec:theory:experiment}}
Table \ref{tab:theory:predictions} summarizes the theory predictions under the experimental parameters used in the experiment. All calculations are based on Propositions 1-4 in Section  \ref{sec:theory}. The parameters are in Section  \ref{sec:parameters}.

\begin{table}[H]
\centering
\caption{Theoretical Predictions of Entrepreneur's Share (in \%) Given Experimental Parameterization}
\label{tab:theory:predictions}
\small
\begin{tabular}{@{}lccp{0.1in}cc@{}}
\toprule
          & \multicolumn{2}{c}{\textbf{PoorEnt}} & \multicolumn{1}{l}{} & \multicolumn{2}{c}{\textbf{RichEnt}} \\ \cmidrule{2-3}\cmidrule{5-6}
Condition & \textbf{Common} & \textbf{Preferred} & \textbf{}            & \textbf{Common} & \textbf{Preferred} \\ \midrule
SI & 33.3 & 45.5 &  & 46.7 & 63.6 \\
TI & 46.7 & 56.4 &  & 57.3 & 70.9 \\ \bottomrule
\end{tabular}
\end{table}

\subsection{Model Robustness: Risk Aversion \label{sec:risk:aversion}}
For analytical tractability, the models above were solved under the assumption that all parties were risk neutral. However, it is natural to wonder how the results hold up, especially Corollary \ref{corollary: ent_share}, if the parties involved are risk averse. Unfortunately, the model becomes analytically intractable to solve. We are able to show that, for the parameters that we implement in the experiment, so long as risk aversion is not too great, there will still be equilibria in which both investors choose to invest and that the entrepreneur's ranking from Corollary \ref{corollary: ent_share} still holds. The following illustrates an example for the comparison of the entrepreneur's share under SI and TI when $d_\text{e}=0$. Specifically, let $u_i = x^{1-\rho_i}$ denote player $i$'s utility function, where $\rho_i = 0$ indicates risk neutrality and $\rho_i > 0$ indicates risk aversion. In the experiment, as we outlined in Section ~\ref{sec:theory}, we assume that $e = 200$ and $(\alpha_H,\alpha_L,p) = (11, 1 ,0.2)$. Table \ref{tab:ent:share:risk} gives the entrepreneur's share under various assumptions on risk preferences.

\begin{table}[htbp]
\small
  \centering\caption{\label{tab:ent:share:risk}The Entrepreneur's Share Under Risk Aversion}
\subfloat[Common Stock contracts]{
  \begin{tabular}{ccc} \hline
  Risk Parameters & SI-PoorEnt (\%) &   TI-PoorEnt (\%) \\ \hline
  $\rho_\text{e} = \rho_s = \rho_1 = \rho_2 = 0$          & 33.33 &   46.67 \\
  $\rho_\text{e} = 0$; $\rho_s = \rho_1 = \rho_2 = 0.25$  & 31.23 &   44.48 \\
$\rho_\text{e} = 0.25$; $\rho_s = \rho_1 = \rho_2 = 0$    & 28.57 &   46.44 \\
$\rho_\text{e} = 0.25$; $\rho_s = \rho_1 = \rho_2 = 0.25$ & 26.98 &   44.32 \\ \hline
  \end{tabular}
  }

  \subfloat[Preferred Stock contracts]{
  \begin{tabular}{cccc} \hline
  Risk Parameters & SI-PoorEnt (\%)   & TI-PoorEnt (\%) \\ \hline
  $\rho_\text{e} = \rho_s = \rho_1 = \rho_2 = 0$          & 45.46   & 56.36 \\
  $\rho_\text{e} = 0$; $\rho_s = \rho_1 = \rho_2 = 0.25$  & 49.15   & 58.45 \\
$\rho_\text{e} = 0.25$; $\rho_s = \rho_1 = \rho_2 = 0$    & 38.96   & 55.78 \\
$\rho_\text{e} = 0.25$; $\rho_s = \rho_1 = \rho_2 = 0.25$ & 42.80   & 57.83 \\ \hline
  \end{tabular}
  }
\end{table}

As can be seen, in all cases, the entrepreneur earns the least when bargaining against a single investor and the most when bargaining with two investors simultaneously. Note that entrepreneur risk aversion is detrimental to their share, but the effects are largest in the single investor case where the entrepreneur's bargaining power is weakest. It is also interesting to note that investor risk aversion is also detrimental to the entrepreneur under the Common Stock contracts but beneficial to the entrepreneur under the Preferred Stock contract. Under the Common Stock contracts, by investing in the business, the investor is putting money at risk and, therefore, requires compensation for that risk. Moreover, disagreement would also be a better outcome compared to successfully negotiating and having the business be a failure. \citet{RR1982} showed that increased risk aversion could, counterintuitively increase a player's share when disagreement is not the worst outcome. It seems that a similar result holds here. Under the Preferred Stock contracts, the investor's downside is protected and effectively the bargaining is regarding the state when the startup value is realized as $\alpha_H$. In this case, the entrepreneur is able to take advantage of the risk aversion of the investors and gain a higher share when bargaining with a more risk-averse investor.

\smallskip

\section{Theory with Alternative Belief\label{RB_sec:appx_alternative}}
We discuss the models under alternative belief about the disagreement points in this appendix.

\subsection{The Two Investors Model Under Common Stock Contract \label{RB_sec:simultaneous_appx_alternative}}

In this section, we present the results and the proofs with the general bargaining powers of the investor(s) relative to the entrepreneur under the alternative disagreement point of the entrepreneur set as $d_\text{e}^{-i}=d_\text{e}$ under the Common Stock contract. All other model settings are the same as in Section ~\ref{sec:theory}. We note that the alternative specification of the entrepreneur's disagreement point does not affect the result in the single investor bargaining problem. It only affects the two investor bargaining problem.

The investments $I_i$ and the share $s_i$ maximize the following Nash product simultaneously:
\begin{align}\label{RB_eq:simultaneous_1_alternative}
\max_{I_i\in[0, e],~ s_i\in[0, 1]}~&\left[\pi_i(\bm{I}, \bm{s})-d_i\right]^{\theta_i}\left[\pi_\text{e}(\bm{I}, \bm{s})-d_\text{e}^{-i}\right]^{1-\theta_i}\\
&\pi_i(\bm{I}, \bm{s})\geq d_i,~\pi_\text{e}(\bm{I}, \bm{s})\geq d_\text{e}^{-i}.\nonumber
\end{align}

\begin{proposition}[Two investors under alternative belief] ~ \label{RB_prop:sim_appx_alternative}
\begin{itemize}
\item When ${\mu_\alpha}\geq \max \bigg\{\frac{3-2 \theta_1 - \theta_1\theta_2}{2(1-\theta_1)}-\frac{d_\text{e} (\theta_1-\theta_1\theta_2)}{ e (1-\theta_1)},~\frac{3-2\theta_2-\theta_1\theta_2}{2(1-\theta_2)}-\frac{d_\text{e}(\theta_2-\theta_1\theta_2)}{ e (1-\theta_2)}\bigg\}$, there exists an equilibrium bargaining outcome in which both investors invest with $I_i^{TI}=\nicefrac{e}{2}$ for $i\in\{1,2\}$, and the equilibrium share of investor $i$ is
\begin{align}\label{RB_eq:simultaneous_7_alternative}
&s_i^{TI}=\frac{1-2 \theta_i +2 {\mu_\alpha} \theta_i +(1-2 {\mu_\alpha}) \theta_1\theta_2 }{2 {\mu_\alpha} (1-\theta_1\theta_2)}-\frac{d_\text{e}(\theta_i-\theta_1\theta_2)}{ e {\mu_\alpha}(1-\theta_1\theta_2)}.
\end{align}
\item There exists an equilibrium bargaining outcome in which only Investor~$i$ invests when ${\mu_\alpha}<\frac{2-\theta_i}{1-\theta_i}-\frac{2 d_\text{e} \theta_i}{ e (1-\theta_i)}$. The equilibrium investment level $I_i^{TI}=e$ and $I_j^{TI}=0$ for $i,j\in\{1,2\}$ and $i\neq j$, and the equilibrium share of Investor~$i$ is
 \begin{align*}
&s_i^{TI}=\frac{({\mu_\alpha}-1)\theta_i+1}{{\mu_\alpha}}-\frac{2 d_\text{e} \theta_i}{ e {\mu_\alpha}}.
\end{align*}

\end{itemize}
\end{proposition}

\noindent{\bf Proof of Proposition~\ref{RB_prop:sim_appx_alternative}.} Recall that the expected profit of the entrepreneur is
\begin{align}\label{RB_eq:profit-ent_alternative}
\pi_\text{e}(\bm{I}, \bm{s})= {\mu_\alpha} (I_1+I_2)  (1-s_1-s_2),
\end{align}
and the expected profit of Investor~$i$ is
\begin{align}\label{RB_eq:profit_inv_alternative}
\pi_i(\bm{I}, \bm{s})= {\mu_\alpha} (I_1+I_2)  s_i+ \frac{e}{2} -I_i.
\end{align}
The disagreement point of the entrepreneur when negotiating with Investor~$i$ is $d_\text{e}^{-i}=d_\text{e},$ which the value of the outside option when the entrepreneur does not reach agreement with either investor.  The disagreement point of Investor~$i$ is $d_i=\nicefrac{e}{2}$ since the investor has $\nicefrac{e}{2}$ units of capital as the endowment.

We first solve the bargaining problem between the entrepreneur and Investor~$1$ as specified in \eqref{RB_eq:simultaneous_1_alternative}. Following the similar analysis as in the proof of Proposition~\ref{prop:single_appx}, we have that
\begin{align}\label{RB_eq:simultaneous_2_alternative}
&\pi_1(\bm{I}, \bm{s})-d_1= \theta_1 \left( \pi_1(\bm{I}, \bm{s})+\pi_\text{e}(\bm{I}, \bm{s})-d_1-d_\text{e}^{-1}\right);\\
&\pi_\text{e}(\bm{I}, \bm{s})-d_\text{e}^{-1}= (1-\theta_1) \left( \pi_1(\bm{I}, \bm{s})+\pi_\text{e}(\bm{I}, \bm{s})-d_1-d_\text{e}^{-1}\right).\nonumber
\end{align}
Note that the best-response investment level
\begin{align}\label{RB_eq:simultaneous_3_alternative}
&I_1(I_2, s_2)=\arg\max_{I_1\in[0, e]} \left\{\pi_1(\bm{I}, \bm{s})+\pi_\text{e}(\bm{I}, \bm{s})-d_1-d_\text{e}^{-1}\right\}=\left\{
\begin{array}{ll}
 \nicefrac{e}{2}&\text{ if } {\mu_\alpha} (1-s_2)\geq 1;\\[1ex]
 0&\text{ otherwise}.\\
\end{array}\right.
\end{align}
By Eq.~\eqref{RB_eq:simultaneous_2_alternative}, the best-response share for Investor~$1$ is
\begin{align}\label{RB_eq:simultaneous_4_alternative}
s_1(I_2, s_2)=\left\{
\begin{array}{ll}
\frac{\theta_1\left[{\mu_\alpha} (e+2 I_2)  (1-s_2)-e - 2 d_\text{e}\right]+e}{{\mu_\alpha} (e+ 2 I_2)} &\text{ if } {\mu_\alpha} (1-s_2)\geq 1;\\[1ex]
 0&\text{ otherwise}.\\
\end{array}\right.
\end{align}

Similarly, we have that the best-response investment level and share for Investor~$2$ are
\begin{align}\label{RB_eq:simultaneous_5_alternative}
&I_2(I_1, s_1)=\left\{
\begin{array}{ll}
 e&\quad\quad\quad\quad\quad\quad\quad\text{ if } {\mu_\alpha} (1-s_1)\geq1;\\[1ex]
 0&\quad\quad\quad\quad\quad\quad\quad\text{ otherwise};\\[1ex]
\end{array}\right.\\\label{RB_eq:simultaneous_6_alternative}
&s_2(I_1, s_1)=\left\{
\begin{array}{ll}
 \frac{\theta_2\left[ {\mu_\alpha} (2 I_1+ e)  (1-s_1)-e - 2 d_\text{e} \right]+e}{{\mu_\alpha} (2 I_1+ e)}&\text{ if }{\mu_\alpha} (1-s_1)\geq 1;\\[1ex]
 0&\text{ otherwise}.\\
\end{array}\right.
\end{align}
Solving the system of the best-response functions Eqs.~\eqref{RB_eq:simultaneous_3_alternative} through \eqref{RB_eq:simultaneous_6_alternative}, we have that if ${\mu_\alpha}\geq \max \bigg\{\frac{3-2 \theta_1 - \theta_1\theta_2}{2(1-\theta_1)}-\frac{d_\text{e} (\theta_1-\theta_1\theta_2)}{  e (1-\theta_1)},~\frac{3-2\theta_2-\theta_1\theta_2}{2(1-\theta_2)}-\frac{d_\text{e}(\theta_2-\theta_1\theta_2)}{  e (1-\theta_2)}\bigg\}$, there exists an equilibrium in which both investors invest $I_i^{TI}=e$ with the share for Investor~$i$ as
\begin{align*}
s_i^{TI}=\frac{1-2 \theta_i +2 {\mu_\alpha} \theta_i +(1-2 {\mu_\alpha}) \theta_1\theta_2 }{2 {\mu_\alpha} (1-\theta_1\theta_2)}-\frac{d_\text{e}(\theta_i-\theta_1\theta_2)}{ e {\mu_\alpha}(1-\theta_1\theta_2)}.
\end{align*}
Similarly, we have that, if ${\mu_\alpha}<\frac{2-\theta_i}{1-\theta_i}-\frac{2 d_\text{e} \theta_i}{e (1-\theta_i)}$, there exists an equilibrium in which Investor~$i$ is the only investor with the investment level $I_i^{TI}=\nicefrac{e}{2}$ in equilibrium and the share for Investor~$i$ is
\begin{align*}
s_i^{TI}=
\frac{\left({\mu_\alpha} -1\right)\theta_i+1}{{\mu_\alpha}}-\frac{2 d_\text{e} \theta_i}{e {\mu_\alpha}}.
\end{align*}
\hfill $\blacksquare$

\medskip

\subsection{The Two Investors Model Under Preferred Stock Contract \label{RB_sec:simultaneous_appx_protection_alternative}}

In this section, we present the results and the proofs with the general bargaining powers of the investor(s) relative to the entrepreneur under the alternative disagreement point of the entrepreneur set as $d_\text{e}^{-i}=d_\text{e}$ under the Preferred Stock contract. All other model settings are the same as in Section~\ref{sec:preferred}. We note that the alternative specification of the entrepreneur's disagreement point does not affect the result in the single investor bargaining problem. It only affects the two investor bargaining problem.

Similar to the analysis before, we can show that the equilibrium bargaining share will not be too extreme such that in any equilibrium bargaining outcome $(\bm{I}, \bm{s})$, the following conditions are satisfied:
$\alpha_H (I_1+I_2)(1-s_i)\geq I_{3-i},~i=1,2.$ Therefore, the expected profit of the entrepreneur is
\begin{align}\label{RB_eq:profit-ent_protection_new_alternative}
\pi_\text{e}(\bm{I}, \bm{s})
&=\alpha_H(I_1+I_2)p-\sum_{i=1}^2\max\bigg\{\alpha_H (I_1+I_2)s_i, I_i\bigg\} p ,
\end{align}
and the expected profit of Investor~$i$ as
\begin{align}\label{RB_eq:profit_inv_protection_new_alternative}
\pi_i(\bm{I}, \bm{s})
&=\max\bigg\{\alpha_H (I_1+I_2)s_i, I_i\bigg\} p + e -I_i p.
\end{align}
The (alternative) disagreement point of the entrepreneur when negotiating with Investor~$i$ is
$d_\text{e}^{-i}=d_\text{e}$, and the disagreement point of Investor~$i$ is $d_i=\nicefrac{e}{2}$ since the investor has $\nicefrac{e}{2}$ units of capital as the endowment.

Then, the investments $\bm{I}$ and the shares $\bm{s}$ maximize the following Nash product simultaneously:
\begin{align}\label{RB_eq:simultaneous_1_protection_alternative}
\max_{I_i\in[0, e],~ s_i\in[0, 1]}~&\left[\pi_i(\bm{I}, \bm{s})-d_i\right]^{\theta_i}\left[\pi_\text{e}(\bm{I}, \bm{s})-d_\text{e}^{-i}\right]^{1-\theta_i}\\
&\pi_i(\bm{I}, \bm{s})\geq d_i,~\pi_\text{e}(\bm{I}, \bm{s})\geq d_\text{e}^{-i}.\nonumber
\end{align}

\begin{proposition}[Two investors under alternative belief]~   \label{RB_prop:sim_appx_protection_alternative}
\begin{itemize}
\item There exists an equilibrium bargaining outcome in which both investors invest; i.e., $\tilde{I}_i^{TI}=\nicefrac{e}{2}$ for $i\in\{1,2\}$, and the equilibrium share of investor $i$ is as follows.
\begin{itemize}
\item If $\alpha_H\geq  \max\left\{\frac{3}{2}+\frac{(1-\theta_1)\theta_2}{2(1-\theta_2)}-\frac{d_\text{e}(1-\theta_1)\theta_2}{ e p(1-\theta_2)}, \frac{3}{2}+\frac{(1-\theta_2)\theta_1}{2(1-\theta_1)}-\frac{d_\text{e}(1-\theta_2)\theta_1}{ e p(1-\theta_1)}\right\}$,
\begin{align}\label{RB_eq:sim_protection_1_alternative}
&\tilde{s}_1^{TI}=\frac{  \theta_1  (2  \alpha_H  (1- \theta_2 )+ \theta_2 -2)+1}{2  \alpha_H   (1- \theta_1   \theta_2 )}-\frac{d_\text{e}  \theta_1  (1- \theta_2 )}{  \alpha_H  e p (1- \theta_1   \theta_2 )},\\
&\tilde{s}_2^{TI}=\frac{  \theta_2  (2  \alpha_H  (1-  \theta_1 )+ \theta_1 -2)+1}{2  \alpha_H   (1- \theta_1   \theta_2 )}-\frac{d_\text{e} (1- \theta_1 )  \theta_2 }{  \alpha_H  e p (1- \theta_1   \theta_2)}
\end{align}
\end{itemize}

\item If $\alpha_H<\frac{2-\theta_i}{1-\theta_i}-\frac{2 d_\text{e} \theta_i}{e p (1-\theta_i)}$, an equilibrium bargaining outcome in which only Investor~$i$ invests exists; i.e., $\tilde{I}_i^{TI}=\nicefrac{e}{2}$ and $\tilde{I}_j^{TI}=0$ for $i,j\in\{1,2\}$ and $i\neq j$, and the equilibrium share of Investor~$i$ is
 \begin{align*}
&\tilde{s}_i^{TI}=\frac{\left(\alpha_H -1\right)\theta_i+1}{\alpha_H}-\frac{2\theta_i d_\text{e}}{e \alpha_H  p}.
\end{align*}

\end{itemize}
\end{proposition}

\noindent{\bf Proof of Proposition~\ref{RB_prop:sim_appx_protection_alternative}.}
We solve the bargaining problem between the entrepreneur and Investor~$1$ as specified in \eqref{RB_eq:simultaneous_1_protection_alternative} for the best-response investment level and the share of the investor. We first note that if $s_1< \frac{I_1}{\alpha_H(I_1+I_2)}$, it follows that $\pi_1(\bm{I}, \bm{s})=\nicefrac{e}{2}$ and the Nash product for the bargaining between the entrepreneur and Investor~$1$ in problem~\eqref{RB_eq:simultaneous_1_protection_alternative} is zero. In the following analysis, we restrict attention to the case where Investor~$1$'s share $s_1\geq  \frac{I_1}{\alpha_H(I_1+I_2)}$ and later verify that the equilibrium bargaining outcome leads to a strictly positive Nash product. In this case, by the first order condition, we have that
\begin{align}\label{RB_eq:simultaneous_2_protection_alternative}
&\pi_1(\bm{I}, \bm{s})-d_1= \theta_1 \left( \pi_1(\bm{I}, \bm{s})+\pi_\text{e}(\bm{I}, \bm{s})-d_1-d_\text{e}^{-1}\right);\\
&\pi_\text{e}(\bm{I}, \bm{s})-d_\text{e}^{-1}= (1-\theta_1) \left( \pi_1(\bm{I}, \bm{s})+\pi_\text{e}(\bm{I}, \bm{s})-d_1-d_\text{e}^{-1}\right).\nonumber
\end{align}
Note that the best-response investment level
\begin{align}\label{RB_eq:simultaneous_3_protection_alternative}
\tilde{I}_1(I_2, s_2)&=\arg\max_{I_1\in[0, e]} \left\{\pi_1(\bm{I}, \bm{s})+\pi_\text{e}(\bm{I}, \bm{s})-d_1-d_\text{e}^{-1}\right\}
=\left\{
\begin{array}{ll}
 \nicefrac{e}{2}&\text{ if }  s_2 \leq \frac{\alpha_H-1}{\alpha_H} ;\\[1ex]%
 0 &\text{ if } s_2 > \frac{\alpha_H-1}{\alpha_H}.\\ 
\end{array}\right.
\end{align}

By Eq.~\eqref{RB_eq:simultaneous_2_protection_alternative}, the best-response share for Investor~$1$ is
  \begin{align}\label{RB_eq:simultaneous_4_protection_prime_alternative}
\tilde{s}_1(I_2, s_2)=\left\{
\begin{array}{ll}
\frac{\theta_1\left[\alpha_H (e + 2 I_2)  (1-s_2)-e  \right]+e }{\alpha_H(e+2 I_2)}-\frac{2 d_\text{e} \theta_1}{ \alpha_H p (e+ 2 I_2)}&\text{ if } s_2\leq \frac{\alpha_H-1}{\alpha_H} ;\\[1ex]
 0&\text{ if } s_2 > \frac{\alpha_H-1}{\alpha_H}.\\
\end{array}\right.
\end{align}


Similarly, we have that the best-response investment level and share for Investor~$2$ are
\begin{align}\label{RB_eq:simultaneous_5_protection_alternative}
&\tilde{I}_2(I_1, s_1)=\left\{
\begin{array}{ll}
 \nicefrac{e}{2}&\text{ if }  s_1 \leq \frac{\alpha_H-1}{\alpha_H};\\[1ex]%
 0&\text{ if } s_1 > \frac{\alpha_H-1}{\alpha_H}.\\
\end{array}\right.\\[1ex]
\label{RB_eq:simultaneous_6_protection_alternative}
&\tilde{s}_2(I_1, s_1)=\left\{
\begin{array}{ll}
\frac{\theta_2\left[\alpha_H (2 I_1+e)  (1-s_1)-e  \right]+e }{\alpha_H(2 I_1+e)}-\frac{2 d_\text{e} \theta_2}{\alpha_H p (e+2 I_1)}&\text{ if }  s_1\leq \frac{\alpha_H-1}{\alpha_H} ;\\[1ex]
 0&\text{ if } s_1 > \frac{\alpha_H-1}{\alpha_H} .\\
\end{array}\right.
\end{align}


Solving the system of the best-response functions Eqs.~\eqref{RB_eq:simultaneous_3_protection_alternative} through \eqref{RB_eq:simultaneous_6_protection_alternative}, we have that there exists an equilibrium in which both investors invest $\tilde{I}_i^{TI}=\nicefrac{e}{2}$ with the share for Investor~$i$ as follows.
\begin{itemize}
\item If $\alpha_H\geq  \max\left\{\frac{3}{2}+\frac{(1-\theta_1)\theta_2}{2(1-\theta_2)}-\frac{d_\text{e}(1-\theta_1)\theta_2}{ e p(1-\theta_2)}, \frac{3}{2}+\frac{(1-\theta_2)\theta_1}{2(1-\theta_1)}-\frac{d_\text{e}(1-\theta_2)\theta_1}{ e p(1-\theta_1)}\right\}$,
\begin{align*}
&\tilde{s}_1^{TI}=\frac{  \theta_1  (2  \alpha_H  (1- \theta_2 )+ \theta_2 -2)+1}{2  \alpha_H   (1- \theta_1   \theta_2 )}-\frac{d_\text{e}  \theta_1  (1- \theta_2 )}{ \alpha_H  e p (1- \theta_1   \theta_2 )},\\
&\tilde{s}_2^{TI}=\frac{  \theta_2  (2  \alpha_H  (1-  \theta_1 )+ \theta_1 -2)+1}{2  \alpha_H   (1- \theta_1   \theta_2 )}-\frac{d_\text{e} (1- \theta_1 )  \theta_2 }{  \alpha_H  e p (1- \theta_1   \theta_2)}
\end{align*}
\end{itemize}
It is easy to verify that the equilibrium shares satisfies that $\tilde{s}_i^{TI}\geq \frac{I_i^{TI}}{\alpha_H(I_1^{TI}+I_2^{TI})}=\frac{1}{2\alpha_H}$.

Similarly, we have that if $\alpha_H<\frac{2-\theta_i}{1-\theta_i}-\frac{2 d_\text{e} \theta_i}{e p (1-\theta_i)}$, there exists an equilibrium in which Investor~$i$ is the only investor with the investment level $\tilde{I}_i^{TI}=\nicefrac{e}{2}$ in equilibrium and the share for Investor~$i$ is
\begin{align*}
\tilde{s}_i^{TI}=
\frac{\left(\alpha_H -1\right)\theta_i+1}{\alpha_H}-\frac{2 \theta_i d_\text{e}}{ e \alpha_H  p}.
\end{align*}
\hfill $\blacksquare$

\section{Experimental Protocol and Instructions \label{sec:instructions}}
\subsection{Experimental Protocol\label{sec:instructions:protocol}}
In this appendix we provide details on our protocol for running experiments online as well as the instructions. As noted in the main text, we adapted the procedures suggested by \citet{Zoom} and \citet{Zoom-Mich}. Specifically, the steps from recruiting to payment were as follows:
\begin{enumerate}
\item Participants were recruited from the general subject population using the University's recruiting platform (SONA). Subjects were explicitly told that the study would be online and that they would be emailed a Zoom link 1-2 hours before the scheduled time. When participants connected to the Zoom meeting, they were held in a waiting room until they could be checked in.
\item 15 minutes before the start of the session, one experimenter began to check in participants one at a time, checking their ID and changing their name to ``User $x$'', where $x$ is a number. After check-in, the participant's video was turned off and the participant was placed in a breakout room that was also monitored by another experimenter.
\item Once all participants were checked-in, general instructions were provided to all participants. Specifically, they were told that they would be placed in one of two (SI) or three (TI) breakout rooms, where they would be asked to turn on their video feed.\endnote{Requiring subjects to display their video is mentioned by \citet{Zoom-Mich} as an important factor in ensuring that participants remain attentive throughout the experiment.} Participants were informed that they would never interact with another participant in the same breakout room. That is, all entrepreneurs were placed in the same breakout room, and similarly for those in the role of Investor~$1$ and Investor~$2$. Each breakout room was also monitored by an experimenter.
\item Subjects then read the instructions specific to the experiment and answered the comprehension questions. The experimenter in each breakout room was there to handle questions. If necessary, temporarily moving a participant to a private breakout room to answer questions or provide assistance.
\item Subjects participated in the main experiment in their respective breakout room.
\item At the conclusion of the experiment, all breakout rooms were closed and participants were asked to complete a separate survey with their name and address so that payments could be processed. This was done in order to de-link decisions from the experiment and personally identifying information. After participants completed the survey, they were free to leave the Zoom meeting.
\item Consistent with IRB guidelines, subjects were then mailed (within one business day) a debit card with the amount earned in the experiment.
\end{enumerate}

\subsection{Experimental Instructions\label{sec:instructions:instructions}}
Below we reproduce the instructions for the SI-PoorEnt Treatment. The instructions for the remaining treatments and Preferred Stock contracts are analogous. The reproduced text omits the quiz questions as well as the additional measurements (risk aversion and fairness norm elicitation). The entrepreneur's view of instructions is shown (The investor's view is analogous). Indentation and fonts have been adapted for clarity of exposition.
\medskip

\begin{framed}
\small
 \noindent This study is about startup ownership. There are two parties: the  entrepreneur and the investor, who must together decide how to divide the ownership of the startup. You will participate in 10 rounds of this study. In each round you will be the entrepreneur. In each round you will be matched at random with another person participating in this session, and that person will be the investor. Each round will consist of an interactive negotiation exercise between you and the investor. In each round you will have an opportunity to earn "points". At the end of the study one of the 10 rounds will be selected at random. Then, your point earnings from that round will be converted to US Dollars at the rate of 2 cents per point, and added to your participation payment of 8 US Dollar.

\medskip

\noindent The investor has 200 points that she/he can invest in your startup.  However, the share of the startup that  the investor will receive in exchange for her/his investment is not known. Rather, you and the investor will negotiate the share that the investor receives. Then, the following can happen:
\begin{itemize}[
  align=left,
  leftmargin=2em,
  itemindent=0pt,
  labelsep=0pt,
  labelwidth=2em
]
\item[-] Negotiations succeed. If negotiations succeed, there are two possible outcomes:
\begin{itemize}
\item[-] Startup succeeds: If the startup succeeds, it will be worth $200 \times 11  = 2200$ points, and that value will be divided between  you (the entrepreneur) and the investor depending on the outcome of the negotiations.
\item[-] Startup fails: If the startup fails, the value of the investment is multiplied by 1. This means, if the startup fails, it will be worth $200 \times 1 =200$ points, and that value will be divided between you (the entrepreneur) and the investor depending on the outcome of the negotiations.
\end{itemize}
\item[-] Negotiations fail. If the negotiations fail, the investor gets to keep her/his 200 points and you (the entrepreneur) receive zero.
\end{itemize}
\medskip

\noindent Note: the investor cannot invest partial amounts (any amount less than 200 points). In other words, either all 200 points are invested or nothing is invested.
\end{framed}

\begin{framed}
\noindent As mentioned on the previous screen, it is possible that  the startup fails. In particular, if the investor and the entrepreneur come to an agreement, there is an 80\% chance that the startup fails and a 20\% chance that it succeeds.  If the startup fails, its value is equal to 200 points.  If the startup succeeds, the
value of the investment is multiplied by 11 as explained on the previous screen. That is, the startup will be worth $200 \times 11 = 2200$ points. Once the startup value is known, it will be divided between the investor and the entrepreneur according to the outcome of the negotiations.
\medskip

\noindent For example, suppose you have accepted an offer that gives the other participant (the investor) 65 percent of the startup and gives you (the entrepreneur) 35 percent of the startup. Then, if the startup succeeds, the other participant (the investor) will receive $2200 \times 0.65 =1430$ points, and you will receive the remainder; i.e.,  770 points. In contrast, if the startup fails, the investor will receive $200 \times 0.65 = 130$ points and you will receive 70 points.
\medskip

\noindent Alternatively, suppose you made an offer that gives the other participant (the investor) 25 percent of the startup and gives you (the entrepreneur) 75 percent of the startup, and that offer has been accepted. Then, if the startup succeeds, the other participant (the investor) will receive $2200 \times 0.25 = 550$ points, and you will receive the remainder, 1650 points. In contrast, if the startup fails, the investor will receive $200 \times 0.25 = 50$ points and you will receive 150 points.

\end{framed}

\noindent {\normalfont [...Subjects complete quiz questions...]}

\begin{framed}
\small
 On the next screen, you will have 90 seconds to reach an agreement between you (the entrepreneur) and the other participant (the investor). Specifically, you will make and receive offers regarding the percentage amount that the investor receives in exchange for her/his investment of 200 points. Both you and the other participant can make as many offers as you wish in the  90 seconds available for negotiations. If you wish, you can also accept the most recent  offer made by the other participant (the investor). If time expires without an offer being accepted, then the round ends and no investment is made. In this case, the investor keeps her/his 200 units and you (the entrepreneur) receive zero.
\end{framed}

\begin{figure}[h]
  \centering\caption{\label{fig:screen:SI} Negotiation Interface: SI Treatment}
  \includegraphics[width=0.9\textwidth]{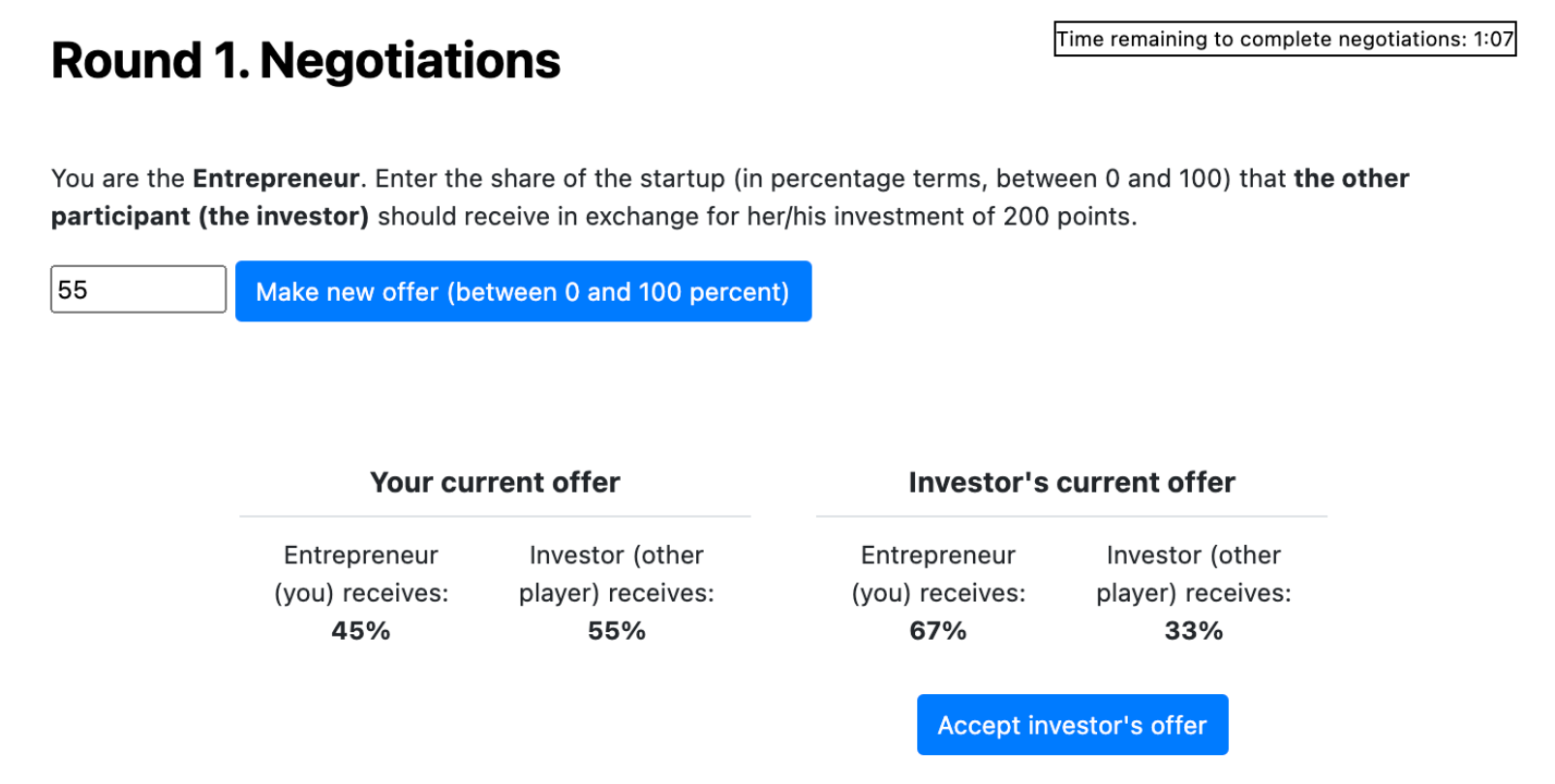}

\end{figure}

\begin{figure}[h]
  \centering\caption{\label{fig:screen:TI:S} Negotiation Interface: TI Treatment}
  \includegraphics[width=0.9\textwidth]{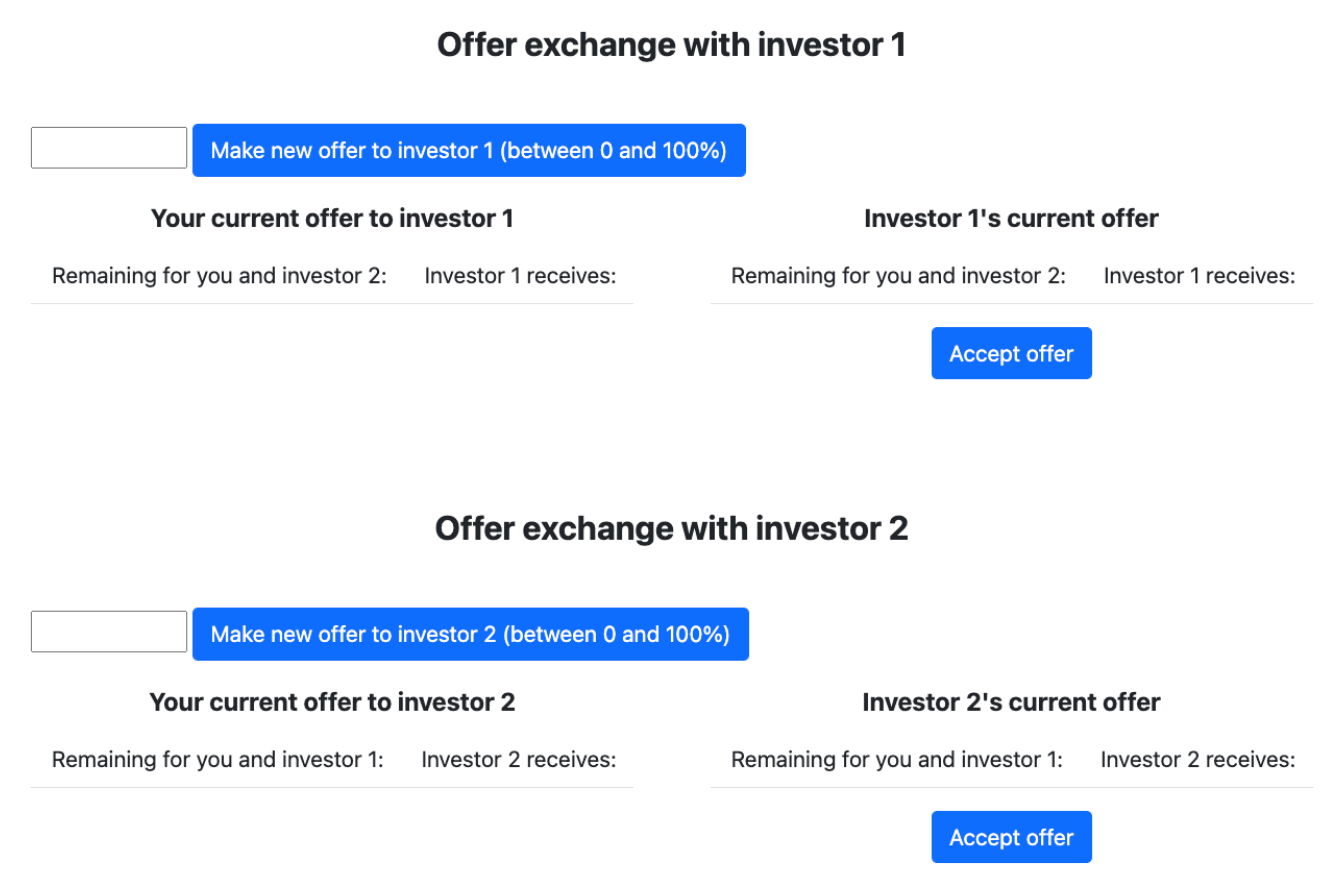}

\end{figure}

\section{Insights From the Bargaining Process \label{sec:bargaining:appendix}}
In this section, we delve into the bargaining process in order to better understand the drivers of our results in \textsection\ref{sec:num:inv}-\ref{sec:preferred}. We begin by reproducing the standard result in the bargaining literature that first offers typically have an anchoring effect on negotiated outcomes. Building on this result, we then show that there are noticeable treatment differences in opening offers, which, combined with anchoring, go a long way to explaining our key findings.

\subsection{Anchoring \label{sec:process:anchoring}}
Table \ref{tab:first:offers:anchoring} reports the results of a series of random effects regressions analyzing how first offers (of both investors and entrepreneurs) affect the likelihood of successfully reaching an agreement and also on the investors' shares conditional on an agreement. Here we focus on SI and TI treatment conditions, in which bargaining is one-dimensional (negotiators only negotiate equity percentage, and not the investment size). Note that because we focus on bilateral negotiations between individual investor $i$ and the entrepreneur we use investor $i$ share (as opposed to entrepreneur share) as our measurement.

Table \ref{tab:first:offers:anchoring} suggests strong anchoring effects for both agreements and shares. Higher opening offers to investors by entrepreneurs are significantly positively associated with agreement, while higher opening demands by investors are negatively (and, for SI, significantly) associated with agreements. That is, the more generous the entrepreneur, the more likely is an agreement, while the more demanding the investor, the less likely is an agreement. Turning now to the agreed investor shares, we see that both the offers made to and the demands made by investors are significantly positively associated with the investor's final share. These regression results show that players generally face a trade-off in bargaining. The more generous they are to their bargaining partner, the more likely are they to reach an agreement, but the less favorable the agreement will be to themselves.

\begin{table}[bth]
\centering \footnotesize
\makebox[\linewidth][c]{%
\begin{minipage}{\textwidth}
\caption{Anchoring Effects of Initial Offers on Agreements and Investor Shares}
\label{tab:first:offers:anchoring}
\begin{tabular}{ld{3.6}@{}ld{3.6}@{}ld{3.6}@{}ld{3.6}@{}l} \toprule
     & \multicolumn{4}{c}{SI}                                           & \multicolumn{4}{c}{TI}   \\ \cmidrule{2-9}
   \multicolumn{1}{r}{Dep. Var.:}  & \multicolumn{2}{C{2.4cm}}{Agreement Reached} & \multicolumn{2}{C{2.4cm}}{Investor $i$'s Share} & \multicolumn{2}{C{2.4cm}}{Agreement Reached} & \multicolumn{2}{C{2.4cm}}{Investor $i$'s Share} \\ \midrule
Initial Offer to Investor $i$ by Entrepreneur             & 0.005\sym{***}           & (0.001)          & 0.291\sym{***}           & (0.047)         & 0.007\sym{***}           & (0.002)          & 0.399\sym{***}           & (0.057)         \\
Initial Demand by Investor $i$            & -0.006\sym{***}          & (0.001)          & 0.384\sym{***}           & (0.051)         & -0.001             & (0.001)          & 0.129\sym{***}           & (0.046)         \\
Constant                 & 1.076\sym{***}           & (0.051)          & 17.894\sym{***}          & (3.785)         & 0.767\sym{***}           & (0.106)          & 17.889\sym{***}          & (3.340)         \\ \midrule
$R^2$                    & \dor{0.138}              &                  & \dor{0.421}              &                 & \dor{0.048}              &                  & \dor{0.390}              &                 \\
$N$                        & \dor{1342}               &                  & \dor{1104}               &                 & \dor{838}                &                  & \dor{724}                &                \\ \bottomrule
\end{tabular}
\vspace{0.1cm}

\begin{minipage}{1.00\textwidth}
  \scriptsize Note: Random effects regression coefficients are reported (Standard errors in parentheses). All regressions include controls for treatment variables and standard errors are corrected for clustering at the session level. The number of observations ($N$) corresponds to the number of interactions in which both negotiators (entrepreneur and investor $i$) made at least one offer. * p$<$0.10, ** p$<$0.05, *** p$<$0.01.
\end{minipage}
\end{minipage}
}
\end{table}

\renewcommand{\theresult}{B.\arabic{result}}
\setcounter{result}{0}

\begin{result}
There are strong anchoring effects: first offers predict both the likelihood of agreement and the final share received conditional on an agreement. Negotiators face a trade-off. A more aggressive opening offer increases the chance of disagreement, but, conditional on an agreement, leads to a more favorable outcome.
\end{result}

\subsection{Bargaining Process: Initial Offers in PoorEnt and RichEnt\label{sec:process:initial}}
We next examine the initial offers made by investors and entrepreneurs in PoorEnt and RichEnt. In PoorEnt offers were one-dimensional and consisted of the share offered to the investor in exchange for an investment of 200 (in SI) or 100 (in TI).  In Table \ref{tab:ti:alt:first:offers} we report summary statistics on the proposed share to the investor.  

\begin{table}[htbp]
\centering
\caption{Opening Offers: Proposed Share to Investor for First Offer}
\label{tab:ti:alt:first:offers}
\footnotesize
\begin{tabular}{lcccC{2.5cm}cccC{2.5cm}} \toprule
\multirow{3}{*}{\parbox{0.75in}{Treatment condition}} & \multirow{3}{*}{\parbox{0.75in}{\centering Proposed Investment}} & \multicolumn{3}{c}{Entrepreneur's Proposal} & & \multicolumn{3}{c}{Investor $i$'s Proposal} \\ \cmidrule{3-5} \cmidrule{7-9}
& & PoorEnt & RichEnt & PoorEnt vs. RichEnt $p$-value & & PoorEnt & RichEnt & PoorEnt vs. RichEnt $p$-value \\\cmidrule{1-5} \cmidrule{7-9}
SI                                               & 200    & 36.97   & 40.30 &   0.348       & & 70.77   & 65.13 &   0.313     \\
TI                                             & 100    & 23.78   & 23.31 &    0.807      & & 44.29   & 44.54 &   0.989       \\  \bottomrule
\end{tabular}

\vspace{0.05in}
\begin{minipage}{0.96\textwidth}
  \scriptsize Note: $p-$values are obtained from $t-$tests on the subject average of first offers for each investment amount.
\end{minipage}

\end{table}

On average, entrepreneurs offer investors  36.97\% of the pie in exchange for an investment of 200 in SI and 23.78\% in exchange for an investment of 100 TI; thus, while the funding amount is halved, the share offered to each investor drops by less than half as we go from SI to TI. A similar pattern emerges if we look at investors' first offers. Further, the final share obtained by investor $i$ is 46.70\% in SI PoorEnt and 32.60\% in TI PoorEnt (51.65\% in SI RichEnt and 32.82\% in TI RichEnt, see Tables \ref{tab:poor:ent:summary}-\ref{tab:rich:ent:summary}), which is close to the midpoint between the entrepreneur's and investor's opening offers, suggesting that both parties make substantial concessions to get to an agreement. These comparisons suggest that the main results in \textsection\ref{sec:shares:results} are well explained by the investors' and entrepreneurs' initial negotiation strategies (opening offers), and not by the differences in bargaining dynamics.

\begin{result}
Opening offers of both parties largely explain the differences in final outcomes in SI and TI. Poor entrepreneurs prefer larger investments, while rich entrepreneurs prefer smaller investments. Investors prefer smaller investments. Investors bargain more aggressively when the entrepreneur is rich.
\end{result}

\subsection{Differences Across Contracts\label{sec:process:contracts}}
We conclude this section by examining the bargaining dynamics that cause the differences between Common vs. Preferred Stock contracts. Table \ref{tab:first:offers} shows summary statistics on first offers for each role and for each of our treatment conditions in the SI/TI treatment arm. There are several interesting observations. First, consistent with theory, entrepreneurs' first offers to investors were generally lower under Preferred Stock contracts than under Common Stock contracts ($p<0.05$ in three out of four comparisons). The only scenario in which entrepreneurs do not increase their demands with Preferred (relative to Common) Stock contracts is SI-PoorEnt. This is consistent with Result 3 in \textsection \ref{sec:preferred}, which also reports that Preferred Stock contracts disadvantage entrepreneurs in the PoorEnt treatment arms, with the effect being particularly stark in the SI condition.

\begin{table}[b]
\centering \footnotesize
\caption{First Offers (Share to Investor) by Bargaining Environment and Player Role}
\label{tab:first:offers}
\subfloat[SI]{
\begin{tabular}{ccccp{0.1in}ccc} \toprule
          & \multicolumn{3}{c}{Entrepreneur's First Offer} && \multicolumn{3}{c}{Investor $i$'s First Offer}   \\ \cmidrule{2-4} \cmidrule{6-8}
          & Common  & Preferred & $p-$value  && Common & Preferred & $p-$value \\ \midrule
PoorEnt   & 36.97   & 39.70     & 0.303      && 70.77  & 71.60     & 0.858     \\
RichEnt   & 40.30   & 34.25     & $\ll 0.01$ && 65.13  & 66.71     & 0.500     \\
\bottomrule
\end{tabular}
}
\vspace{0.0cm}
\subfloat[TI]{
\begin{tabular}{ccccp{0.1in}ccc} \toprule
          & \multicolumn{3}{c}{Entrepreneur's First Offer} && \multicolumn{3}{c}{Investor $i$'s First Offer}   \\ \cmidrule{2-4} \cmidrule{6-8}
          & Common  & Preferred & $p-$value  && Common & Preferred & $p-$value \\ \midrule
PoorEnt   & 23.78   & 20.48     & $\ll 0.01$ && 44.29  & 49.59     & 0.004      \\
RichEnt   & 23.31   & 20.30     & 0.031      && 44.54  & 49.38     & $\ll 0.01$ \\
\bottomrule
\end{tabular}}
\vspace{-0.2cm}
\begin{minipage}{0.63\textwidth}
  \scriptsize Note: $p-$values are derived from random effects regressions of first offers on treatment indicators, with standard errors corrected for clustering at the session level.
\end{minipage}
\end{table}

Second, investors' opening demands were either approximately the same (SI) or significantly higher (TI) for Preferred relative to Common Stock contracts. In particular, in TI, investors increased their demands from an average of 44.29\% to 49.59\% in PoorEnt ($p=0.004$) and from 44.54 to 49.38\% ($p\ll0.01$) in RichEnt. That is, despite their reduced theoretical bargaining power with Preferred Stock contracts, investors adopted more extreme opening positions in the two investor (TI) setting. While we cannot provide definitive evidence for the driver of this behavior, it is consistent with peer-induced fairness \citep{ho2009,ho2014} -- each investor's desire to outcompete the other investor, resulting in more aggressive bargaining positions. Conversely, the absence of a peer investor in the SI scenarios may lead to less aggressive investor behaviors.

\begin{result}
Insufficient reflection of contract type in the share allocation (Result 2  in \textsection \ref{sec:preferred}) is explained by insufficient adjustment of entrepreneurs' opening offers in SI PoorEnt scenario, and by more aggressive opening offers by investors, particularly in TI scenarios.
\end{result}

\end{document}